\documentclass[journal,onecolumn]{IEEEtran}
\usepackage{amsmath,amsfonts}
\usepackage{amsthm}
\usepackage{amssymb}
\newtheorem{theorem}{Theorem}
\newtheorem{lemma}{Lemma}
\newtheorem{proposition}{Proposition}
\newtheorem{corollary}{Corollary}
\newtheorem{definition}{Definition}
\newtheorem{remark}{Remark}
\newtheorem{example}{Example}
\usepackage{algorithmic}
\usepackage{algorithm}
\usepackage{array}

\usepackage{qcircuit}
\usepackage{tikz}
\usetikzlibrary{arrows.meta,decorations.pathreplacing,positioning}
\usepackage[caption=false,font=normalsize,labelfont=sf,textfont=sf]{subfig}
\usepackage{textcomp}

\usepackage{stfloats}
\usepackage{xcolor}
\usepackage{url}
\usepackage{verbatim}
\usepackage{graphicx}
\usepackage{cite}
\usepackage{hyperref}  
\def\BibTeX{{\rm B\kern-.05em{\sc i\kern-.025em b}\kern-.08em
    T\kern-.1667em\lower.7ex\hbox{E}\kern-.125emX}}
\usepackage{balance}

\newcommand{\qlb}{\large\textbf{[[}}
\newcommand{\qrb}{\large\textbf{]]}}
\usepackage{float}

\begin{document}
\title{Entanglement-Assisted Quantum Locally Recoverable Codes: Bounds, Optimal Constructions, and Achievability}
\author{Vijay~Kumar,~\IEEEmembership{Student Member,~IEEE} 
        and~Ramakrishna~Bandi,~\IEEEmembership{Senior Member,~IEEE}%

\thanks{The authors are with the Department of Science and Mathematics, International Institute of Information Technology, Naya Raipur, India, 493661 (e-mail: vijayk@iiitnr.edu.in, ramakrishna@iiitnr.edu.in).}}


\maketitle

\begin{abstract}
Entanglement-assisted quantum error-correcting codes (EAQECCs) remove the dual-containment
constraint of Calderbank--Shor--Steane (CSS) constructions through pre-shared entanglement,
while quantum locally recoverable codes (qLRCs) allow efficient localized recovery of
erased qudits in distributed quantum storage. This paper studies entanglement-assisted
quantum locally recoverable codes (EA-qLRCs) built via a CSS-like stabilizer construction
from pairs of classical locally recoverable codes (cLRCs), without requiring
dual-containment. We define such codes through local recovery channels, give a sufficient
stabilizer criterion for this CSS-like construction, and derive Singleton-like,
Griesmer-like, Plotkin-like, and sphere-packing-like converse bounds on the achievable
parameters of the resulting pure CSS-like EA-qLRCs, along with a Cadambe--Mazumdar-like
bound that, as in the classical case, is not explicit in closed form, and a comparative
analysis of their relative tightness across finite-length and asymptotic regimes. We give
necessary and sufficient conditions for a pure CSS-like EA-qLRC to attain the Singleton-like
bound with equality; for the single-code case $\mathcal{C}_1=\mathcal{C}_2=\mathcal{C}$,
this reduces to a simple condition on the hull dimension
$s=\dim(\mathcal{C}\cap\mathcal{C}^\perp)$, which also fixes the entanglement count via
$c=n-k-s$. We present CSS-like EA-qLRC constructions from classical LRC families - Tamo--Barg codes and cyclic codes --- and characterize precisely when these attain the Singleton-like bound, showing in particular that the cyclic families yield optimal codes while the Tamo--Barg construction, though valid, attains the bound only in the degenerate regime $k \le r$, where the locality constraint is vacuous. We complement these explicit
constructions with two Gilbert--Varshamov-like achievability bounds, established via a
classical parity-check augmentation and via a sharper concatenated-code construction, and
show both are met unconditionally for field size $q>3$ via a monomial-equivalence argument.
Finally, we unify all bounds established in this paper --- converse and achievability alike
--- under a common maximally entangled regime, giving a single comparison of the achievable
and forbidden rate--distance--locality region for CSS-like EA-qLRCs.
\end{abstract}


\section{Introduction}
Classical locally recoverable codes (LRCs) address single-node failures in distributed storage --- the most common failure mode in such systems~\cite{gopalan2012locality} --- by recovering an erased symbol from only a small subset of the other code symbols, rather than from an entire codeword as classical maximum distance separable (MDS) codes effectively require. This locality property has driven a substantial body of work on fundamental bounds~\cite{cadambe2015bounds,hao2020bounds,prakash2012optimal,tamo2016bounds} and optimal explicit constructions~\cite{tamo2014family,tamo2015cyclic,papailiopoulos2014locally,hao2020bounds,chen2017constructions} for LRCs, where optimality means attaining one of these bounds with equality.

\medskip

Quantum locally recoverable codes (qLRCs)—introduced by Golowich and Guruswami~\cite{Golowic}, who defined a qLRC of locality $r$ as a quantum code in which each qudit participates in two low-weight stabilizers—extend this idea to the quantum setting. A substantial line of follow-up work has since developed qLRCs via Calderbank--Shor--Steane (CSS) or Hermitian constructions, obtaining Singleton-like bounds~\cite{Golowic,luo2025bounds,galindo2026quantum}, constructions from good polynomials~\cite{sharma2025quantum}, codes with intersecting recovery sets~\cite{bu2025quantum}, tighter bounds and explicit families from the Hermitian construction~\cite{li2025improved,li2025optimal}, constructions from trace codes~\cite{xie2025two}, and a generalization to $(r,\rho)$-locality~\cite{galindo2026quantum}. Every one of these constructions, however, is built on the CSS or Hermitian stabilizer mechanism, which requires the underlying classical code (or pair of codes) to satisfy a dual-containment or self-orthogonality condition. This restriction is severe specifically for locally recoverable codes: the best-known optimal classical LRC constructions are almost never dual-containing, so most of classical LRC theory remains inaccessible to CSS-qLRC constructions.

\medskip

The entanglement-assisted quantum error-correcting code (EAQECC) framework~\cite{brun2007general} was developed precisely to remove this restriction in the general, non-local setting: pre-shared entanglement between encoder and decoder allows quantum stabilizer codes to be built from \emph{arbitrary} pairs of classical linear codes, with no dual-containment requirement at all. Whether this same mechanism could be extended to \emph{local} recovery, however, remained open --- a question Luo \textit{et al.}~\cite{luo2025bounds} posed explicitly: ``Do entanglement-assisted quantum codes admit the property of being locally recoverable?" This paper answers that question by developing a systematic theory of \emph{entanglement-assisted quantum locally recoverable codes (EA-qLRCs)}.

\medskip

\emph{Our conceptual contribution.} We show that the answer is yes because locality and dual-containment were never fundamentally linked in the first place. In the stabilizer formalism, the existence of two \emph{local} dual codewords, one from each classical constituent code, whose combined support covers a given coordinate with weight at most $r+1$, suffices to guarantee local recoverability there (Theorem~\ref{Stb_EAqLRC}); this is a purely local, support-based sufficient condition. Dual-containment, by contrast, is a \emph{global} condition on the code pair, needed only to make the resulting stabilizer generators commute in the absence of shared entanglement. Because entanglement assistance is exactly the resource that removes this global commutativity requirement, it can be layered onto the local-recovery mechanism without disturbing it: the same local support condition that gives locality continues to hold verbatim, whether or not the two classical codes are dual-containing. Locality and dual-containment are, in this sense, orthogonal requirements on a CSS-type code pair, and entanglement assistance is precisely the resource that decouples them. The price of this decoupling is quantified by a single classical invariant: for the important case of a single constituent code, its \emph{hull dimension} $\dim(\mathcal{C}\cap\mathcal{C}^\perp)$ simultaneously determines the entanglement consumed and, via a sharp threshold, whether the resulting code is optimal --- linear complementary dual (LCD) code. This decoupling is what allows the constructions in this paper to draw on classical LRC families, such as Tamo--Barg and LCD codes, that are almost never dual-containing and were therefore previously inaccessible to CSS-qLRC constructions.

\subsection{Our Contributions}
The technical contributions below develop this idea from its stabilizer foundations through to explicit optimal constructions.

\begin{enumerate}

\item \textbf{Stabilizer Foundations and CSS-Like Construction:} We begin by defining an EA-qLRC through its recovery channels, establishing a sufficient stabilizer criterion (Theorem~\ref{Stb_EAqLRC}) for this locality property, and using it to formulate a CSS-like construction of EA-qLRCs from a pair of classical linear codes (Proposition~\ref{CSS-EA-qLRC}).

\item \textbf{Explicit Converse Bounds:} Building on this construction, we establish four
explicit converse bounds for CSS-like EA-qLRCs, characterizing the fundamental trade-offs among
blocklength, dimension, minimum distance, locality, and entanglement consumption. These
bounds generalize classical LRC bounds and CSS-based qLRC bounds to the
entanglement-assisted setting:
    \begin{enumerate}
    \item \textbf{Singleton-like Bound:}
    \[
        2\delta \le n - \kappa + c - 2\left\lceil \frac{\kappa}{r} \right\rceil + 4.
    \]
    \item \textbf{Griesmer-like Bound:}
    \[
        n \;\ge\; 2\max_{0\le\tau\le \lceil\frac{\kappa}{r}\rceil-1} \left\{ \tau(r+1) + \sum_{i=0}^{\kappa-\tau r-1}\!\left\lceil\frac{\delta}{q^i}\right\rceil \right\} - \kappa - c.
    \]
    \item \textbf{Plotkin-like Bound:}
    \[
        2\delta \;\le\; \min_{0\le \tau \le \lceil \frac{\kappa}{r}\rceil -1} \left\{\frac{q^{\kappa -\tau r-1}(q-1)\bigl(n+\kappa +c-2\tau(r+1)\bigr)}{q^{\kappa -\tau r}-1}\right\}.
    \]
    \item \textbf{Sphere-Packing-like Bound:}
    \[
        \kappa \;\le\; n + c - 2\max_{0\le\tau\le \lfloor\frac{n+\kappa+c-2}{2(r+1)}\rfloor} \left\{ \tau + \log_q \!\left[ V_q\left(\frac{n+\kappa+c}{2}-\tau(r+1),\, t\right) \right] \right\},
    \]
    where $t = \lfloor \frac{\delta-1}{2} \rfloor$ and $V_q(M, t) = \sum_{j=0}^{t} \binom{M}{j}(q-1)^j$ is the volume of a classical $q$-ary Hamming ball of radius $t$ in $\mathbb{F}_q^M$.
    \end{enumerate}
We additionally establish a Cadambe--Mazumdar-like bound (Corollary~\ref{cor:CM}) that, as
in the classical case, is not explicit in closed form, but yields tighter,
alphabet-dependent constraints than the Singleton-like bound when $k_{\mathrm{opt}}^{(q)}$
is known or bounded.

\item \textbf{Optimality Conditions for Pure CSS-Like EA-qLRCs:} We answer the question of when these bounds can be met by deriving necessary and sufficient conditions for a pure CSS-like EA-qLRC to attain the Singleton-like bound with equality (Theorem~\ref{pure_EA-qLRC}). Specifically, a pure CSS-like EA-qLRC with locality $r$ attains equality in \eqref{EA-qLRC_SB} if and only if:
    \begin{enumerate}
    \item $\mathcal{C}_1$ and $\mathcal{C}_2$ are classical codes with identical parameters $[n,k,d]_q$.
    \item $\mathcal{C}_1$ and $\mathcal{C}_2$ attain the classical LRC Singleton bound~\eqref{cLRC Singleton Bound}, with $\left\lceil \frac{k}{r} \right\rceil = \left\lceil \frac{2k-n+c}{r} \right\rceil$.
    \end{enumerate}

\item \textbf{Explicit Constructions and Their Optimality Boundary:} Instantiating this
optimality criterion, we construct two families of CSS-like EA-qLRCs in
Section~\ref{Construction Optimal EA-qLRCs} from Tamo--Barg codes and general cyclic code
pairs. We identify an exact algebraic obstruction showing that the Tamo--Barg-based
construction can never attain the Singleton-like bound in any nontrivial regime, since
optimality would require $k\le r$, a regime in which the locality constraint is vacuous; the general cyclic-pair construction is likewise shown, by explicit example, not to attain optimality in general. This motivates the specialized, provably optimal
construction of the next contribution.

\item \textbf{Maximally Entangled Cyclic CSS-like EA-qLRCs:} We characterize the LCD condition for any cyclic code (Lemma~\ref{cyclic_LCD}) and leverage it to construct several explicit families of maximally entangled CSS-like EA-qLRCs of lengths $q-1$ and $q+1$ that attain the Singleton-like bound with equality.

\item \textbf{Gilbert--Varshamov-Like Achievability Bounds:} We establish two Gilbert--Varshamov-like asymptotic achievability bounds for CSS-like EA-qLRCs derived from classical LRC achievability results:
    \begin{enumerate}
    \item \textbf{Basic Parity-Check Augmentation Bound:}
    \[
        R \;\ge\; \frac{r}{r+1} - H_q(\Delta), \qquad \gamma = 1 - R, \qquad R_{\mathrm{net}} = 2R - 1.
    \]
    \item \textbf{Sharper Concatenated-Code Bound:}
    \[
        R \;\ge\; 1 - H_q(\Delta) - \frac{1}{r+1}\log_q\!\left[1+(q-1)\left(1-\Delta\cdot\frac{q}{q-1}\right)^{r+1}\right], \qquad \gamma = 1 - R.
    \]
    \end{enumerate}
    We prove that both bounds are achievable unconditionally for all field sizes $q > 3$ via a locality-preserving monomial-equivalence transformation, identifying $q=2,3$ as the remaining open regime.

\item \textbf{Unified Parameter Landscape:} We unify all converse bounds (Singleton-, Griesmer-, Plotkin-, and sphere-packing-like) and both achievability bounds under the shared maximally entangled regime $\gamma = 1 - R$. We show that each bound collapses to a closed form in $\Delta$ alone, providing a single visual representation of the achievable and forbidden parameter regions for CSS-like EA-qLRCs.

\end{enumerate}

\textit{Relation to concurrent work.} Independently and concurrently with this work,
Li \emph{et al.}~\cite{li2026eaqlrc} also introduced entanglement-assisted quantum
locally recoverable codes. They likewise defined locality via recovery channels and proposed
a CSS-like construction from classical code pairs without requiring dual-containment,
together with a Singleton-like bound and a necessary and sufficient optimality
characterization for pure codes analogous to our Theorem~\ref{pure_EA-qLRC}. Their work additionally
establishes an upper bound on achievable locality,
$r \le \min\{n-1, \dim(C_X + C_Z)\}$,
which we do not consider here, and gives explicit constructions from $\ell$-intersection pairs of MDS codes and block parity-check matrices built from generalized Reed--Solomon codes. In contrast, the present paper establishes three additional converse bounds (Griesmer-like, Plotkin-like, and sphere-packing-like), together with a full finite-length and asymptotic comparison among all four bounds, and provides disjoint explicit constructions from Tamo--Barg codes, general cyclic LRC pairs, and LCD cyclic codes of length dividing $q \pm 1$, along with a Gilbert--Varshamov-like existence result. The two works thus develop complementary aspects of the same emerging framework.

The remainder of the paper is organized as follows. Section~\ref{Preliminary} presents the necessary preliminaries along with several existing results. Section~\ref{EA-qLRC} introduces new results on the locality of entanglement-assisted quantum codes. Section~\ref{Bounds EA-qLRC} establishes four explicit bounds for CSS-like EA-qLRCs via the classical LRC bounds and analyzes the tightness of these bounds in finite length and asymptotically. Section~\ref{Conditions EA-qLRC} characterizes the necessary and sufficient conditions for pure CSS-like EA-qLRCs to achieve optimality with respect to \eqref{EA-qLRC_SB}. Section~\ref{Construction Optimal EA-qLRCs} constructs CSS-like EA-qLRCs with explicit parameters using cyclic codes and Tamo–Barg codes. Section~\ref{sec:maximally-entangled-constructions} introduces maximally CSS like  EA-qLRCs based on classical cyclic LRCs and their complementary dual codes. Section~\ref{EA-qLRC GV} derives a basic and a sharper Gilbert--Varshamov-like achievability
bound for CSS-like EA-qLRCs from corresponding classical achievability results for LRCs, establishes their unconditional validity for $q>3$ via monomial equivalence to LCD codes, and closes with a numerical analysis of the sharper bound and a unified comparison of every bound established in this paper under the maximally entangled regime. Finally, Section~\ref{conclusion} concludes the paper and outlines several open problems for future research.

\section{Preliminaries}\label{Preliminary}
This section introduces the foundational concepts required for the remainder of the paper. We review classical linear codes and their duals, classical LRCs together with their associated distance bounds, and the framework of entanglement-assisted quantum error-correcting codes. Throughout this paper, we use the following notation: $q$ is a prime power, $\mathbb F_q$ is denoted by the finite field of order $q$; for $m,n\in\mathbb{N}$, we denote $[m]^\dagger:=\{0,1,\ldots,m\}$ and $[n]=\{1,\dots,n\}$. For $\mathbf u=(u_1,\dots,u_n)\in\mathbb F_q^n$, $\operatorname{supp}(\mathbf u)=\{i\in[n]: u_i\neq 0\}$, and $\mathrm{wt}(\mathbf u)=\vert\operatorname{supp}(\mathbf u)\vert$. Let  $S\subseteq\mathbb F_q^n$, $\mathrm{wt}(S)=\min_{\mathbf u\in S,\,\mathbf u\neq \mathbf 0}\mathrm{wt}(\mathbf u)$ and for a linear code $\mathcal C$, $\mathrm{wt}(\mathcal C)$ denote the minimum weight of a nonzero codeword. $\operatorname{tr}(\cdot)$ denotes the trace map from the relevant extension field to $\mathbb F_p$, and $\omega=e^{2\pi i/p}$, where $p$ is the characteristic of $\mathbb F_q$.

\subsection{Classical Linear Codes} 
A linear code $\mathcal{C}$ over $\mathbb{F}_q$ of length $n$ and dimension $k$ is a $k$-dimensional subspace of $\mathbb{F}_q^n$, and is denoted by an $[n,k]_q$ code. The minimum distance of $\mathcal{C}$ is defined by $d=\operatorname{wt}(\mathcal{C})$ then; we denote $\mathcal{C}$ by an $[n,k,d]_q$ . A code $\mathcal{C}$ is generated by a generator matrix $G\in\mathbb{F}_q^{k\times n}$, whose rows form a basis of $\mathcal{C}$. Thus, $\mathcal{C} = \{uG : u\in\mathbb{F}_q^k\}$. The Euclidean dual of $\mathcal{C}$ is defined by
\[
\mathcal{C}^{\perp} = \left\{ (x_1,\ldots,x_n)\in\mathbb{F}_q^n : \sum_{i=1}^{n}x_ic_i=0 \text{ for all } (c_1,\ldots,c_n)\in\mathcal{C} \right\}.
\]
Similarly, when $\mathcal{C}$ is a code over $\mathbb{F}_{q^2}$, its Hermitian dual is defined by
\[
\mathcal{C}^{\perp_H}= \left\{ (x_1,\ldots,x_n)\in\mathbb{F}_{q^2}^n : \sum_{i=1}^{n}x_ic_i^{\,q}=0 \text{ for all } (c_1,\ldots,c_n)\in\mathcal{C} \right\}.
\]

If $\mathcal{C}$ is an $[n,k,d]_q$ code, then $\mathcal{C}^{\perp}$ is an $[n,n-k,d^{\perp}]_q$ code, where $d^{\perp}$ denotes the minimum distance of $\mathcal{C}^{\perp}$. A parity-check matrix $H\in\mathbb{F}_q^{(n-k)\times n}$ for $\mathcal{C}$ satisfies $\mathcal{C} = \ker(H) = \left\{ x\in\mathbb{F}_q^n : Hx^{\top}=\mathbf{0} \right\}$,
where the rows of $H$ form a basis of $\mathcal{C}^{\perp}$. The generator and parity-check matrices satisfy the orthogonality relation $GH^{\top}=0$.
A code is called \emph{systematic} if its generator matrix contains the identity matrix $I_k$ as a submatrix. In this case, $G=[I_k\mid A],\; H=[-A^{\top}\mid I_{n-k}]$,
where the first $k$ coordinates correspond to information symbols and the remaining $n-k$ coordinates correspond to parity symbols.

A code $\mathcal{C}$ is called \emph{self-orthogonal} (or weakly self-dual) if $\mathcal{C} \subseteq \mathcal{C}^\perp$, which requires $k \leq n/2$; \emph{self-dual} if $\mathcal{C} = \mathcal{C}^\perp$, which requires $k = n/2$; \emph{dual-containing} if $\mathcal{C}^\perp \subseteq \mathcal{C}$; and a \emph{linear complementary dual} (LCD) code if $\mathcal{C} \cap \mathcal{C}^\perp =\{0\}$.

The Singleton bound $d\le n-k+1$ provides a fundamental upper bound on the minimum distance of an $[n,k,d]_q$ code. A code attaining equality in the Singleton bound is called a \emph{maximum distance separable} (MDS) code.
\subsection{Classical Locally Recoverable Codes}\label{ssec:classical-LRC}

\begin{definition}\label{def:locality}
A code $\mathcal{C}$ of length $n$ over $\mathbb{F}_q$ is said to have locality $r$ if for every coordinate $i \in [n]$, there exists a set $\Gamma_i \subseteq [n] \setminus \{i\}$ with $|\Gamma_i| \leq r$ such that the symbol $c_i$ can be recovered from $\{ c_j \mid j \in \Gamma_i \}$. The set $\Gamma_i$ is called a recovery set for coordinate $i$.
\end{definition}

An $[n, k, d]_q$ code with locality $r$ is denoted an $(n, k, d, r)_q$-LRC. It is well known [\cite{luo2018optimal},\cite{tan2023minimum}] that for a linear code, the locality condition is equivalent to requiring that for each coordinate $i$, there exists a local parity-check of support size at most $r + 1$ involving coordinate $i$. Throughout this paper, we adopt the standard convention that $1 \le r \le k$ \cite{gopalan2012locality}: if $r \ge k$, then every coordinate of an $[n,k,d]_q$ code with $d \ge 2$ can be recovered from an information set of size $k$ that avoids it, so the locality constraint is vacuous.

The first fundamental distance bound for codes with locality was established by Gopalan \textit{et al.} \cite{gopalan2012locality}, who showed that the presence of locality constraints strictly weakens the classical Singleton bound.
\begin{lemma}[\cite{gopalan2012locality}]
Let $\mathcal{C}$ be an $(n, k, d, r)_q$-LRC. Then
\begin{equation}\label{cLRC Singleton Bound}
    d \le n - k - \left\lceil \frac{k}{r} \right\rceil + 2.
\end{equation}
\end{lemma}

 A code achieving this bound with equality is called an {\it{optimal LRC}}. The Singleton-like bound is tight only when the underlying alphabet is sufficiently large. For codes over fixed or small alphabets, a tighter information-theoretic upper bound was derived by Cadambe et al. \cite{cadambe2015bounds}.
 
 \begin{lemma}[CM-Bound  \cite{cadambe2015bounds}]
     Let $\mathcal{C}$ be an $(n, k, d, r)_q$-LRC. Then
     \begin{equation}\label{CM-Bound}
    k \le \min_{ \tau \in\mathbb{Z}_{+}}\left\{ \tau r + k_{\mathrm{opt}}^{(q)}(n-\tau(r+1),d) \right\},
\end{equation}
 \end{lemma}

where $\tau\;\le\; min\left\{\lceil\frac{n}{r+1}\rceil,\lceil\frac{k}{r}\rceil\right\}$ and $k_{\mathrm{opt}}^{(q)}(n', d)$ denotes the maximum dimension of a $q$-ary code of length $n'$ and minimum distance $d$. 

Finding $k_q^{\text{opt}}(n',d)$ is not always straightforward, so the CM bound \eqref{CM-Bound} is not explicit. To get more explicit bounds for $(n, k, d, r)_q$-LRC. Hao \textit{et al.}  \cite{hao2020bounds}  used the Griesmer and Plotkin bounds for linear codes in \eqref{CM-Bound} and found two new explicit bounds: the Griesmer-like bound and the Plotkin-like bound of $(n, k, d, r)_q$-LRC. In \cite[Lemma 1]{li2025improved}, Li \textit{et al.} give a more general form of the sphere packing bound for $(n, k, d, r)_q$-LRC. All these bounds are summarized as follows:

\begin{lemma}[\cite{hao2020bounds,li2025improved}]\label{lemma:classical-bounds}
Let $\mathcal{C}$ be an $(n, k, d, r)_q$-LRC and  $\tau$ be an integer. Then the following bounds hold.
\begin{enumerate}
    \item \textbf{Griesmer-like bound}: We have that
    \begin{equation}
    n \ge \max_{0 \le \tau \le \lceil \frac{k}{r} \rceil - 1} \left\{ \tau(r+1) + \sum_{i=0}^{k-\tau r-1} \left\lceil \frac{d}{q^i} \right\rceil \right\}. \label{eq:griesmer}
    \end{equation}

    \item \textbf{Plotkin-like bound}: We have that
    \begin{equation}
    d \le \min_{0 \le \tau \le \lceil \frac{k}{r} \rceil - 1} \left\{ \frac{q^{k-\tau r-1}(q-1)(n-\tau(r+1))}{q^{k-\tau r} - 1} \right\}. \label{eq:plotkin}
    \end{equation}

    \item \textbf{Sphere-packing-like bound}: We have that
    \begin{equation}\label{eq:sphere-packing}
    k \le n - \max_{0 \le \tau \le \lfloor \frac{n-1}{r+1} \rfloor} \left\{ \tau + \log_q \left( \sum_{i=0}^{\lfloor \frac{d-1}{2} \rfloor} \binom{n-\tau(r+1)}{i} (q-1)^i \right) \right\}. 
    \end{equation}
\end{enumerate}
\end{lemma}

\subsection{Entanglement-Assisted Quantum Codes}
Let $\mathbb{C}$ denote the field of complex numbers and $\mathcal{H}=\mathbb{C}^q$ a $q$-dimensional Hilbert space over $\mathbb{C}$. An $n$-qudit quantum system is represented by a nonzero vector in the tensor product space $\mathcal{H}=(\mathbb{C}^q)^{\otimes n} \cong \mathbb{C}^{q^n}$. For $x=(x_1,\ldots,x_n) \in \mathbb{F}_q^n$, the standard basis of $(\mathbb{C}^q)^{\otimes n}$ over $\mathbb{C}$ is given by $\{\, |x\rangle := |x_1\rangle \otimes \cdots \otimes |x_n\rangle : x \in \mathbb{F}_q^n \,\}$. Any quantum state in $(\mathbb{C}^q)^{\otimes n}$ can be expressed as $|\psi\rangle = \sum_{x \in \mathbb{F}_q^n} u_x |x\rangle$, where the coefficients $u_x \in \mathbb{C}$ satisfy $\sum_{x \in \mathbb{F}_q^n} |u_x|^2 = 1$. Let $u=(u_1,\ldots,u_n)$ and $v=(v_1,\ldots,v_n)$ be elements of $\mathbb{F}_q^n$.  The operators $X(u)$ and $Z(v)$ act on a basis state $|x\rangle$ as
\[
X(u)|x\rangle = |x+u\rangle, \qquad
Z(v)|x\rangle = \omega^{\,\mathrm{tr}(\langle v,x\rangle)} |x\rangle,
\]
 The quantum error operators acting on $(\mathbb{C}^q)^{\otimes n}$ form the generalized Pauli group $\mathcal{G}_n = \{\, \omega^{\,j} X(u)Z(v) : u,v \in \mathbb{F}_q^n,\; j \in \mathbb{F}_p \,\}$. For any operator $E=\omega^{\,j}X(u)Z(v) \in \mathcal{G}_n$, the quantum support is defined as $\mathrm{supp}(E) = \{\, i \in [n] : u_i \neq 0 \text{ or } v_i \neq 0 \,\}$.

We briefly recall the notation for density operators and quantum channels. Let $\mathcal{H}_{\Omega}=(\mathbb{C}^{q})^{\otimes |\Omega|}$ denote the Hilbert space associated with a finite set $\Omega$ of qudits, and  ${D}(\Omega)$ denote the set of density operators acting on $\mathcal{H}_{\Omega}$. A density operator $\rho\in{D}(\Omega)$ is a positive semidefinite Hermitian operator on $\mathcal{H}_{\Omega}$ satisfying $\operatorname{Tr}(\rho)=1$. For a pure state $|\psi\rangle\in\mathcal{H}_{\Omega}$, the corresponding density operator is $\rho=|\psi\rangle\langle\psi|$. If $A\subseteq\Omega$, the reduced density operator on the subsystem $A$ is obtained by taking the partial trace over its complement, $\rho_{A}=\operatorname{Tr}_{\Omega\setminus A}(\rho)\in{D}(A).$

A quantum channel from the subsystem $\Omega$ to another subsystem $\Lambda$ is a completely positive trace-preserving (CPTP) map $\mathcal{N}:{D}(\Omega)\longrightarrow{D}(\Lambda)$. Equivalently, every quantum channel admits a Kraus representation $\mathcal{N}(\rho)=\sum_k K_k \rho K_k^\dagger,$ where the Kraus operators $K_{k}:\mathcal{H}_{\Omega}\longrightarrow\mathcal{H}_{\Lambda}$
satisfy the completeness relation $\sum_k K_k^\dagger K_k=I.$

A quantum code with subspace $\mathcal Q\in\mathcal{H}_{\Omega}$ is correctable against a noise channel $\mathcal{N}$ if there exists a CPTP recovery map $\mathcal{R}$ such that $\mathcal{R} \circ \mathcal{N}$ acts as the identity on the code subspace, i.e. $ \mathcal{R}\circ \mathcal{N}(\rho )=\rho \; \text{for every } \rho \in D(\Omega)$. The necessary and sufficient conditions for the existence of such a recovery map $\mathcal{R}$ are given by the Knill-Laflamme conditions \cite{knill2000theory} adapted to the entanglement-assisted framework. Let $\mathcal{H}_{\mathcal{E}}$ denote the Hilbert space of the encoder's $n$ physical qudits and $\mathcal{H}_{\mathcal{D}}$ denote the Hilbert space of the decoder's halves of the $c$ shared ebits. The errors $\mathcal E \subset \mathcal{G}_n$ act locally on the encoder's system, while the decoder's system remains noiseless. The entanglement-assisted quantum code subspace $\mathcal Q_E \subset \mathcal{H}_{\mathcal{E}} \otimes \mathcal{H}_{\mathcal{D}}$ with projector $P$ can correct the set of errors $E$ if and only if
\[P (E_{a} \otimes I_{\mathcal{D}})^{\dagger }(E_{b} \otimes I_{\mathcal{D}}) P = \alpha_{ab} P, \qquad \forall\; E_{a}, E_{b} \in  E,
\]
where $I_{\mathcal{D}}$ is the identity operator on the decoder's qudits, and $\alpha=(\alpha_{ab})$ is a complex Hermitian matrix, i.e. $\alpha_{ab} = \alpha_{ba}^*$.

A general (non-entanglement-assisted) quantum error-correcting code $\mathcal{Q}$ is a $q^\kappa$-dimensional subspace of the $q^n$-dimensional complex Hilbert space $\mathbb{C}^{q^n}$. It encodes $\kappa$ logical qudits into $n$ physical qudits and is denoted by an {\qlb}$n,\kappa${\qrb}$_q$ code. The encoding operation follows two steps. First, one appends $n-\kappa$ ancilla qudits, typically $|0\rangle^{\otimes (n-\kappa)}$, to the state $|\psi\rangle$ of $\kappa$ information qudits to be sent or stored. Then an encoding unitary operator $U_{\mathrm{enc}}$ is applied:
\[
|\psi\rangle \mapsto |\psi\rangle \otimes |0\rangle^{\otimes (n-\kappa)} \mapsto |\Psi_L\rangle= U_{\mathrm{enc}}\big( |\psi\rangle \otimes |0\rangle^{\otimes (n-\kappa)}\big),
\]
where $|\Psi_L\rangle$ is called the encoded state, or logical state.
If $\mathcal{Q}$ is a stabilizer code, its code subspace is defined as the common $+1$-eigenspace of the stabilizer group $\mathcal{S}$. Let $\mathcal{S} = \langle S_1, S_2, \ldots, S_{n-\kappa} \rangle \subset \mathcal{G}_n$ be an abelian subgroup generated by $n-\kappa$ independent and commuting Pauli operators. The corresponding stabilizer code space is given by
\[
\mathcal{Q}_{\mathcal{S}} = \left\{ |\psi\rangle \in \mathbb{C}^{q^n} \;\middle|\; S_i|\psi\rangle = |\psi\rangle, \, \forall i \in \{1, 2, \ldots, n-\kappa\} \right\}.
\]

Entanglement-assisted quantum error-correcting codes (EAQECCs) extend this framework by allowing the use of pre-shared entanglement between the encoder ($\mathcal{E}$) and the decoder ($\mathcal{D}$). A quantum code $\mathcal{Q}_E$ is called an {\qlb}$n,\kappa,\delta;c${\qrb}$_q$  entanglement-assisted quantum error-correcting code (EAQECC) if it encodes $\kappa$ logical qudits into $n$ physical qudits with the assistance of $c$ copies of a maximally entangled state, and has minimum distance $\delta$, meaning it can correct up to $\left\lfloor(\delta-1)/2\right\rfloor$ arbitrary quantum errors. The logical state is given by
\[
|\psi\rangle \otimes |0\rangle^{\otimes (n-\kappa-c)} \otimes |\Phi\rangle^{\otimes c}_{\mathcal{ED}}
\;\longmapsto\;
|\Psi_L\rangle= (U_{\mathrm{enc}} \otimes I_\mathcal{D})\Big(|\psi\rangle \otimes |0\rangle^{\otimes (n-\kappa-c)} \otimes |\Phi\rangle^{\otimes c}_{\mathcal{ED}}\Big).
\]
The states $|\Phi\rangle_{\mathcal{ED}}$ represent maximally entangled Einstein–Podolsky–Rosen (EPR) \cite{einstein1935can} pairs shared between the encoder and decoder, where $I_{\mathcal{D}}$ denotes the identity operator acting on the decoder's halves of the shared entangled pairs. We assume that noise affects only the first $n$ qudits transmitted or stored through the noisy channel $\mathcal{N}$, while the decoder's half of the shared $c$ maximally entangled pairs remains unaffected. When an error or erasure occurs in the transmitted or stored $n$ qudits, the decoder performs a syndrome measurement on them together with his half of the $c$ maximally entangled qudits to correct the error, as illustrated in Figure~\ref{fig:1}.

In the stabilizer formalism for EAQECCs, every element of the generalized Pauli group $\mathcal{G}_{n+c}$ can be written as $E=E_{\mathcal E}\otimes E_{\mathcal D}$, where $E_{\mathcal E}\in\mathcal{G}_n$ acts on the encoder's $n$ physical qudits and $E_{\mathcal D}\in\mathcal{G}_c$ acts on the decoder's $c$ halves of the shared maximally entangled pairs. To construct the entanglement-assisted stabilizer code, we begin with a set of independent generators and partition them into symplectic and isotropic subsets. Let $X_i^{\mathcal D}$ and $Z_i^{\mathcal D}$ denote the generalized Pauli operators acting on the $i$-th decoder qudit. The corresponding extended stabilizer generators are defined by
$g'_i=g_i^{\mathcal E}\otimes Z_i^{\mathcal D}, \;\text{for}\;1\le i\le c;\;
h'_i=h_i^{\mathcal E}\otimes X_i^{\mathcal D},\;\text{for}\; 1\le i\le c;\;
g'_j =g_j^{\mathcal E}\otimes I_{\mathcal D},\; \text{for}\;c+1\le j\le n-\kappa,$
where $g_i^{\mathcal{E}}, h_i^{\mathcal{E}}, g_j^{\mathcal{E}} \in \mathcal{G}_n$ and $Z^{\mathcal{D}}_{i},X^{\mathcal{D}}_{i}\in\mathcal{G}_c$. These operators satisfy the following commutation relations for all distinct indices $i \neq j$:
$[g_i^{\mathcal E},g_j^{\mathcal E}]=[h_i^{\mathcal E},h_j^{\mathcal E}]=[g_i^{\mathcal E},h_j^{\mathcal E}]=0$.
However, the pairs $\{g_i^{\mathcal{E}}, h_i^{\mathcal{E}}\}$ form symplectic partners that do not commute. Using these extended generators, we define the symplectic subgroup $\mathcal{S}_S$ and the isotropic subgroup $\mathcal{S}_I$ as follows $\mathcal{S}_S = \langle g'_1, \dots, g'_c, h'_1, \dots, h'_c \rangle,\;\mathcal{S}_I = \langle g'_{c+1}, \dots, g'_{n-\kappa} \rangle.$ The stabilizer group $\mathcal{S}_E \subset \mathcal{G}_{n+c}$ for the entanglement-assisted code is
\begin{equation}\label{eq:EAQECC-stabilizer}
    \mathcal{S}_E = \mathcal{S}_S \times \mathcal{S}_I=\langle g'_1, \dots, g'_{n-\kappa}, h'_1, \dots, h'_c \rangle.
\end{equation}
Finally, the code subspace $\mathcal{Q}_{\mathcal{S}_E}$ is defined as the simultaneous $+1$-eigenspace of this stabilizer group:
\[
\mathcal{Q}_{\mathcal{S}_E} = \left\{ |\psi\rangle \in(\mathbb{C}^q)^{\otimes(n+c)} \;\big\vert\; S |\psi\rangle = |\psi\rangle, \; \forall S \in \mathcal{S}_E \right\}.
\]

\begin{figure}[ht!]
    \centering
    \includegraphics[width=0.8\linewidth]{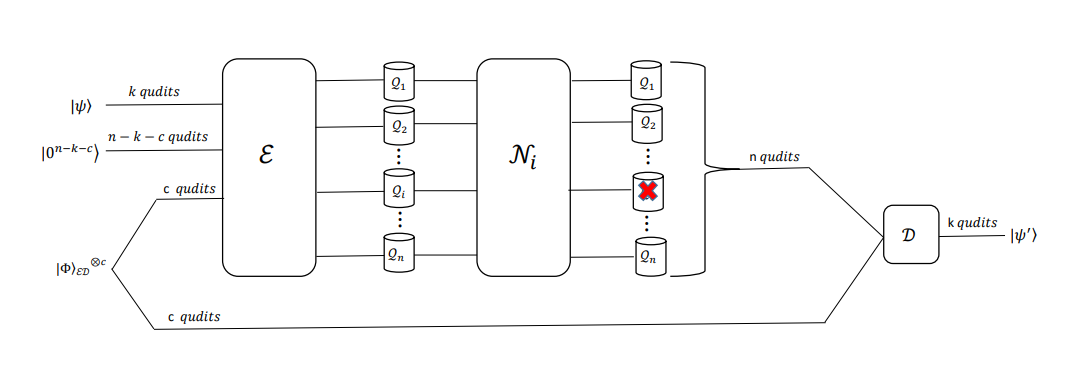}
    \caption{An entanglement-assisted quantum distributed storage system. A quantum state $|\psi\rangle$ is encoded into $n$ quantum storage systems $Q_1,Q_2,\ldots ,Q_n$ by an encoder with entanglement assistance. If any one of these $n$ systems is erased by a noisy channel, a decoder $\mathcal{D}$ must recover the quantum information from the remaining (unerased) storage systems and the entanglement assistance. In particular, for EA-qLRCs, any one of these $n$ systems erased by a channel can be recovered by accessing a small subset of $r$ surviving nodes together with the decoder's $\mathcal{D}$ share of the entangled resource.}
    \label{fig:1}
\end{figure}

The fundamental performance bound for entanglement-assisted quantum codes, analogous to the classical Singleton bound, takes the following form.
\begin{lemma}\cite[Theorem 6]{lai2017linear} \label{lem:EA-singleton}
Let $\mathcal{Q}_E$ be an {\qlb}$n,\kappa,\delta;c${\qrb}$_q$ entanglement-assisted quantum code. If $\delta \le \frac{n+2}{2}$, then
\[
2\delta \leq n-\kappa+c+2.
\]
\end{lemma}

A code achieving this bound with equality is called an EA-MDS code. Note that when $c = 0$, this reduces to the standard quantum Singleton bound $2\delta \leq n-\kappa+2$ for quantum codes.

A general framework for constructing EAQECCs from stabilizer generators that do not necessarily commute was developed for the qubit case by Brun \textit{et al.} \cite{brun2014catalytic}, establishing a correspondence between arbitrary classical codes over $\mathbb{F}_4$ (equivalently, pairs of binary codes over $\mathbb{F}_2$) and entanglement-assisted quantum codes. This approach has been further generalized to higher-dimensional systems (qudits) with $q>2$ in Galindo \textit{et al.} \cite{galindo2019entanglement}. We have the following construction.

\begin{theorem}(\cite{galindo2019entanglement}, Theorem~4)\label{PC-EAQEC}
Let $\mathcal{C}_i$ be an $[n,k_i,d_i]_q$ linear code over $\mathbb{F}_q$ with parity-check matrix $H_i$, for $i=1,2$. Then there exists an EAQEC code with parameters {\qlb}$n,\, k_1+k_2-n+c,\,\delta;\, c${\qrb}$_q$, where $\delta=\min\{\mathrm{wt}(\mathcal{C}_1\setminus(\mathcal{C}_1\cap\mathcal{C}_2^\perp)),\,\mathrm{wt}(\mathcal{C}_2\setminus(\mathcal{C}_1^\perp\cap\mathcal{C}_2))\}$ and $c=\operatorname{rank}(H_1H_2^{T})=n-k_1-\dim_{\mathbb{F}_q}(\mathcal{C}_1^\perp\cap\mathcal{C}_2)$ is the minimum number of required maximally entangled pairs.
\end{theorem}

In particular, if $\mathrm{wt}(\mathcal{C}_1\setminus(\mathcal{C}_1\cap\mathcal{C}_2^\perp))=\mathrm{wt}(\mathcal{C}_1)=d_1$ and, $\mathrm{wt}(\mathcal{C}_2\setminus(\mathcal{C}_1^\perp\cap\mathcal{C}_2))=\mathrm{wt}(\mathcal{C}_2)=d_2$ , then the resulting EAQECC is said to be pure (or nondegenerate) if $\delta=\min\{d_1,d_2\}$ and impure (or degenerate) otherwise, i.e. , if $\delta >\min\{d_1,d_2\}$.

\section {Entanglement-Assisted Quantum Locally Recoverable Codes}\label{EA-qLRC}
The notion of locality has been well established for both classical codes and \emph{CSS} quantum codes. We now generalize this concept to the entanglement-assisted framework by introducing EA-qLRCs and examining the role of pre-shared entanglement in enabling local recovery.
\begin{definition}\label{EAqLRC}
An entanglement-assisted quantum code $\mathcal{Q}_E$
is called an EA-qLRC with \emph{locality} $r$ if the following conditions are satisfied.
\begin{enumerate}
\item[(i)] For every coordinate $i\in[n]$, there exists a recovery set $\Gamma_i\subseteq[n]$ with $i\in\Gamma_i$ and 
$|\Gamma_i|\le r+1,$
and a CPTP recovery channel
\[\mathcal R_i:
\mathsf D\!\left(
\mathcal H_{\Gamma_i\setminus\{i\}}
\otimes
\mathcal H_{\mathcal D}
\right)
\longrightarrow
\mathsf D\!\left(
\mathcal H_{\Gamma_i}
\right),\]
where $\mathcal H_{\mathcal D}$ denotes the decoder's share of the pre-shared maximally entangled state.
\item[(ii)] For every code state $\rho\in\mathcal Q_E$,
the erased qudit at coordinate $i$ can be recovered by applying $\mathcal R_i$ 
\[
\left(
\mathcal{R}_i
\otimes
I_{[n]\setminus\Gamma_i}
\right)
\left(
\rho_{([n]\setminus\{i\})}
\right)
=
\rho.
\]
\end{enumerate}
\end{definition}

The stabilizer construction in \eqref{eq:EAQECC-stabilizer} naturally extends to EA-qLRCs. We next derive a sufficient stabilizer criterion for EA-qLRCs by incorporating the locality constraints into the entanglement-assisted stabilizer formalism.

\begin{theorem}\label{Stb_EAqLRC}
Let $\mathcal{Q}_{\mathcal{S}_E}\subseteq(\mathbb{C}^q)^{\otimes(n+c)}$ be an
entanglement-assisted quantum stabilizer code with stabilizer group $\mathcal{S}_E$.
Suppose that for every encoder coordinate $i\in[n]$, there exist two stabilizer generators
\[
E_1 = E_{\mathcal{E}} \otimes E_{\mathcal{D}}, \qquad E_2 = E'_{\mathcal{E}} \otimes E'_{\mathcal{D}} \;\in\; \mathcal{S}_E,
\]
where $E_{\mathcal{E}} = X^\alpha Z^\beta$ and $E'_{\mathcal{E}} = X^{\alpha'} Z^{\beta'}$
satisfy $(\alpha_i,\beta_i)=(1,0)$ and $(\alpha_i',\beta_i')=(0,1)$. If the associated
recovery set $\Gamma_i = \operatorname{supp}(E_{\mathcal{E}})\cup\operatorname{supp}(E'_{\mathcal{E}})$
satisfies $|\Gamma_i|\le r+1$, then qudit $i$ can be reconstructed, following its erasure,
using only the surviving qudits in $\Gamma_i\setminus\{i\}$ together with the decoder's
shared ebits. Consequently, $\mathcal{Q}_{\mathcal{S}_E}$ is an EA-qLRC with locality $r$.
\end{theorem}

\begin{proof}
Fix a coordinate $i\in[n]$, and let $E_1,E_2\in\mathcal{S}_E$ be as in the statement. Write
\[
E_1 = (X_i\otimes A)\otimes E_{\mathcal{D}}, \qquad E_2 = (Z_i\otimes B)\otimes E'_{\mathcal{D}},
\]
where $A=\bigotimes_{j\in\Gamma_i\setminus\{i\}}X_j^{\alpha_j}Z_j^{\beta_j}$ and
$B=\bigotimes_{j\in\Gamma_i\setminus\{i\}}X_j^{\alpha_j'}Z_j^{\beta_j'}$ act strictly on the
surviving encoder qudits in $\Gamma_i\setminus\{i\}$, and $E_{\mathcal{D}},E_{\mathcal{D}}'$
act on the decoder's ebits $\mathcal{H}_{\mathcal{D}}$; this decomposition follows from
$(\alpha_i,\beta_i)=(1,0)$, $(\alpha_i',\beta_i')=(0,1)$, and
$\operatorname{supp}(E_{\mathcal{E}})\cup\operatorname{supp}(E'_{\mathcal{E}})=\Gamma_i$.

Since $E_1\in\mathcal{S}_E$, every code state $|\psi\rangle\in\mathcal{Q}_{\mathcal{S}_E}$
satisfies $E_1|\psi\rangle=|\psi\rangle$, and hence $E_1^a|\psi\rangle=|\psi\rangle$ for every
$a\in\mathbb{F}_q$. Writing this out as
$(X_i^a\otimes A^a\otimes E_{\mathcal{D}}^a)|\psi\rangle=|\psi\rangle$ and rearranging gives
\begin{equation}\label{eq:push-X}
    (X_i^a\otimes I\otimes I)|\psi\rangle = (I\otimes A^{-a}\otimes E_{\mathcal{D}}^{-a})|\psi\rangle,
    \qquad \forall\,|\psi\rangle\in\mathcal{Q}_{\mathcal{S}_E}.
\end{equation}
By the identical argument applied to $E_2^b|\psi\rangle=|\psi\rangle$,
\begin{equation}\label{eq:push-Z}
    (Z_i^b\otimes I\otimes I)|\psi\rangle = (I\otimes B^{-b}\otimes (E_{\mathcal{D}}')^{-b})|\psi\rangle,
    \qquad \forall\,|\psi\rangle\in\mathcal{Q}_{\mathcal{S}_E}.
\end{equation}
For any $|\psi\rangle\in\mathcal{Q}_{\mathcal{S}_E}$ and any $a,b\in\mathbb{F}_q$,
\[
(X_i^aZ_i^b\otimes I\otimes I)|\psi\rangle
= (X_i^a\otimes I\otimes I)\bigl[(Z_i^b\otimes I\otimes I)|\psi\rangle\bigr]
\overset{\eqref{eq:push-Z}}{=} (X_i^a\otimes I\otimes I)(I\otimes B^{-b}\otimes(E_{\mathcal{D}}')^{-b})|\psi\rangle.
\]
Since $(X_i^a\otimes I\otimes I)$ acts only on qudit $i$ while
$(I\otimes B^{-b}\otimes(E_{\mathcal{D}}')^{-b})$ acts only on $\Gamma_i\setminus\{i\}$ and
$\mathcal{D}$, these operators act on disjoint tensor factors and therefore commute, so
\[
= (I\otimes B^{-b}\otimes(E_{\mathcal{D}}')^{-b})(X_i^a\otimes I\otimes I)|\psi\rangle
\overset{\eqref{eq:push-X}}{=} (I\otimes B^{-b}A^{-a}\otimes (E_{\mathcal{D}}')^{-b}E_{\mathcal{D}}^{-a})|\psi\rangle.
\]
Thus, for every $a,b\in\mathbb{F}_q$,
\begin{equation}\label{eq:push-general}
    (X_i^aZ_i^b\otimes I\otimes I)|\psi\rangle
    = (I\otimes \widehat{O}_{a,b}\otimes I)|\psi\rangle,
    \qquad \widehat{O}_{a,b} := B^{-b}A^{-a}\otimes (E_{\mathcal{D}}')^{-b}E_{\mathcal{D}}^{-a},
\end{equation}
where $\widehat{O}_{a,b}$ acts purely on $\mathcal{H}_{\Gamma_i\setminus\{i\}}\otimes\mathcal{H}_{\mathcal{D}}$.
Since $\{X_i^aZ_i^b:a,b\in\mathbb{F}_q\}$ is a basis for $\mathcal{B}(\mathcal{H}_i)$, any
operator $O_i=\sum_{a,b}c_{a,b}X_i^aZ_i^b\in\mathcal{B}(\mathcal{H}_i)$ satisfies, by
linearity of \eqref{eq:push-general},
\[
    (O_i\otimes I\otimes I)|\psi\rangle = (I\otimes\widehat{O}_i\otimes I)|\psi\rangle,
    \qquad \widehat{O}_i := \sum_{a,b}c_{a,b}\,\widehat{O}_{a,b},
\]
for every $|\psi\rangle\in\mathcal{Q}_{\mathcal{S}_E}$, where $\widehat{O}_i$ acts purely on
$\mathcal{H}_{\Gamma_i\setminus\{i\}}\otimes\mathcal{H}_{\mathcal{D}}$. Hence, the action of
the entire single-qudit operator algebra $\mathcal{B}(\mathcal{H}_i)$ on
$\mathcal{Q}_{\mathcal{S}_E}$ can be reproduced by operators acting only on
$\Gamma_i\setminus\{i\}$ and $\mathcal{D}$, without ever acting on $\mathcal{H}_i$ itself. In
particular, for any code states $|\psi\rangle,|\phi\rangle\in\mathcal{Q}_{\mathcal{S}_E}$ and
any error operators $E_a,E_b\in\{X_i^aZ_i^b\}\otimes I_{[n]\setminus\{i\}}$ supported solely
on qudit $i$,
\[
    \langle\psi|E_a^\dagger E_b|\phi\rangle
    = \langle\psi|(I\otimes\widehat{E}_a^\dagger\widehat{E}_b\otimes I)|\phi\rangle
\]
depends only on data available from $\Gamma_i\setminus\{i\}$ and $\mathcal{D}$; taking
$|\psi\rangle,|\phi\rangle$ over an orthonormal basis of $\mathcal{Q}_{\mathcal{S}_E}$
recovers the Knill--Laflamme condition $PE_a^\dagger E_bP=\alpha_{ab}P$ of Section~II applied
to the error set $\{X_i^aZ_i^b\}$, confirming that erasure of qudit $i$ is correctable. By
the standard equivalence between the Knill--Laflamme condition and the existence of a CPTP
recovery map, there exists
\[
    \mathcal{R}_i:\mathsf{D}\bigl(\mathcal{H}_{\Gamma_i\setminus\{i\}}\otimes\mathcal{H}_{\mathcal{D}}\bigr)
    \longrightarrow \mathsf{D}(\mathcal{H}_{\Gamma_i})
\]
such that $(\mathcal{R}_i\otimes I_{[n]\setminus\Gamma_i})(\operatorname{Tr}_i\rho)=\rho$ for
every code density operator $\rho$ on $\mathcal{Q}_{\mathcal{S}_E}$, where
$\operatorname{Tr}_i$ denotes erasure (partial trace) of qudit $i$; crucially, $\mathcal{R}_i$
is never given access to $\mathcal{H}_i$, consistent with Definition~\ref{EAqLRC}.

Finally, since $\alpha_i=1\ne0$ places $i\in\operatorname{supp}(E_{\mathcal{E}})\subseteq\Gamma_i$,
we have $i\in\Gamma_i$, and therefore
\[
    |\Gamma_i\setminus\{i\}| = |\Gamma_i|-1 \le (r+1)-1 = r,
\]
so the recovery channel accesses at most $r$ surviving encoder qudits, together with the
pre-shared entanglement. Since $i\in[n]$ was arbitrary, every coordinate admits such a
recovery set and channel, and hence $\mathcal{Q}_{\mathcal{S}_E}$ is an EA-qLRC with
locality $r$.
\end{proof}

\begin{remark}\label{rem:sufficiency-only}
Theorem~\ref{Stb_EAqLRC} is a \emph{sufficient} stabilizer criterion: it shows that the
existence of the two local generators $E_1,E_2$ for every coordinate $i$ implies
$\mathcal{Q}_{\mathcal{S}_E}$ is an EA-qLRC with locality $r$. We do not claim, and the proof
above does not establish, the converse implication --- that every EA-qLRC with locality $r$
(in the sense of Definition~\ref{EAqLRC}, i.e.\ admitting the required CPTP recovery channels)
must arise from stabilizer generators of exactly this local form. Proving such a converse for
general entanglement-assisted stabilizer codes would require relating the recovery-channel
definition to the Knill--Laflamme condition on the full single-qudit error set $\{X_i^aZ_i^b\}$
via the stabilizer formalism, and is not pursued in this paper; we use Theorem~\ref{Stb_EAqLRC}
only in the forward (sufficiency) direction throughout.
\end{remark}

\begin{remark}
Theorem~\ref{Stb_EAqLRC} guarantees that a recovery channel $\mathcal{R}_i$ exists, acting
only on $\mathcal{H}_{\Gamma_i\setminus\{i\}}\otimes\mathcal{H}_{\mathcal{D}}$
(Figure~\ref{fig:recovery_channel_schematic}); in particular, reconstruction of the erased
qudit $i$ never requires access to $\mathcal{H}_i$ itself. As throughout this paper, we
assume the shared entangled states are noiseless.
\end{remark}

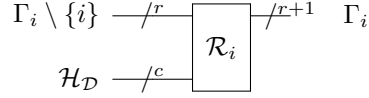
\begin{figure}[htbp]
\centering
\[
\Qcircuit @C=1.5em @R=1.5em {
\lstick{\Gamma_i\setminus\{i\}}   & \qw/^{r}   & \multigate{1}{\mathcal{R}_i} & \qw/^{r+1} & \rstick{\Gamma_i}\\
\lstick{\mathcal{H}_{\mathcal{D}}} & \qw/^{c}   & \ghost{\mathcal{R}_i}        &            &
}
\]
\caption{Schematic of the recovery channel $\mathcal{R}_i$ of Theorem~\ref{Stb_EAqLRC}:
surviving qudits $\Gamma_i\setminus\{i\}$ and the decoder's ebits $\mathcal{H}_{\mathcal{D}}$
enter $\mathcal{R}_i$, which outputs the reconstructed set $\Gamma_i$; the ebit register is
consumed, matching $\mathcal{R}_i:\mathsf{D}(\mathcal{H}_{\Gamma_i\setminus\{i\}}\otimes\mathcal{H}_{\mathcal{D}})\to\mathsf{D}(\mathcal{H}_{\Gamma_i})$.
The channel never has access to qudit $i$ itself, and its explicit gate decomposition is not
addressed in this work.}
\label{fig:recovery_channel_schematic}
\end{figure}
Theorem~\ref{Stb_EAqLRC} provides a sufficient stabilizer criterion for EA-qLRC, which enables the construction of EA-qLRCs from pairs of constituent classical linear codes via the entanglement-assisted CSS construction in Theorem~\ref{PC-EAQEC}.

\begin{proposition}\label{CSS-EA-qLRC}
Let $\mathcal{C}_i$ be an $[n,k_i,d_i]_q$ linear code over $\mathbb{F}_q$ with
parity check matrix $H_i$ for $i=1,2$ and $c = \operatorname{rank}(H_1 H_2^{T})$. If for every coordinate $i \in [n]$, 
there exist parity-check vectors 
$c_1^{(i)} \in \mathcal{C}_1^\perp$ and $c_2^{(i)} \in \mathcal{C}_2^\perp$ such that
$i \in \operatorname{supp}(c_1^{(i)}) \cap \operatorname{supp}(c_2^{(i)}),$
and $\left| \operatorname{supp}(c_1^{(i)}) \cup \operatorname{supp}(c_2^{(i)})\right| \le r+1$, then $\mathcal{Q}_E$ is an  {\qlb}$n,\kappa,\delta;c${\qrb}$_q$ EA-qLRC with locality $r$.
\end{proposition}
\begin{proof}
Let $\mathcal{Q}_E$ be the entanglement-assisted quantum code obtained from the classical linear codes $\mathcal{C}_1$ and $\mathcal{C}_2$ via the EA-CSS construction of Theorem~\ref{PC-EAQEC}. By construction, the encoder-side components of the stabilizer generators are determined by the classical dual codes $\mathcal{C}_1^\perp=\operatorname{rowspan}(H_1),\;
\mathcal{C}_2^\perp=\operatorname{rowspan}(H_2).$
Consequently, every encoder-side stabilizer operator has the form $E_{\mathcal E}=X^uZ^v,$
where $u\in\mathcal{C}_1^\perp,\;
v\in\mathcal{C}_2^\perp.$

Assume that for every coordinate $i\in[n]$, there exist parity-check vectors $c_1^{(i)}\in\mathcal{C}_1^\perp,
c_2^{(i)}\in\mathcal{C}_2^\perp$
such that $i\in
\operatorname{supp}(c_1^{(i)})
\cap
\operatorname{supp}(c_2^{(i)}),$
and $\left|
\operatorname{supp}(c_1^{(i)})
\cup
\operatorname{supp}(c_2^{(i)})
\right|
\le r+1.$
Since $i\in\operatorname{supp}(c_1^{(i)})$, we have $(c_1^{(i)})_i\neq0$. Multiplying $c_1^{(i)}$ by the nonzero scalar $((c_1^{(i)})_i)^{-1}$, we may assume without loss of generality that $(c_1^{(i)})_i=1.$
The corresponding encoder-side stabilizer operator is $E_{\mathcal E}=X^{c_1^{(i)}}Z^0,$
so that $(\alpha_i,\beta_i)=(1,0).$
Similarly, after normalizing $c_2^{(i)}$, we obtain $E'_{\mathcal E}=X^0Z^{c_2^{(i)}},$
which satisfies $(\alpha_i',\beta_i')=(0,1).$
Moreover, $\operatorname{supp}(E_{\mathcal E})=
\operatorname{supp}(c_1^{(i)}),
\operatorname{supp}(E'_{\mathcal E})=
\operatorname{supp}(c_2^{(i)}),$
and hence $\left|
\operatorname{supp}(E_{\mathcal E})
\cup
\operatorname{supp}(E'_{\mathcal E})
\right|
=
\left|
\operatorname{supp}(c_1^{(i)})
\cup
\operatorname{supp}(c_2^{(i)})
\right|
\le r+1.$
Therefore, Theorem~\ref{Stb_EAqLRC} is satisfied for every coordinate $i$, implying that $\mathcal Q_E$ is an EA-qLRC with locality $r$.
\end{proof}

If $\mathcal{C}_{1}=\mathcal{C}_{2}$, then the following corollary is an immediate consequence of Proposition~\ref{CSS-EA-qLRC}.

\begin{corollary}
Let $\mathcal{C}$ be an $[n,k,d]_q$ linear code with parity-check matrix $H$ and locality $r$. Then there exists an EA-qLRC {\qlb}$n,\kappa,\delta;c${\qrb}$_q$ with locality $r$, where $\kappa = 2k - n + c ,\; \delta = \min_{z \,\in\, \mathcal{C} \setminus (\mathcal{C} \cap \mathcal{C}^\perp)} \mathrm{wt}(z)$, and $c = \operatorname{rank}(HH^T)$.
\end{corollary}

\section{Converse Bounds for CSS-like EA-qLRCs}\label{Bounds EA-qLRC}

Having established the CSS-based construction of EA-qLRCs in the preceding section, we now turn to the fundamental question of how good such codes can be. Specifically, we derive a family of upper bounds on the parameters of a CSS-like EA-qLRC with locality $r$, thereby delineating the region of achievable parameters and quantifying the inherent trade-offs among code length, dimension, minimum distance, locality, and entanglement consumption.

We establish the relationship between the fundamental parameters of CSS-like EA-qLRCs and their classical counterparts. This relationship enables the derivation of new bounds for CSS-like EA-qLRCs based on known results for classical LRCs. Particularly, established bounds \eqref{cLRC Singleton Bound}, \eqref{CM-Bound}, \eqref{eq:griesmer}, \eqref{eq:plotkin}, and \eqref{eq:sphere-packing} for classical LRCs are used to derive corresponding bounds for CSS-like EA-qLRCs. In \cite{luo2025bounds}, the Singleton-like bound for qLRCs was derived using the \emph{CSS} construction. In this work, following a similar approach and using \cite[Theorem~21]{luo2022much}, we derive four explicit bounds for CSS-like EA-qLRC. We carry out a systematic comparative analysis of these bounds in both the finite-length and asymptotic regimes. Throughout this section, we use the following notation for optimal classical LRC parameters. Let $d_{\mathrm{opt}}^q(n,k;r)$ denote the
maximum minimum distance achievable by a code over $\mathbb{F}_q$ with length $n$, dimension $k$, and locality $r$. Let $k_{\mathrm{opt}}^q(n,d;r)$ denote the maximum dimension of such a code, and $n_{\mathrm{opt}}^q(k,d;r)$ the minimum length.

The proof of Theorem~\ref{Bounds_theorem} below passes from a classical code with known locality to a subcode, and then punctures that subcode at a block of always-zero coordinates. Both steps preserve the locality bound, but neither is entirely automatic, so we isolate them as a lemma.

\begin{lemma}\label{lem:subcode-locality}
Let $D_2$ be a linear code of length $n$ with locality $r$, and let $D\subseteq D_2$ be any linear subcode. Then $D$ also has locality at most $r$. Moreover, if every codeword of $D$ vanishes identically on a fixed coordinate set $Z\subseteq[n]$, then the code $D'$ obtained by puncturing $D$ at the coordinates in $Z$ also has locality at most $r$.
\end{lemma}

\begin{proof}
Since $D \subseteq D_2$, every vector orthogonal to all of $D_2$ is in particular orthogonal to all of $D$, so $D_2^\perp\subseteq  D^\perp$. Because $D_2$ has locality $r$, for every coordinate $i\in[n]$ there exists $h_i\in D_2^\perp$ with $i\in\operatorname{supp}(h_i)$ and $|\operatorname{supp}(h_i)|\le r+1$. As $h_i\in D_2^\perp\subseteq D^\perp$, the same vector $h_i$ witnesses locality $\le r$ for coordinate $i$ of $D$ as well; since $i\in[n]$ was arbitrary, $D$ has locality at most $r$.

For the second claim, let $i\in [n]\setminus Z$ and let $h_i\in D^\perp$ be a weight-$\le(r+1)$ witness for coordinate $i$, as just established. For any codeword $c'\in D'$, write $c\in D$ for the corresponding unpunctured codeword, so $c_j=0$ for all $j\in Z$. Then
\[
0 = h_i\cdot c = \sum_{j\in Z} (h_i)_j\, c_j \;+\; \sum_{j\notin Z} (h_i)_j\, c_j = \sum_{j\notin Z} (h_i)_j\, c_j = h_i|_{[n]\setminus Z}\cdot c',
\]
so the restriction $h_i|_{[n]\setminus Z}$ annihilates every codeword of $D'$, has weight at most $|\operatorname{supp}(h_i)|\le r+1$ (restricting a vector cannot increase its weight), and retains a nonzero entry at coordinate $i$ since $i \notin Z$. Hence $h_i|_{[n]\setminus Z}$ witnesses locality $\le r$ for coordinate $i$ of $D'$. As $i\in [n] \setminus Z$ was arbitrary, $D'$ has locality at most $r$.
\end{proof}

\begin{theorem}\label{Bounds_theorem}
  Let $\mathcal{Q}_E$ be an 
  {\qlb}$n,\kappa=k_1+k_2-n+c,\delta;c${\qrb}$_q$ CSS-like EA-qLRC with locality $r$ constructed from Proposition~\ref{CSS-EA-qLRC} using two classical $[n,k_i,d_i]_q$ codes $\mathcal{C}_i$ of locality $r$, for $i=1,2$. Then
  \begin{eqnarray*}
\qquad & 2\delta \leq d_{\mathrm{opt}}^q(k_1+c,\kappa;r)
                   + d_{\mathrm{opt}}^q(k_2+c,\kappa;r),\\
\qquad &2\kappa \leq k_{\mathrm{opt}}^q(k_1+c,\delta;r)
                   + k_{\mathrm{opt}}^q(k_2+c,\delta;r),\\
\qquad & n+\kappa+c \geq 2\,n_{\mathrm{opt}}^q(\kappa,\delta;r). 
\end{eqnarray*}
\end{theorem}

\begin{proof}
Let $\mathcal{C}_i$ be classical linear codes over $\mathbb{F}_q$ with parameters $[n,k_i,d_i]_q$ for $i=1,2$. Then, by Theorem~\ref{PC-EAQEC}, there exists an entanglement-assisted quantum code $\mathcal{Q}_E$ and the minimum distance of $\mathcal{Q}_E$ satisfies $\delta
  = \min\!\Bigl(
      \mathrm{wt}\!\bigl(\mathcal{C}_2\setminus
        (\mathcal{C}_1^{\perp}\cap\mathcal{C}_2)\bigr),\;
      \mathrm{wt}\!\bigl(\mathcal{C}_1\setminus
        (\mathcal{C}_1\cap\mathcal{C}_2^{\perp})\bigr)
    \Bigr).$
We consider two cases according to which term achieves this minimum.

\smallskip
\textbf{Case 1.}
 If $\mathrm{wt}\!\left(\mathcal{C}_2 \setminus (\mathcal{C}_1^\perp \cap \mathcal{C}_2)\right) \leq \mathrm{wt}\!\left(\mathcal{C}_1 \setminus (\mathcal{C}_1 \cap \mathcal{C}_2^{\perp})\right)$, then $\delta=\mathrm{wt}\!\left(\mathcal{C}_2 \setminus (\mathcal{C}_1^\perp \cap \mathcal{C}_2)\right)$.
Since $(\mathcal{C}_1^\perp \cap \mathcal{C}_2)\subseteq \mathcal{C}_2$,
the subcode $\mathcal{C}_1^\perp \cap \mathcal{C}_2$ has dimension
$n - k_1 - c$ and admits a generator matrix
$\begin{pmatrix} I_{n-k_1-c} & A_1 \end{pmatrix}$.
Completing this to a generator matrix of $\mathcal{C}_2$ gives
\[
  G_{\mathcal{C}_2}
  = \begin{pmatrix}
      I_{n-k_1-c} & A \\
      O_{\kappa\times(n-k_1-c)} & P
    \end{pmatrix}.
\]
The matrix
$\begin{pmatrix} O_{\kappa\times(n-k_1-c)} & P\end{pmatrix}$
generates a code $\mathcal{C}$ with parameters $[n,\kappa,d]_q$.
Every nonzero codeword of $\mathcal{C}$ lies in
$\mathcal{C}_2 \setminus (\mathcal{C}_1^\perp \cap \mathcal{C}_2)$,
so $d \geq \delta$. Let $\mathcal{C}'$ be the code generated by
$P$ alone and has parameters
$[k_1+c,\,\kappa,\,d(\mathcal{C}')\ge\delta]_q$.
Since, $\mathcal{C}_2$ has  locality $r$ and
$\mathcal{C}\subseteq\mathcal{C}_2$ and, by construction, every codeword of $\mathcal{C}$ vanishes
identically on the first $n-k_1-c$ coordinates, Lemma~\ref{lem:subcode-locality} gives that
$\mathcal{C}'$ --- obtained from $\mathcal{C}$ by puncturing exactly those coordinates --- has
locality at most $r$. Therefore,
\[
  \delta\;\leq\; d_{\mathrm{opt}}^q(k_1+c,\kappa;r),\;
  \kappa\;\leq\; k_{\mathrm{opt}}^q(k_1+c,\delta;r),
\]
\begin{equation}\label{equ_3}
  k_1+c\;\geq\; n_{\mathrm{opt}}^q(\kappa,\delta;r).
\end{equation}

\smallskip
\textbf{Case 2.}
If $\mathrm{wt}\!\left(\mathcal{C}_2 \setminus (\mathcal{C}_1^\perp \cap \mathcal{C}_2)\right) \geq \mathrm{wt}\!\left(\mathcal{C}_1 \setminus (\mathcal{C}_1 \cap \mathcal{C}_2^{\perp})\right)$, then $\delta=\mathrm{wt}\!\left(\mathcal{C}_1 \setminus (\mathcal{C}_1 \cap \mathcal{C}_2^\perp)\right)$.
Since, $\mathcal{C}_1$ likewise has classical locality $r$. Applying the same argument to $\mathcal{C}_1$ with the subcode
$\mathcal{C}_1\cap\mathcal{C}_2^{\perp}\subseteq\mathcal{C}_1$, and invoking Lemma~\ref{lem:subcode-locality} exactly as in Case~1 (with the roles of $\mathcal{C}_1$ and $\mathcal{C}_2$ exchanged) to pass to the punctured complement, we obtain a code $\mathcal{C}''$ with parameters
$[k_2+c,\,\kappa,\,d(\mathcal{C}'')\ge\delta]_q$ and locality at most $r$. Therefore,
\[
  \delta\;\leq\; d_{\mathrm{opt}}^q(k_2+c,\kappa;r),\;
  \kappa\;\leq\; k_{\mathrm{opt}}^q(k_2+c,\delta;r),
\]
\begin{equation}\label{equ_4}
  k_2+c\;\geq\; n_{\mathrm{opt}}^q(\kappa,\delta;r).
\end{equation}

\smallskip
Adding the two $\delta$-bounds and the two $\kappa$-bounds from
\textbf{Case~1} and\textbf{~2} yields the first two claimed inequalities. Adding~\eqref{equ_3} and~\eqref{equ_4} gives $(k_1+c)+(k_2+c) \;\geq\; 2\,n_{\mathrm{opt}}^q(\kappa,\delta;r).$
Substituting $k_1+k_2+2c= n+\kappa+c$ yields $ n+\kappa+c \;\geq\; 2\,n_{\mathrm{opt}}^q(\kappa,\delta;r),$
which is the third claimed inequality.
\end{proof}

\begin{corollary}[\textbf{Singleton-like Bound}]\label{cor:singleton}
  Let $\mathcal{Q}_E$ is an {\qlb}$n,\kappa,\delta;c${\qrb}$_q$ CSS-like EA-qLRC with locality $r$ constructed from Proposition~\ref{CSS-EA-qLRC}  using two classical $[n,k_i,d_i]_q$ codes $\mathcal{C}_i$ of locality $r$,  for $i=1,2$. Then
  \begin{equation}\label{EA-qLRC_SB}
    2\delta \;\leq\; n - \kappa + c
      - 2\!\left\lceil \frac{\kappa}{r} \right\rceil + 4.
  \end{equation}
\end{corollary}

\begin{proof}
Applying the classical LRC Singleton bound~\eqref{cLRC Singleton Bound}, namely $d_{\mathrm{opt}}^q(n,k;r)\leq n-k-\lceil k/r\rceil+2$,
to each term in the first inequality of
Theorem~\ref{Bounds_theorem} gives
\begin{align*}
  2\delta
  &\leq d_{\mathrm{opt}}^q(k_1+c,\kappa;r)
       + d_{\mathrm{opt}}^q(k_2+c,\kappa;r) \\
  &\leq \bigl(k_1+c-\kappa-\lceil\kappa/r\rceil+2\bigr)
       +\bigl(k_2+c-\kappa-\lceil\kappa/r\rceil+2\bigr) \\
  &= (k_1+k_2) + 2c - 2\kappa - 2\lceil\kappa/r\rceil + 4.
\end{align*}
Substituting $k_1+k_2 = n+\kappa-c$ yields $2\delta \;\leq\; n - \kappa + c - 2\lceil\kappa/r\rceil + 4.$
\end{proof}

\begin{remark}
Independently, Li \emph{et al.}~\cite{li2026eaqlrc} establish an upper bound on the locality itself for CSS-like EA-qLRCs constructed from Proposition~\ref{CSS-EA-qLRC}, namely $r \le \min\{n-1, \dim(C_1 + C_2)\}$, which is independent of the
minimum distance $\delta$ and the entanglement consumption $c$. We do not pursue this direction here; instead, Corollary~\ref{cor:CM}, Corollary~\ref{Other_theorem}, and Theorem~\ref{Asynmmetric other Bounds} characterize the distance--dimension--locality--entanglement trade-off for a fixed locality $r$. The two types of bounds are complementary: theirs bounds constrain how large $r$ can be for given $C_1$, $C_2$, while ours bounds $\delta$ (or $\kappa$, or $n$) for a fixed, given $r$.
\end{remark}
\begin{remark}
The bound \eqref{EA-qLRC_SB} carries no analogue of the hypothesis
$\delta\le(n+2)/2$ appearing in Lemma~\ref{lem:EA-singleton}: it is derived from the
classical LRC Singleton bound \eqref{cLRC Singleton Bound} via Theorem~\ref{Bounds_theorem},
not from Lemma~\ref{lem:EA-singleton}, and \eqref{cLRC Singleton Bound} holds for all
admissible classical parameters. Example~\ref{ex:TB} below attains
\eqref{EA-qLRC_SB} with equality at $\delta=5>(n+2)/2=4$, outside the range of
Lemma~\ref{lem:EA-singleton}.
\end{remark}

\begin{theorem}\cite[Theorem~3]{li2026eaqlrc}
Let $\mathcal{Q}_E$ be an {\qlb}$n,\kappa,\delta;c${\qrb}$_q$ EA-qLRC with locality $r$ constructed from Proposition~\ref{CSS-EA-qLRC} using two classical codes $\mathcal{C}_i$ with $d_i \ge 2$, $i=1,2$.
Then
\[
r \;\le\; \min\{\,n-1,\; k_1 + k_2 - \dim(C_1 \cap C_2)\,\}.
\]
\end{theorem}

\begin{corollary}[\textbf{Cadambe-Mazumdar-like Bound}]\label{cor:CM}
  Let $\mathcal{Q}_E$ be an {\qlb}$n,\kappa,\delta;c${\qrb}$_q$ CSS-like
 EA-qLRC with locality $r$ constructed from Proposition ~\ref{CSS-EA-qLRC}. Then
   \begin{equation}\label{EA-qLRC_CM}
  %
  2\kappa \;\le\; \min_{\substack{0\le\tau_1\le T_1\\ 0\le\tau_2\le T_2}}
\Bigl\{(\tau_1+\tau_2)r + k^{(q)}_{opt}\bigl(k_1+c-\tau_1(r+1),\delta\bigr)
+ k^{(q)}_{opt}\bigl(k_2+c-\tau_2(r+1),\delta\bigr)\Bigr\},
\end{equation}
where $T_i=\min\bigl\{\lceil (k_i+c)/(r+1)\rceil,\ \lceil\kappa/r\rceil\bigr\}$ and $k_i$ are the dimensions
  of the constituent classical codes $\mathcal{C}_i$ of classical locality $r$, for $i=1,2$.
\end{corollary}

\begin{proof}
Applying the Cadambe-Mazumdar bound~\eqref{CM-Bound}, namely
$k_{\mathrm{opt}}^q(n,d;r)\leq\min_{\tau\ge0}\{\tau\;r+k_{\mathrm{opt}}^{(q)}(n-\tau(r+1),d)\}$,
independently to each term in the second inequality of
Theorem~\ref{Bounds_theorem} and taking the joint minimum over
$\tau_1,\tau_2\in\mathbb{Z}_{\ge0}$ yields~\eqref{EA-qLRC_CM}.
\end{proof}

\begin{theorem}\label{Asynmmetric other Bounds}
Let $\mathcal{Q}_E$ be an {\qlb}$n,\kappa,\delta;c${\qrb}$_q$ CSS-like EA-qLRC with locality $r$ constructed from Proposition ~\ref{CSS-EA-qLRC}  using two classical $[n,k_i,d_i]_q$ codes $\mathcal{C}_i$ of locality $r$,  for $i=1,2$. Then the following bounds hold:
\begin{enumerate}

\item \textbf{Griesmer-like Bound}
\[
    n \;\ge\; 2\max_{0\le\tau\le \left\lceil\frac{\kappa}{r}\right\rceil-1}
      \left\{\tau(r+1)+\sum_{i=0}^{\kappa-\tau r-1}\left\lceil\frac{\delta}{q^i}\right\rceil\right\}-\kappa-c.\]
\
 \item \textbf{Plotkin-like Bound}
\[  2\delta\; \le\;
    \min _{0\le \tau_1,\tau_2 \le \lceil \frac{\kappa }{r}\rceil -1}
    (q-1)\left\{\frac{q^{\kappa -\tau_1 r-1}\big(k_1+c-\tau_1(r+1)\big)}{q^{\kappa -\tau_1 r}-1}+\frac{q^{\kappa -\tau_2 r-1}\big(k_2+c-\tau_2(r+1)\big)}{q^{\kappa -\tau_2 r}-1}\right\}.\]
    
 \item \textbf{Sphere-packing-like Bound}
 \[ \kappa\; \le\; n + c - \left(\max_{0\leq\tau_1\leq \lfloor\frac{k_1+c-1}{r+1}\rfloor, 0\leq\tau_2\leq \lfloor\frac{k_2+c-1}{r+1}\rfloor} \left\{ \tau_1+\tau_2 + \log_q \left[ V_q\big(k_1+c-\tau_1(r+1), t\big) \cdot V_q\big(k_2+c-\tau_2(r+1), t\big) \right] \right\}\right),\]
 where $t = \lfloor \frac{\delta-1}{2} \rfloor$, $V_q(M, t) = \sum_{j=0}^{t} \binom{M}{j}(q-1)^j$ is the volume of the classical $q$-ary Hamming ball.
\end{enumerate}
\end{theorem}

\begin{proof}
Recall from Theorem~\ref{Bounds_theorem} that
\begin{equation}\tag{\ref{equ_3}}
  k_1+c\;\geq\; n_{\mathrm{opt}}^q(\kappa,\delta;r),
\end{equation}
\begin{equation}\tag{\ref{equ_4}}
  k_2+c\;\geq\; n_{\mathrm{opt}}^q(\kappa,\delta;r),
\end{equation}
and from the first two inequalities of that theorem,
\[
  \delta\;\leq\; d_{\mathrm{opt}}^q(k_1+c,\kappa;r),
  \qquad
  \delta\;\leq\; d_{\mathrm{opt}}^q(k_2+c,\kappa;r),
\]
\[
  \kappa\;\leq\; k_{\mathrm{opt}}^q(k_1+c,\delta;r),
  \qquad
  \kappa\;\leq\; k_{\mathrm{opt}}^q(k_2+c,\delta;r).
\]

\smallskip
\textbf{Griesmer-like bound.}
Since $n_{\mathrm{opt}}^q(\kappa,\delta;r)$ is, by definition, the length of an optimal
classical $(\cdot,\kappa,\delta,r)_q$-LRC, it satisfies the classical Griesmer-like
bound~\eqref{eq:griesmer} in its own right:
\[
  n_{\mathrm{opt}}^q(\kappa,\delta;r) \;\ge\;
  \max_{0\le\tau\le\lceil\kappa/r\rceil-1}
  \left\{\tau(r+1)+\sum_{i=0}^{\kappa-\tau r-1}\left\lceil\frac{\delta}{q^i}\right\rceil\right\}.
\]
Combining this with~\eqref{equ_3} and~\eqref{equ_4} separately,
\[
  k_1+c \;\ge\; \max_{0\le\tau_1\le\lceil\kappa/r\rceil-1}
  \left\{\tau_1(r+1)+\sum_{i=0}^{\kappa-\tau_1 r-1}\left\lceil\frac{\delta}{q^i}\right\rceil\right\},
  \qquad
  k_2+c \;\ge\; \max_{0\le\tau_2\le\lceil\kappa/r\rceil-1}
  \left\{\tau_2(r+1)+\sum_{i=0}^{\kappa-\tau_2 r-1}\left\lceil\frac{\delta}{q^i}\right\rceil\right\}.
\]
Denote by $G(\kappa,\delta,r)$ the common maximizing expression on the right-hand side of
both inequalities; note that $G(\kappa,\delta,r)$ depends only on $\kappa,\delta,r$ and not on
$k_1$ or $k_2$, so both inequalities bound $k_1+c$ and $k_2+c$ below by the \emph{same}
quantity $G(\kappa,\delta,r)$. Adding the two inequalities gives
$  (k_1+c)+(k_2+c) \;\ge\; 2\,G(\kappa,\delta,r)$. 
Substituting $k_1+k_2=n+\kappa-c$ (from $\kappa=k_1+k_2-n+c$) yields
\[
  n \;\ge\; 2\max_{0\le\tau\le\lceil\kappa/r\rceil-1}
  \left\{\tau(r+1)+\sum_{i=0}^{\kappa-\tau r-1}\left\lceil\frac{\delta}{q^i}\right\rceil\right\}-\kappa-c.
\]

\smallskip
\textbf{Plotkin-like bound.}
Applying the classical Plotkin-like bound~\eqref{eq:plotkin} to
$\delta\leq d_{\mathrm{opt}}^q(k_1+c,\kappa;r)$ and
$\delta\leq d_{\mathrm{opt}}^q(k_2+c,\kappa;r)$ separately, with independent optimization
variables $\tau_1$ and $\tau_2$, gives
\[
  \delta \;\le\;
  \min_{0\le\tau_1\le\lceil\kappa/r\rceil-1}
  \left\{\frac{q^{\kappa-\tau_1 r-1}(q-1)\bigl(k_1+c-\tau_1(r+1)\bigr)}{q^{\kappa-\tau_1 r}-1}\right\},
\]
\[
  \delta \;\le\;
  \min_{0\le\tau_2\le\lceil\kappa/r\rceil-1}
  \left\{\frac{q^{\kappa-\tau_2 r-1}(q-1)\bigl(k_2+c-\tau_2(r+1)\bigr)}{q^{\kappa-\tau_2 r}-1}\right\}.
\]
Adding these two inequalities and taking the joint minimum over $(\tau_1,\tau_2)$ yields the
stated bound.

\textbf{Sphere-packing-like bound.}
By the second inequality of Theorem~\ref{Bounds_theorem},
\[
  2\kappa \;\le\; k_{\mathrm{opt}}^q(k_1+c,\delta;r) + k_{\mathrm{opt}}^q(k_2+c,\delta;r).
\]
Since $k_{\mathrm{opt}}^q(n',\delta;r)$ is the dimension of an optimal classical
$(n',\cdot,\delta,r)_q$-LRC, it satisfies the classical sphere-packing-like
bound~\eqref{eq:sphere-packing} in its own right. Applying this with $n'=k_1+c$ and
$n'=k_2+c$, using independent optimization variables $\tau_1$ and $\tau_2$, gives
\[
  k_{\mathrm{opt}}^q(k_1+c,\delta;r) \;\le\; (k_1+c) -
  \max_{0\le\tau_1\le\lfloor(k_1+c-1)/(r+1)\rfloor}
  \left\{\tau_1+\log_q\bigl[V_q(k_1+c-\tau_1(r+1),t)\bigr]\right\},
\]
\[
  k_{\mathrm{opt}}^q(k_2+c,\delta;r) \;\le\; (k_2+c) -
  \max_{0\le\tau_2\le\lfloor(k_2+c-1)/(r+1)\rfloor}
  \left\{\tau_2+\log_q\bigl[V_q(k_2+c-\tau_2(r+1),t)\bigr]\right\}.
\]
Adding these two inequalities and using
$\max_{\tau_1}\{f(\tau_1)\}+\max_{\tau_2}\{g(\tau_2)\}=\max_{\tau_1,\tau_2}\{f(\tau_1)+g(\tau_2)\}$
together with $\log_q A+\log_q B=\log_q(AB)$ gives
\[
  k_{\mathrm{opt}}^q(k_1+c,\delta;r)+k_{\mathrm{opt}}^q(k_2+c,\delta;r)
  \;\le\; (k_1+k_2+2c) -
  \max_{\tau_1,\tau_2}
  \left\{\tau_1+\tau_2+\log_q\Bigl[V_q(k_1+c-\tau_1(r+1),t)\cdot
  V_q(k_2+c-\tau_2(r+1),t)\Bigr]\right\}.
\]
Combining with the inequality $2\kappa\le k_{\mathrm{opt}}^q(k_1+c,\delta;r)+k_{\mathrm{opt}}^q(k_2+c,\delta;r)$
above, and substituting $k_1+k_2=n+\kappa-c$, yields
\[
  2\kappa \;\le\; n+\kappa+c -
  \max_{\tau_1,\tau_2}
  \left\{\tau_1+\tau_2+\log_q\Bigl[V_q(k_1+c-\tau_1(r+1),t)\cdot
  V_q(k_2+c-\tau_2(r+1),t)\Bigr]\right\},
\]
which rearranges to the stated bound.
\end{proof}

If $\mathcal{C}_1=\mathcal{C}_2$, then, using Theorem~\ref{Asynmmetric other Bounds}, we obtain the following corollary:

\begin{corollary}\label{Other_theorem}
Let $\mathcal{Q}_E$ be an {\qlb}$n,\kappa=2k-n+c,\delta;c${\qrb}$_q$ CSS-like  EA-qLRC with locality $r$ constructed from Proposition~\ref{CSS-EA-qLRC} with $\mathcal{C}_1=\mathcal{C}_2=\mathcal{C}$, of classical locality $r$. Then the following bounds hold.
\begin{enumerate}
    \item \textbf{Griesmer-like Bound}
 \begin{equation}\label{EA-qLRC Griesmer Bound}
      n \;\ge\;
        2\max_{0\leq\tau\leq \lceil\frac{\kappa}{r}\rceil-1}
        \left\{
          \tau(r+1)
          +\sum_{i=0}^{\kappa-\tau r-1}\!\left\lceil\frac{\delta}{q^i}\right\rceil
        \right\}-\kappa-c.
 \end{equation}

    \item \textbf{Plotkin-like Bound}

  \begin{equation}\label{EA-qLRC Ploktin Bound}
      2\delta\; \le\;
    \min _{0\le \tau \le \lceil \frac{\kappa }{r}\rceil -1}
    \left\{\frac{q^{\kappa -\tau r-1}(q-1)\big(n+\kappa +c-2\;\tau(r+1)\big)}{q^{\kappa -\tau r}-1}\right\}.
  \end{equation}
 \item \textbf{Sphere-packing-like Bound}
    \begin{equation}\label{EA-qLRC Sphere Packing Bound}
         \kappa\; \le\; n+ c -2\;\max_{0\leq\tau\leq \lfloor\frac{n+\kappa+c-2}{2(r+1)}\rfloor} \left\{ \tau + \log_q \left[ V_q\big(\frac{n+\kappa+c}{2}-\tau(r+1), t\big) \right] \right\},
    \end{equation}
where $t = \lfloor \frac{\delta-1}{2} \rfloor$, $V_q(M, t) = \sum_{j=0}^{t} \binom{M}{j}(q-1)^j$ is the volume of the classical $q$-ary Hamming ball.

\end{enumerate}
\end{corollary}

\begin{remark}
While the Singleton-like bound is entirely alphabet-independent, the Griesmer, Plotkin, and Sphere-packing-like bounds become tighter as the finite field alphabet size $q$ decreases (e.g., in binary where $q=2$).
\end{remark}

Based on Corollaries~\ref{cor:singleton} and \ref{Other_theorem}, four explicit bounds for CSS-like EA-qLRCs have been established. An important question is to determine the relative performance of these bounds and to investigate whether they improve upon the Singleton-like bound in \eqref{EA-qLRC_SB}. Because these bounds apply to different parameter classes, a direct comparison is not straightforward. Consequently, we distinguish between finite-length and asymptotic code regimes and carry out a separate comparative analysis for each case, taking into account the relevant parameter variations.

\textbf{For finite code lengths}, no single bound is universally tightest. The maximum possible dimension $\kappa$ for a given code length $n$ depends entirely on the specific regime; specifically, the alphabet size $q$, the distance $\delta$, and the locality parameter $r$. Here is the comparison of how each bound constrains $\kappa$ relative to $n$, and the specific regimes where each bound becomes tight.

\begin{itemize}

\item \textbf{The Singleton-like Bound}: This bound is typically tight only for very large alphabet sizes ($q \gg n$) and small distances $\delta$. For small $q$, this bound is not tight (Refer to Figure~\ref{T_Singleton Bound}).

\item \textbf{The Griesmer-like Bound}: This is the tightest bound for small finite fields (especially $q=2, 3$) and low-to-moderate dimensions $\kappa$. It strictly outperforms the Singleton bound over small alphabets (Refer to Figure~\ref{T_Griesmer}).

\item \textbf{The Plotkin-like Bound}: This is the absolute tightest bound in the high distance regime (Refer to Figure~\ref{T_Ploktin Bound}).

\item \textbf{The Sphere-Packing-like Bound}:  This bound is tightest in the asymptotic regime for moderate relative distances ($\delta/n$) and small $q$. While Griesmer is excellent for small $\kappa$, the sphere-packing bound usually provides the strictest limits on the asymptotic capacity (rate $\kappa/n$) for highly entangled, large-blocklength codes (Refer to Figure~\ref{Finite All Bound}).

\end{itemize}
Numerical comparisons for different parameters are illustrated in Figures~\ref{T_Singleton Bound}, \ref{T_Griesmer},  \ref{T_Ploktin Bound} and \ref{Finite All Bound}.

\begin{figure}[h!]
\begin{minipage}[l]{8cm}
    \centering
    \includegraphics[width=2.8in]{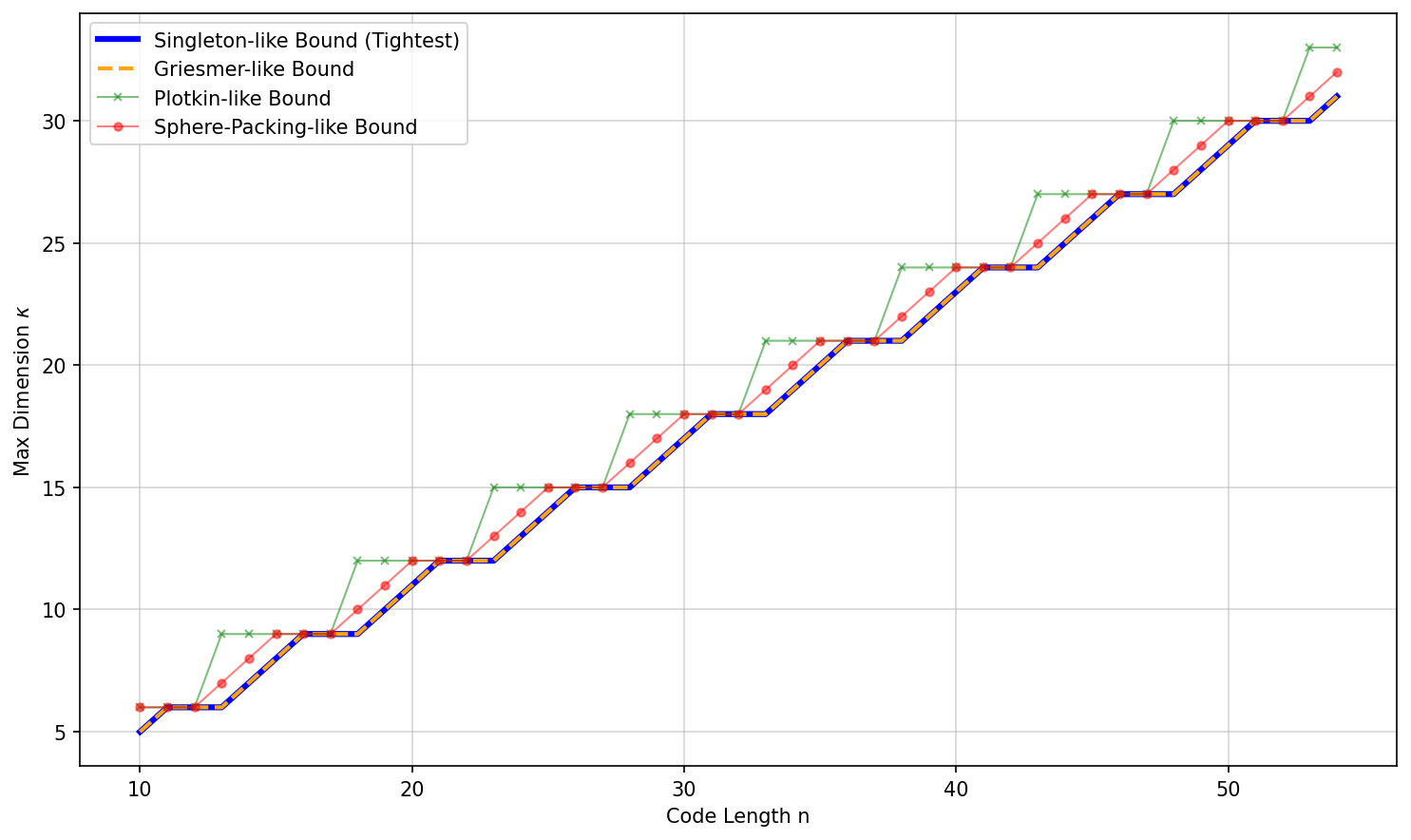}
\caption{Comparison of the bounds  with  $q = 64 ,\delta = 3,r = 3\;\& \;c = 1$.}
\label{T_Singleton Bound}
\end{minipage}
\hfill
\begin{minipage}[r]{8cm}
\includegraphics[width=2.9in]{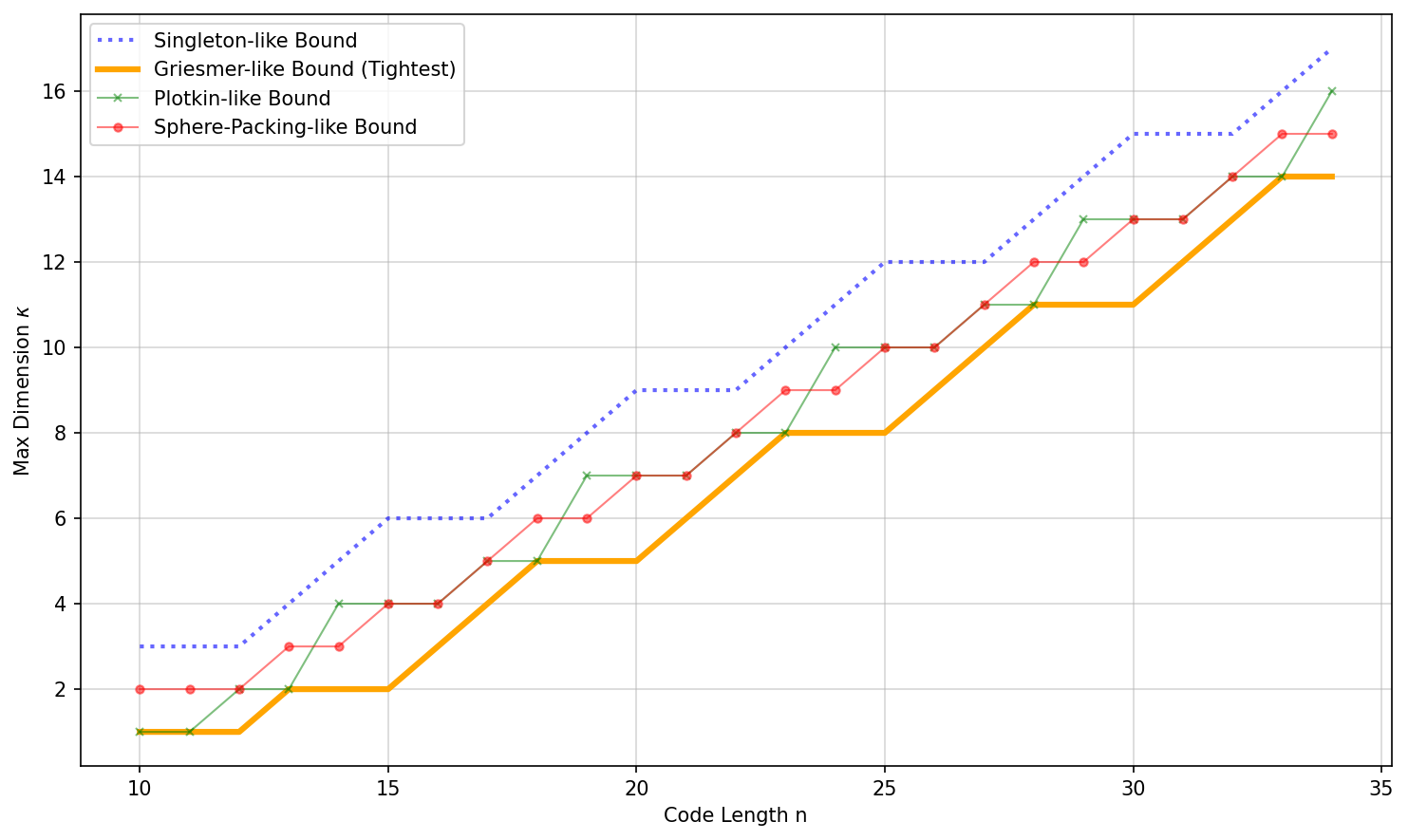}%
\caption{Comparison of the bounds   with  $q =2 ,\delta = 5,r = 3\;\& \;c = 1$.}
\label{T_Griesmer}
\end{minipage}
\end{figure}
\begin{figure}[h!]
\begin{minipage}[l]{8cm}
    \centering
    \includegraphics[width=2.8in]{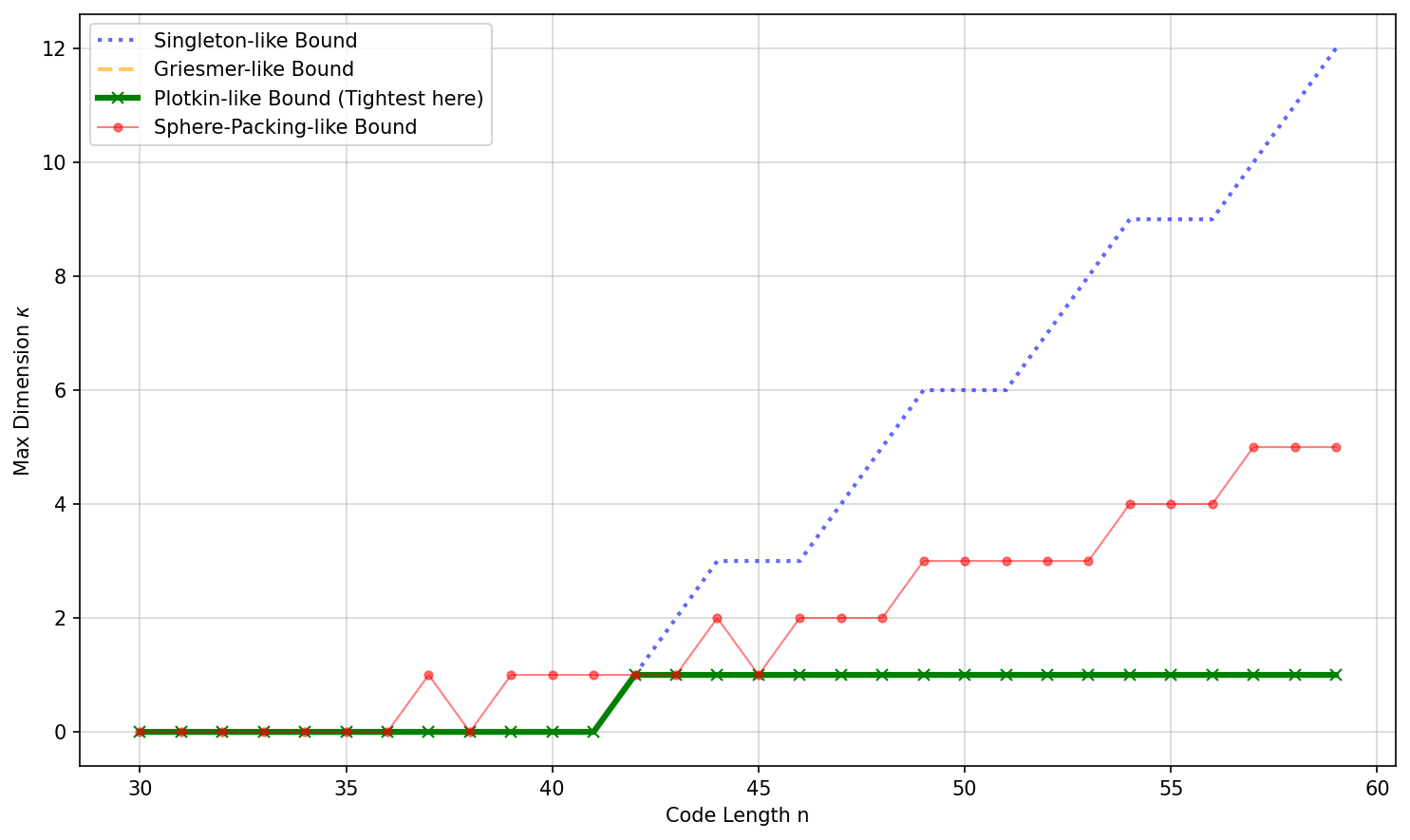}
    \caption{Comparison of the bounds  with $q =2 ,\delta = 22,r = 3\;\& \;c = 1$.}
\label{T_Ploktin Bound}
\end{minipage}
\hfill
\begin{minipage}[r]{8cm}
\centering
  \includegraphics[width=2.8in]{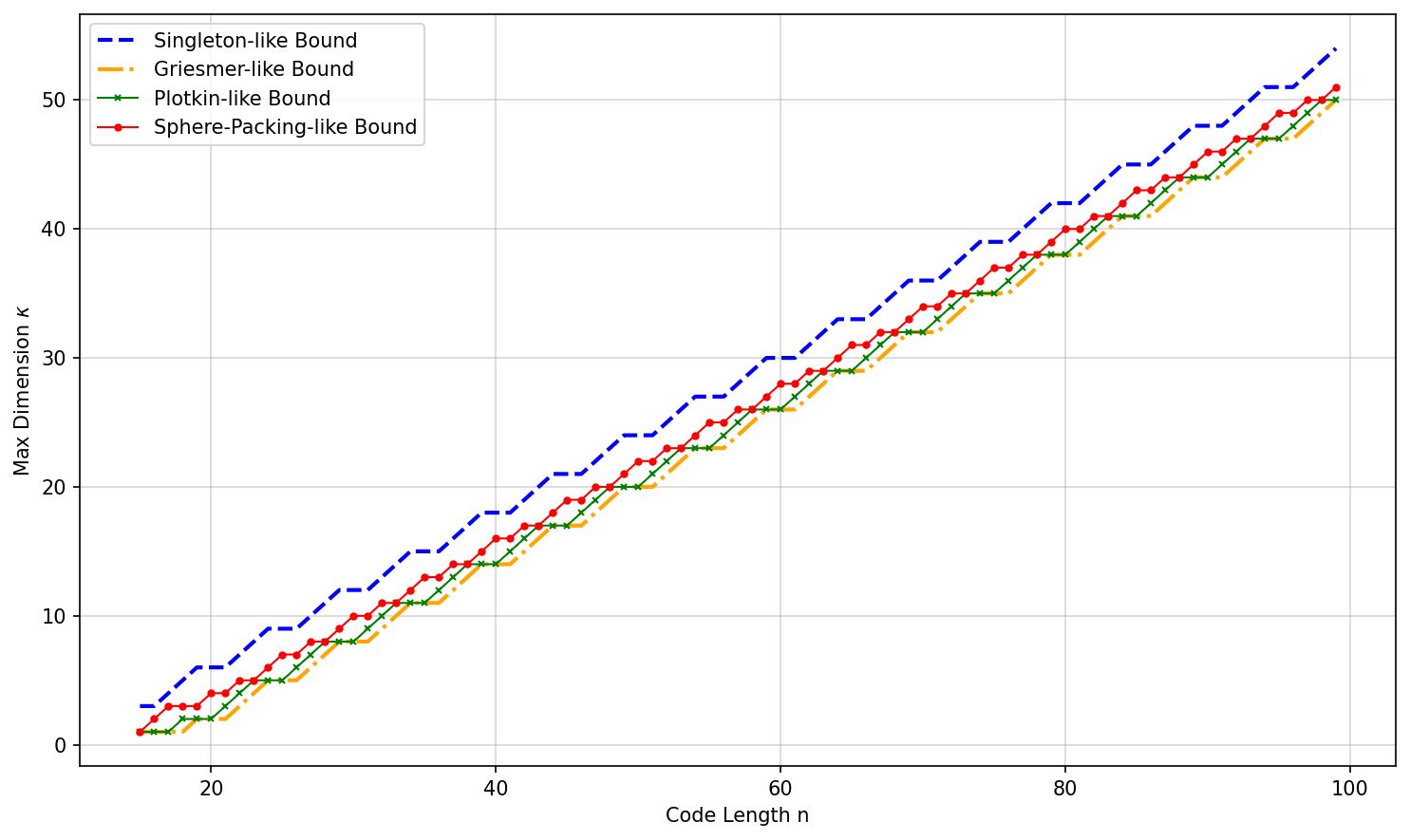}
 \caption{Comparison of the bounds   with $q = 2, \delta = 7,r=3$ \& $c = 1$.}
 \label{Finite All Bound}
\end{minipage}
\end{figure}

 \medskip
 \textbf{In the asymptotic regime}, where the code length $n$ tends to infinity, it is more meaningful to compare the bounds through their asymptotic expressions rather than their finite-length counterparts. It can be shown that all CSS-like EA-qLRC bounds obtained in  Corollary~\ref{Other_theorem} are strictly stronger than the  Singleton-like bound in \eqref{EA-qLRC_SB}. Moreover, the asymptotic analysis allows us to characterize the relative tightness among the proposed bounds. To facilitate this comparison, we first derive their asymptotic forms. Given an {\qlb}$n,\kappa,\delta;c${\qrb}$_q$ CSS-like EA-qLRC with locality $r$, let $R = \lim_{n \to \infty} \frac{\kappa}{n}$,\;$\Delta = \lim_{n \to \infty} \frac{\delta}{n}$,\;$\gamma = \lim_{n \to \infty} \frac{c}{n}$,\;\text{and}\; $R_{net} =\lim_{n \to \infty} \frac{\kappa-c}{n}$ be its rate, relative distance, entanglement rate, and net rate, respectively.

  To examine the relative performance of the bounds asymptotically, we now present the asymptotic forms of these three bounds. The asymptotic version of the Singleton-like bound in \eqref{EA-qLRC_SB} is obtained by dividing each term of \eqref{EA-qLRC_SB} by $n$ and appropriately rearranging the terms, as follows:
\begin{equation}\label{ASB}
    R  \le \frac{r}{r+2}(1 + \gamma) - \frac{2r}{r+2}\Delta+o(1),\;\; n\to\infty.
\end{equation}

Since the improved transmission rate of an entanglement-assisted quantum code is achieved at the cost of consuming additional entanglement, the net rate is the natural performance measure for the entanglement-assisted quantum codes. While dividing \eqref{EA-qLRC_SB} by $n$ and rearranging the terms, we get
\[ R_{net}\leq \frac{r}{r+2}-\frac{2}{r+2}\;\gamma-\frac{2r}{r+2}\;\Delta+o(1),\;\; n\to\infty.\]

 Unlike the other bounds, the Griesmer-like bound in \eqref{EA-qLRC Griesmer Bound} is formulated as a lower bound on the code length $n$ for given values of the dimension, minimum distance, and locality. Consequently, taking the asymptotic limit $n\to\infty$ directly in \eqref{EA-qLRC Griesmer Bound} and \eqref{EA-qLRC Ploktin Bound} does not yield a meaningful asymptotic characterization, as both yield the identical asymptotic bound. Therefore, to enable an explicit comparison with other bounds, and following the similar approach of \cite{li2025improved} and \cite{li2025optimal} to obtain an explicit asymptotic bound, consider $\tau=\lceil\kappa/r\rceil-1$  in the Griesmer-like bound \eqref{EA-qLRC Griesmer Bound}, which gives
\begin{equation}\label{griesmer1}
n \;\ge\; 2\Bigl(\Bigl\lceil\frac{\kappa}{r}\Bigr\rceil-1\Bigr)(r+1)
+ 2\sum_{i=0}^{r-1}\Bigl\lceil\frac{\delta}{q^{i}}\Bigr\rceil-\kappa-c .
\end{equation}

Since $\lceil\kappa/r\rceil=\kappa/r+o(1)$, dividing by $n$ and letting $n\to\infty$
absorbs the rounding into the $o(1)$ term, and the asymptotic form is unchanged.

For any fixed integer $m$ with $1\leq m \leq r$, we can write the sum
\[\sum_{i=0}^{r-1}\left\lceil\frac{\delta}{q^i}\right\rceil=\sum_{i=0}^{m-1}\left\lceil\frac{\delta}{q^i}\right\rceil+ \sum_{i=m}^{r-1}\left\lceil\frac{\delta}{q^i}\right\rceil. \] This implies that \[\sum_{i=0}^{r-1}\left\lceil\frac{\delta}{q^i}\right\rceil
\ge
\sum_{i=0}^{m-1}\frac{\delta}{q^i} + (r-m).\] 
From \eqref{griesmer1}, which further implies that

\begin{eqnarray*}
\qquad & \ge &
2\left(\frac{\kappa}{r}-1\right)(r+1)
+2\sum_{i=0}^{m-1}\frac{\delta}{q^i}
+2(r-m)
-\kappa-c \\
\qquad &=&
2\left(\frac{\kappa}{r}-1\right)(r+1)
+ 2\delta\;\frac{q^m-1}{q^m-q^{m-1}} +2(r-m) -\kappa-c \\
\qquad & = & \frac{r+2}{r} \kappa + 2\delta\;\frac{q^m-1}{q^m-q^{m-1}} -2(m+1)-c. 
\end{eqnarray*}
Thus
\begin{equation}\label{AGB}
    R \le
\frac{r}{r+2}\bigl(1+\gamma\bigr)-
\frac{2r}{r+2}\;
\frac{q^m-1}{q^m-q^{m-1}}\;
\Delta
+o(1),
\; n\to\infty,
\end{equation}
and
\[R_{net}\leq \frac{r}{r+2}-\frac{2}{r+2}\;\gamma-\frac{2r}{r+2}\;
\frac{q^m-1}{q^m-q^{m-1}}\;
\Delta
+o(1)\; n\to\infty.\]
Similarly, the Plotkin-like bound in \eqref{EA-qLRC Ploktin Bound} implies that
\[2\delta\le
\frac{q^{r-1}(q-1)}{q^r-1}
\left(
n+c-\frac{r+2}{r}\kappa+2(r+1)
\right),\]
which implies that
\begin{equation}\label{APB}
    R\le
\frac{r}{r+2}(1+\gamma)-
\frac{2r}{r+2}\;
\frac{q^r-1}{q^r-q^{r-1}}\;\Delta
+o(1),\; n\to\infty,
\end{equation}

and 

\[    R_{net}\le
\frac{r}{r+2}-\frac{2}{r+2}\;\gamma-
\frac{2r}{r+2}\;
\frac{q^r-1}{q^r-q^{r-1}}\;\Delta
+o(1),\; n\to\infty.\]

 From the inequalities \eqref{ASB}, \eqref{AGB} and \eqref{APB}, we can see that

 \begin{equation}\label{Com_equ}
    1\;\le\; \frac{q^m-1}{q^m-q^{m-1}}\;\le\; \frac{q^r-1}{q^r-q^{r-1}} ,\; \mbox{ for all } \; 1\leq m \leq r.
 \end{equation}

\begin{remark}\label{rem:AGB-APB-coincide}
The inequality~\eqref{Com_equ} shows that the coefficient
$\frac{q^m-1}{q^m-q^{m-1}}$ governing the tightness of the asymptotic Griesmer-like bound
\eqref{AGB} is non-decreasing in $m$, and is maximized at the endpoint $m=r$. At $m=r$, the
$\sum_{i=m}^{r-1}\lceil\delta/q^i\rceil \ge (r-m)$  becomes vacuous, since the summation
range $i=m,\dots,r-1$ is empty. Consequently, setting $m=r$ in \eqref{AGB} yields
\[
    R \le \frac{r}{r+2}(1+\gamma) - \frac{2r}{r+2}\,\frac{q^r-1}{q^r-q^{r-1}}\,\Delta + o(1),
    \qquad n\to\infty,
\]
which coincides exactly with the asymptotic Plotkin-like bound \eqref{APB}. That is, the
asymptotic Griesmer-like bound, optimized over the free parameter $m$, recovers the
asymptotic Plotkin-like bound as its tightest member; for $1\le m<r$, \eqref{AGB} is a
strictly weaker (looser) bound than \eqref{APB}, consistent with~\eqref{Com_equ}. 
\end{remark}

Finally, we turn to the asymptotic form of the sphere-packing-like bound
\eqref{EA-qLRC Sphere Packing Bound}. To this end, we first recall the $q$-ary entropy
function $H_q(x)$, defined as
\[
H_{q}(x)=
\begin{cases}
0, & \text{if}\;x=0\\
x\log_{q}(q-1)-x\log_{q}x-(1-x)\log_{q}(1-x),
& \text{if}\; 0<x\le 1-q^{-1}.
\end{cases}
\]
By \cite[Lemma~2.10.3]{huffman2010fundamentals}, for $0\le \Delta' \le 1-q^{-1}$ and $q\ge 2$,
\begin{equation}\label{asyHB}
    \lim_{M\to\infty}\;M^{-1}\;\log_{q}\bigl[V_{q}(M,\lfloor \Delta' M\rfloor)\bigr]
=H_{q}(\Delta').
\end{equation}
Fix any integer $\tau\ge 0$ with $0\le\tau\le\lfloor(n+\kappa+c-2)/(2(r+1))\rfloor$, and set
$M = \tfrac{n+\kappa+c}{2}-\tau(r+1)$; since $\tau$ is fixed while $n\to\infty$,
\[
    \frac{M}{n} \;\longrightarrow\; \mu:=\frac{1+R+\gamma}{2}.
\]
Writing $t=\lfloor(\delta-1)/2\rfloor$, the ball's relative radius satisfies
$t/M \to \Delta/(2\mu)$, and applying~\eqref{asyHB} with $\Delta'=\Delta/(2\mu)$ gives
\[
    \log_q V_q(M,t) \;=\; \mu n\, H_q\!\left(\frac{\Delta}{2\mu}\right) + o(n),
    \qquad n\to\infty.
\]
Since $2\tau/n\to 0$ for fixed $\tau$, substituting into
\eqref{EA-qLRC Sphere Packing Bound} and dividing by $n$ yields
\begin{equation}\label{asySPB}
    R \;\le\; 1+\gamma-2\mu\,H_q\!\left(\frac{\Delta}{2\mu}\right)+o(1),
    \qquad \mu=\frac{1+R+\gamma}{2},\qquad n\to\infty,
\end{equation}
and correspondingly
\[    R_{net} \;\le\; 1-2\mu\,H_q\!\left(\frac{\Delta}{2\mu}\right)+o(1),
    \qquad \mu=\frac{1+R_{net}+2\gamma}{2}, \qquad n\to\infty.\]

\begin{remark}
Unlike the Singleton, Griesmer, and Plotkin-like asymptotic bounds, \eqref{asySPB} is
implicit in $R$, since $\mu$ itself depends on $R$; for fixed $\Delta,\gamma,q$, the bound
should be read as defining the largest $R$ satisfying the inequality. Note also that
\eqref{asySPB} carries no explicit dependence on the locality parameter $r$: at this order of
approximation, the vanishing contribution of $\tau(r+1)/n\to0$ removes $r$ from the
leading-order term, in contrast to its prominent role in \eqref{ASB} and \eqref{AGB}.
\end{remark}

\begin{remark}\label{rem:sphere-loosest}
Unlike the Griesmer and Plotkin-like asymptotic bounds, which coincide exactly at $m=r$
(Remark~\ref{rem:AGB-APB-coincide}), the sphere-packing-like bound \eqref{asySPB} does not
coincide with any of the other three asymptotic bounds for any choice of parameters. At
$\Delta=0$ it gives $R\le1+\gamma$, strictly weaker than the common value
$\frac{r}{r+2}(1+\gamma)$ attained there by \eqref{ASB}, \eqref{AGB}, and \eqref{APB} for
every finite $r$. In the large-alphabet limit $q\to\infty$, since $H_q(x)\to x$ for fixed
$x$, \eqref{asySPB} reduces to $R\le1+\gamma-\Delta$, which remains strictly weaker than the
large-locality limit $R\le(1+\gamma)-2\Delta$ of \eqref{ASB} as $r\to\infty$. Nonetheless, for
finite $r,q$, the sphere-packing-like bound can be strictly tighter than the other three in
specific parameter regimes, as illustrated numerically below.
\end{remark}

\begin{figure}[h!]
\centering
\includegraphics[width=\linewidth]{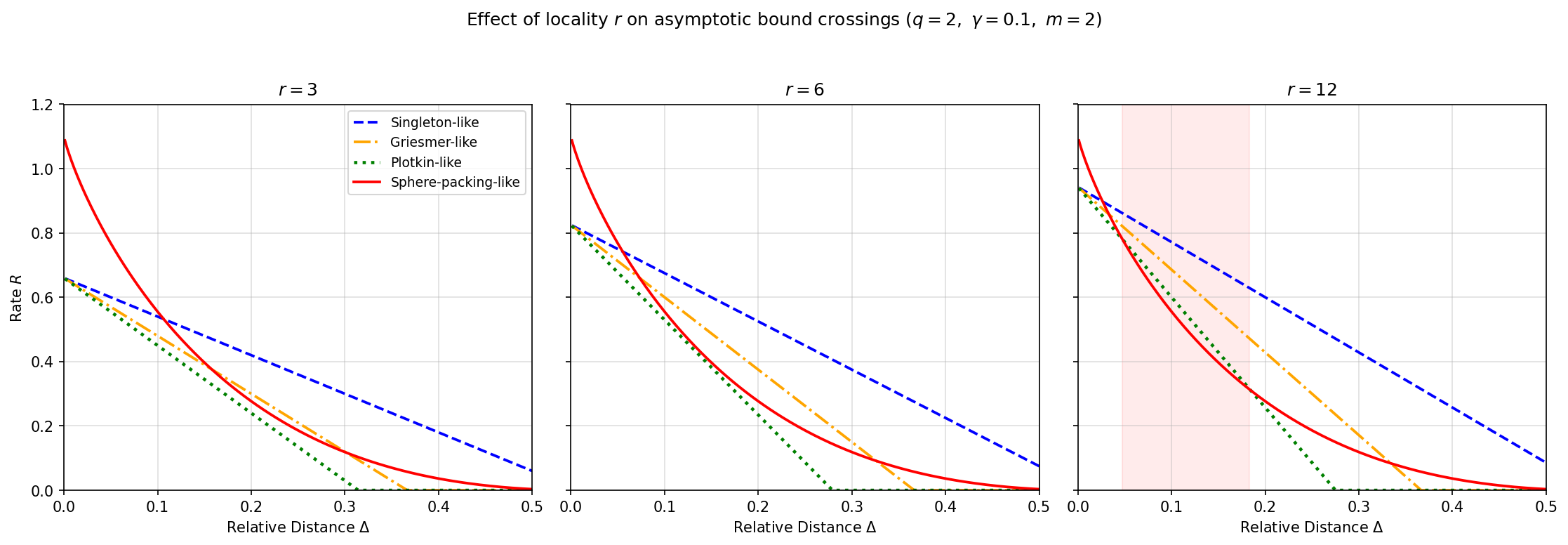}
\caption{Comparison of the four asymptotic bounds at $r=3,6,12$ ($q=2,\gamma=0.1,m=2$). The
sphere-packing-like bound (red) never becomes tightest at $r=3$ or $r=6$; a crossing window
(shaded) first appears at $r=12$, since \eqref{asySPB} carries no explicit $r$-dependence
while the other three bounds' tightness scales with $r$.}
\label{fig:multi-r-comparison}
\end{figure}

\begin{figure}[h!]
\centering
\includegraphics[width=\linewidth]{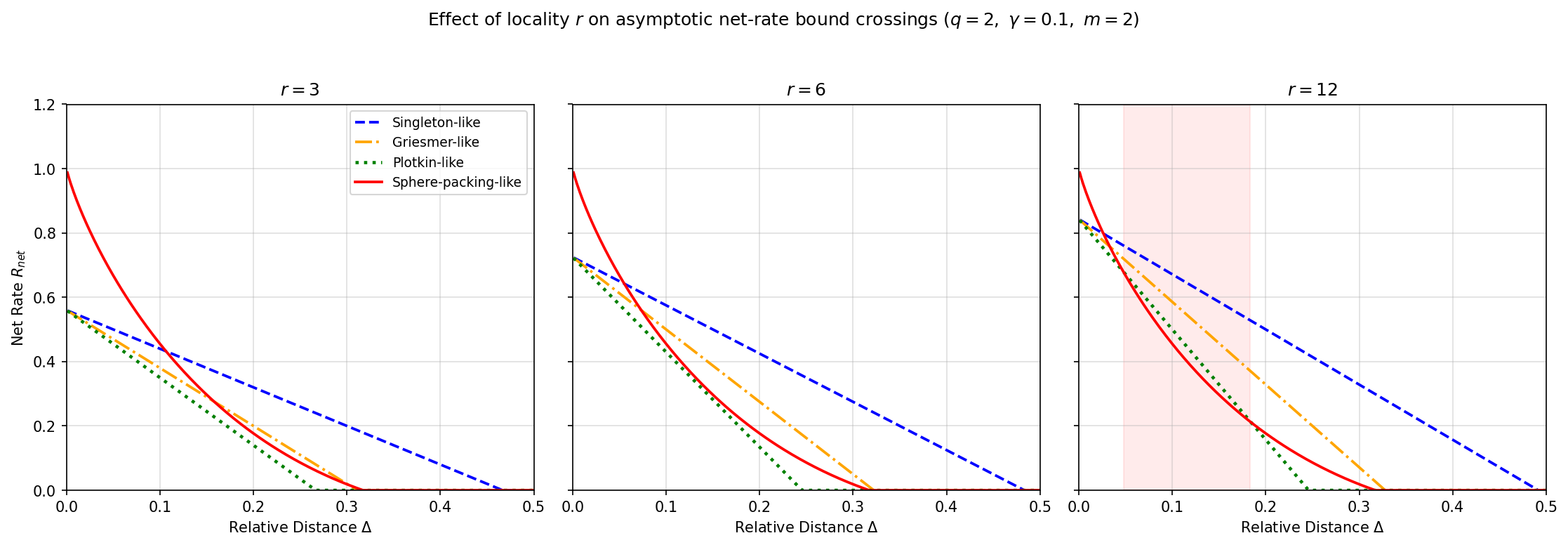}
\caption{Net-rate analogue of Figure~\ref{fig:multi-r-comparison}: the crossing window
$\Delta\in[0.048,0.183]$ at $r=12$ is unchanged from the rate comparison, since
$R_{net}=R-\gamma$ is a fixed vertical shift that does not affect where the bounds cross.}
\label{fig:multi-r-netrate-comparison}
\end{figure}

The comparative analysis reveals a clear hierarchy among the bounds. In the finite-length
regime, the Griesmer-like and Plotkin-like bounds provide the tightest restriction on the
achievable dimension $\kappa$ for a given code length $n$, outperforming the Singleton-like
bound across all tested parameters. In the asymptotic regime, the situation is precisely
characterized by \eqref{Com_equ}, which implies that the asymptotic Plotkin-like bound
\eqref{APB} is strictly tighter than the asymptotic Griesmer-like bound \eqref{AGB}, which is
in turn strictly tighter than the asymptotic Singleton-like bound \eqref{ASB}, for any fixed
locality $r\ge3$. To evaluate the sphere-packing-like bound's tightness relative to the
other three for infinite code lengths, numerical simulations were performed at $r=3,6,12$
($q=2,\gamma=0.1,m=2$); the corresponding results are illustrated in
Figures~\ref{fig:multi-r-comparison} and~\ref{fig:multi-r-netrate-comparison}. No crossing
occurs at $r=3$ or $r=6$ --- the sphere-packing-like bound remains the loosest throughout ---
while at $r=12$ it becomes strictly tighter than the other three for
$\Delta\in[0.048,0.183]$, consistent with its absence of explicit $r$-dependence against the
other bounds' $r$-scaling advantage.

\begin{table}[!ht]
\caption{Summary of Tightness of Bounds\label{tab:table1}}
\centering
\renewcommand{\arraystretch}{1.35}
\begin{tabular}{|c|c|c|c|}
 \hline
       \textbf{Regime}  & \textbf{Tightest Bound} & \textbf{Numerical Comparison} & \textbf{Reference} \\
        \hline
        Large $q$, small $\delta$ & Singleton-like Bound &   $q=64,\delta=3,r=3\ \&\ c=1$ & Figure~\ref{T_Singleton Bound} \\
        \hline
        Small $q$, small $\kappa$ &  Griesmer-like Bound &   $q=2,\delta=5,r=3\ \&\ c=1$ & Figure~\ref{T_Griesmer} \\
        \hline
       Large $\delta\;(\propto n)$ & Plotkin-like Bound &   $q=2,\delta=22,r=3\ \&\ c=1$ & Figure~\ref{T_Ploktin Bound}\\
        \hline
        Finite length &  Griesmer-like/Plotkin-like Bounds  &   $q=2,\delta=7,r=3\ \&\ c=1$ & Figure~\ref{Finite All Bound}\\
        \hline
       Moderate $\Delta$, large $r$ & Sphere-Packing-like Bound & $r=12$ (vs.\ $r=3,6$: never); $\gamma=0.1,q=2\ \&\ m=2$ & Figures~\ref{fig:multi-r-comparison}, \ref{fig:multi-r-netrate-comparison} \\ \hline
\end{tabular}
\end{table}

The Singleton-like bound~\eqref{EA-qLRC_SB} is particularly significant from a structural standpoint; it is entirely alphabet-independent and depends only on the code parameters $n$, $\kappa$, $\delta$, $c$, and $r$. It therefore provides a universal benchmark against which any CSS-like EA-qLRC (in the sense of Proposition~\ref{CSS-EA-qLRC}) must be measured, irrespective of the underlying field size $q$. Codes that meet this bound with equality would be analogues of maximum-distance separable codes in the CSS-like EA-qLRC, and the existence and explicit construction of such optimal CSS-like EA-qLRCs are discussed in a later section.

\section{Optimality Conditions for Pure CSS-like EA-qLRCs}\label{Conditions EA-qLRC}
Having established the Singleton-type bound in~\eqref{EA-qLRC_SB} for 
CSS-like EA-qLRCs, a natural question  is: which pure CSS-like EA-qLRCs actually achieve this bound, and what structural constraints must the underlying classical codes satisfy for this to occur?  This section addresses both questions precisely.

We approach the problem in two steps. We first show, via Lemma~\ref{lemma_1}, 
that a necessary condition for optimality is that the two classical codes 
share the same minimum distance. This rules out 
a broad class of asymmetric constructions from attaining the bound. We 
then derive, in Theorem~\ref{pure_EA-qLRC}, a necessary and sufficient 
condition for a pure CSS-like EA-qLRC to be optimal, characterizing both the 
individual optimality of the constituent classical codes and an explicit 
compatibility condition on their dimensions relative to the locality parameter 
$r$. Together, these results give a complete algebraic characterization of 
optimal pure CSS-like EA-qLRCs and provide a concrete criterion that can be verified 
directly from the parameters of any candidate construction.

\begin{lemma}\label{lemma_1}
Suppose  $\mathcal{Q}_E$ be a pure CSS-like  EA-qLRC with locality $r$, obtained by Proposition~\ref{CSS-EA-qLRC} using two classical $[n,k_i,d_i]_q$ codes $\mathcal{C}_i$ of locality $r$ that satisfy $\delta = \min\{d_1, d_2\}$. If $d_1 \neq d_2$, then $\mathcal{Q}_E$ cannot attain bound~\eqref{EA-qLRC_SB}.
\end{lemma}

\begin{proof}

Without loss of generality, assume $d_1 < d_2$, and let $\kappa = k_1 + k_2 -  n + c$ denote the dimension of $\mathcal{Q}_E$. Since $c = \mathrm{rank}(H_1 H_2^T) \le \min\{n - k_1, n - k_2\}$, we have $k_1 - \kappa = n - k_2 - c \ge 0$ and $k_2 - \kappa = n - k_1 - c \ge 0$, which implies $k_i \ge \kappa$ for $i \in \{1,2\}$.
 
Since both $\mathcal{C}_1$ and  $\mathcal{C}_2$ are LRCs with locality $r$, so the  Singleton-type bound in~\eqref{cLRC Singleton Bound} gives $ d_1 \le n - k_1 - \left\lceil\frac{k_1}{r}\right\rceil + 2,\;\text{and}\; d_2 \le n - k_2 - \left\lceil\frac{k_2}{r}\right\rceil + 2.$ Adding these two inequalities and using $d_1 < d_2$ yields $2d_1 < d_1 + d_2 \le 2n - (k_1+k_2) - \left\lceil\frac{k_1}{r}\right\rceil  - \left\lceil\frac{k_2}{r}\right\rceil + 4$. This further implies $ 2d_1 < n - \kappa + c 
           - \left\lceil\frac{k_1}{r}\right\rceil 
           - \left\lceil\frac{k_2}{r}\right\rceil + 4$ as $\kappa = k_1 + k_2 - 
n + c$.  Because $k_i \ge \kappa$ for $i \in \{1,2\}$, $\lceil k_i/r\rceil \ge \lceil\kappa/r\rceil$ is given, $\left\lceil\frac{k_1}{r}\right\rceil + \left\lceil\frac{k_2}{r}\right\rceil
    \ge 2\left\lceil\frac{\kappa}{r}\right\rceil$. Therefore,
$2d_1 < n - \kappa + c  - \left\lceil\frac{k_1}{r}\right\rceil - \left\lceil\frac{k_2}{r}\right\rceil + 4\;\le\; n - \kappa + c - 2\left\lceil\frac{\kappa}{r}\right\rceil + 4$. Hence $\mathcal{Q}_E$ cannot achieve the bound 
in~\eqref{EA-qLRC_SB}.
\end{proof}
Lemma~\ref{lemma_1} shows that equal minimum distances in the constituent 
classical codes are a prerequisite for optimality. We now show this condition is also the gateway to a complete 
characterization: the following theorem gives necessary and sufficient 
conditions on the constituent classical codes $\mathcal{C}_1$ and 
$\mathcal{C}_2$ for the resulting pure CSS-like EA-qLRC to achieve 
the bound~\eqref{EA-qLRC_SB}.

\begin{theorem}\label{pure_EA-qLRC}
Let $\mathcal{Q}_E$ be a {\qlb}$n,\kappa=k_1+k_2-n+c,\delta;c${\qrb}$_q$  pure CSS-like EA-qLRC with locality $r$ constructed from
Proposition \ref{CSS-EA-qLRC} using two classical $[n,k_i,d_i]_q$ codes $\mathcal{C}_i$ of locality $r$, for $i=1,2$. Then  $\mathcal{Q}_E$ is optimal if and only if $\mathcal{C}_i$ are optimal LRCs such that $k_1 = k_2$, $d_1=d_2$ and 
          $\left\lceil\dfrac{k_1}{r}\right\rceil 
           = \left\lceil\dfrac{\kappa}{r}\right\rceil$.
\end{theorem}

\begin{proof} $(\Leftarrow)$ Suppose $\mathcal{C}_1$ and $\mathcal{C}_2$ are optimal LRCs with $k_1 = k_2 = k$, $d_1 = d_2 = d$, and $\lceil k/r \rceil = \lceil \kappa/r \rceil$. From~\eqref{cLRC Singleton Bound}, we have $d_1 = n - k_1 - \lceil k_1/r \rceil + 2$ and $d_2= n - k_2 - \lceil k_2/r \rceil + 2$. Thus,
\[
2d =d_1+d_2= 2n - 2k - 2\left\lceil\frac{k}{r}\right\rceil + 4 = n - \kappa + c - 2\left\lceil\frac{\kappa}{r}\right\rceil + 4,
\] as $\kappa = k_1+k_2 - n + c$. Again, since $\delta=\min\{d_1, d_2\}=d$, $\mathcal{Q}_E$ attains bound~\eqref{EA-qLRC_SB} with equality, hence $\mathcal{Q}_E$ is optimal.

$(\Rightarrow)$ Suppose that  $\mathcal{Q}_E$ is an 
{\qlb}$n, \kappa, d;c${\qrb}$_q$ {\it{optimal}} pure CSS-like EA-qLRC with locality $r$. Then, by Lemma~\ref{lemma_1},  $\delta = d=d_1=d_2$, and it follows from \eqref{EA-qLRC_SB} that
\begin{equation}\label{equ_13}
    2d = n - \kappa + c - 2\left\lceil\frac{\kappa}{r}\right\rceil + 4.
\end{equation}
Since $\mathcal{C}_1$ and $\mathcal{C}_2$ have locality $r$, the  bound in~\eqref{cLRC Singleton Bound} gives $    d_1 \le n - k_1 - \left\lceil\frac{k_1}{r}\right\rceil + 2$ and $d_2 \le n - k_2 - \left\lceil\frac{k_2}{r}\right\rceil + 2$. 
Adding these two inequalities, we get $2d \leq d_1+d_2 \le n - \kappa + c - \left\lceil\frac{k_1}{r}\right\rceil - \left\lceil\frac{k_2}{r}\right\rceil + 4$. This further implies from \eqref{equ_13} that $\left\lceil\frac{k_1}{r}\right\rceil + \left\lceil\frac{k_2}{r}\right\rceil\le 2\left\lceil\frac{\kappa}{r}\right\rceil$. Since $k_i \ge \kappa$, so $\left\lceil\frac{k_1}{r}\right\rceil + \left\lceil\frac{k_2}{r}\right\rceil
    \ge 2\left\lceil\frac{\kappa}{r}\right\rceil$. Thus $ \left\lceil\frac{k_1}{r}\right\rceil + \left\lceil\frac{k_2}{r}\right\rceil
    = 2\left\lceil\frac{\kappa}{r}\right\rceil$.
 This in turn implies from~\eqref{equ_13}, and $\kappa=k_1+k_2-n+c$ that  $ 2d 
    = \Bigl(n - k_1 - \left\lceil\frac{k_1}{r}\right\rceil + 2\Bigr)
    + \Bigl(n - k_2 - \left\lceil\frac{k_2}{r}\right\rceil + 2\Bigr)$. Therefore $d = n - k_1 - \left\lceil\frac{k_1}{r}\right\rceil + 2
      = n - k_2 - \left\lceil\frac{k_2}{r}\right\rceil + 2$, i.e., $\mathcal{C}_1$ and $\mathcal{C}_2$ attain the bound \eqref{cLRC Singleton Bound}. Consequently, $k_1 + \left\lceil\frac{k_1}{r}\right\rceil  = k_2 + \left\lceil\frac{k_2}{r}\right\rceil$ forces $k_1 = k_2$. Hence $\lceil k_1/r\rceil =  \lceil\kappa/r\rceil$. 
\end{proof}

\begin{theorem}\label{cor:single-code-optimality}
Let $\mathcal{Q}_E$ be a {\qlb}$n,\kappa=2k-n+c,\delta;c${\qrb}$_q$  pure CSS-like EA-qLRC with locality $r$ constructed from Proposition \ref{CSS-EA-qLRC} using a LRC $\mathcal{C}$ of locality $r$ such that 
$s = \dim(\mathcal{C} \cap \mathcal{C}^{\perp})$. Then  $\mathcal{Q}_E$ is optimal if and only if $\mathcal{C}$ is optimal such that $s \;\le\; (k-1) \bmod r$.
\end{theorem}
\begin{proof}
Since $\mathcal{C}_1=\mathcal{C}_2=\mathcal{C}$ is an optimal LRC, Theorem~\ref{PC-EAQEC} yields $c=n-k-s,$ where $s=\dim(\mathcal{C}\cap\mathcal{C}^{\perp})$. Consequently, $\kappa=2k-n+c=k-s.$ By Theorem~\ref{pure_EA-qLRC}, $\mathcal{Q}_E$ is optimal if and only if $\left\lceil\frac{k}{r}\right\rceil=\left\lceil\frac{k-s}{r}\right\rceil$. 
Write $k=ar+b$, where $a=\lfloor k/r\rfloor$ and $0\le b<r$. If $b=0$, then $\lceil k/r\rceil=a$, while $k-s=(a-1)r+(r-s)$. Hence, $\lceil(k-s)/r\rceil=a$ precisely when $0\le s\le r-1$, whereas $\lceil(k-s)/r\rceil\le a-1$ whenever $s\ge r$. Thus, $\left\lceil\frac{k}{r}\right\rceil=\left\lceil\frac{k-s}{r}\right\rceil\iff s\le r-1.$
Since $k=ar$, $k-1=(a-1)r+(r-1) \geq (a-1)r + s$. This implies that $s\le(k-1)\bmod r$. In the other case, when $1\le b\le r-1$, then $\lceil k/r\rceil=a+1$ and $k-s=ar+(b-s)$. It follows that $\lceil(k-s)/r\rceil=a+1$ if and only if $s<b$ (as $\lceil(k-s)/r\rceil\le a$ whenever $s\ge b$). Therefore, $\left\lceil\frac{k}{r}\right\rceil=\left\lceil\frac{k-s}{r}\right\rceil \iff s\le b-1$. Since $k=ar+b$, $k-1=ar+(b-1) \geq ar + s$, and $\left\lceil\frac{k}{r}\right\rceil=\left\lceil\frac{k-s}{r}\right\rceil \iff s\le(k-1)\bmod r$. Therefore, in both cases, $\left\lceil\frac{k}{r}\right\rceil=\left\lceil\frac{k-s}{r}\right\rceil\iff s\le(k-1)\bmod r$. Hence, proved.
\end{proof}

\begin{remark}\label{rem:why-s}
Theorem~\ref{cor:single-code-optimality} phrases optimality entirely through the hull dimension $s=\dim(\mathcal{C}\cap\mathcal{C}^\perp)$ rather than through $c$ or $\operatorname{rank}(HH^\top)$ directly. {\emph{$s$ isolates the purely classical obstruction from the ambient parameters.}} The quantity $c$ is not intrinsic to $\mathcal{C}$ alone: it depends additively on $n$ and $k$ through $c=n-k-s$. Stating the optimality threshold via $c$ would require re-deriving $n-k$ at every use. Working with $s$ instead expresses the optimality condition $s\le(k-1)\bmod r$ as a comparison between two quantities intrinsic to $\mathcal{C}$ and $r$ alone --- how large the hull is versus how much slack the locality parameter tolerates --- with no reference to the ambient length.
\end{remark}

In the following theorem, we derive a necessary and sufficient condition for an optimal CSS-like EA-qLRC to be maximally entangled.

\begin{theorem}\label{EA-qLRC LCD}
Let $\mathcal{Q}_E$ be a {\qlb}$n,\kappa=2k-n+c,\delta;c${\qrb}$_q$  pure CSS-like EA-qLRC with locality $r$ constructed from Proposition \ref{CSS-EA-qLRC} using an $[n,k,d]_q$ LRC $\mathcal{C}_1 = \mathcal{C}_2 = \mathcal{C}$. Then $\mathcal{Q}_E$ is an optimal maximally entangled pure CSS-like EA-qLRC {\qlb}$n,k,d; n-k ${\qrb}$_q$ if and only if  $\mathcal{C}$ is an optimal LCD code. 
\end{theorem}
\begin{proof}
From Theorem~\ref{PC-EAQEC}, an EA-QEC  $\mathcal{Q}_E$ obtained by a classical code $\mathcal{C}$ is maximally entangled if and only of the constituent code $\mathcal{C}$ is LCD. Hence, the result from Theorem \ref{cor:single-code-optimality}.
\end{proof}

\section{Explicit Constructions and Their Optimality Boundary}\label{Construction Optimal EA-qLRCs}
In the preceding sections, we established the optimality conditions for pure CSS-like EA-qLRCs and characterized when they attain the Singleton-type bound~\eqref{EA-qLRC_SB}. We now address the explicit construction of such optimal codes. By Theorem~\ref{cor:single-code-optimality}, an optimal $[n,k,d]_q$ classical LRC with locality $r$ yields an optimal pure CSS-like EA-qLRC whenever $s = \dim(\mathcal{C} \cap \mathcal{C}^{\perp}) \le (k-1) \bmod r$. In particular, every optimal LCD LRC yields an optimal maximally pure CSS-like EA-qLRC. Based on this criterion, we present two families of CSS-like EA-qLRCs: \emph{Tamo--Barg codes}, which are Reed--Solomon-like polynomial evaluation codes achieving bound~\eqref{cLRC Singleton Bound}, and \emph{cyclic LRCs}, whose cyclic structure enables simultaneous control of the minimum distance, locality, and dual intersection dimension $s$.

\subsection{Construction from Tamo--Barg Code}
Let $\mathbb{F}_q$ be a finite field of characteristic $p$, and let $r \ge 2$ be an integer such that $r+1$ divides $q-1$. Let $\ell$ be an integer satisfying $1 \le \ell < q-1$, and let $\mathcal{P} = \mathbb{F}_q^{*} = \{1, \alpha, \alpha^2, \ldots, \alpha^{q-2}\}$, where $\alpha$ is a primitive element of $\mathbb{F}_q$. The evaluation set is the multiplicative group $\mathbb{F}_q^*$, giving code length $n = |\mathcal{P}| = q-1$. Since $(r+1)\mid(q-1)=n$, the subgroup $H\le\mathbb{F}_q^*$ of order $r+1$ partitions $\mathcal{P}$ into $n/(r+1)$ cosets of size $r+1$. The Tamo--Barg locally recoverable code is defined by
\begin{equation}\label{eq:Tamo-Barg}
C_{\mathrm{TB}} = \left\{ \operatorname{ev}(f) := \bigl(f(1), f(\alpha), \ldots, f(\alpha^{q-2})\bigr) : f(X) = \sum_{i \in S} a_i X^i, \; a_i \in \mathbb{F}_q \right\},
\end{equation}
where the exponent set is $S = \{i \in [\ell-1]^\dagger : i \not\equiv r \pmod{r+1}\}$. Writing any exponent $i=j(r+1)+i_0$ with $i_0\in [r-1]^\dagger$, the factor $X^{j(r+1)}$ is constant on each coset of $H$, so the restriction of $\operatorname{ev}(f)$ to a coset agrees with a polynomial of degree at most $r-1$ in that coset's $r+1$ points; hence any one symbol in a coset is determined by the remaining $r$ symbols of that same coset. The code $C_{\mathrm{TB}}$ defined in \eqref{eq:Tamo-Barg} has dimension $k = |S| = \ell - \left\lfloor\frac{\ell}{r+1}\right\rfloor$, minimum distance $d \ge q - \ell$, and is locally recoverable with locality $r$ in the sense of Definition~\ref{def:locality}, each recovery set $\Gamma_i$ being the $r$ other elements of the coset of $H$ containing coordinate $i$.

To determine the dual code of $C_{\mathrm{TB}}$, let $\mathbf{v}_i = \operatorname{ev}(X^i)$ for $0 \le i \le n-1$ be the evaluation vector of $X^i$ at $\mathcal{P}$. The Euclidean inner product of two monomial evaluation vectors is
\[
\langle\mathbf{v}_i, \mathbf{v}_j\rangle = \sum_{x \in \mathbb{F}_q^{*}} x^{i+j} = \sum_{m=0}^{q-2} (\alpha^{i+j})^m = 
\begin{cases}
q-1, & i+j \equiv 0 \pmod n, \\
0, & i+j \not\equiv 0 \pmod n.
\end{cases}
\]
Since $q-1 \neq 0$ in $\mathbb{F}_q$, $\mathbf{v}_i$ and $\mathbf{v}_j$ are orthogonal if and only if $i+j \not\equiv 0 \pmod n$. Consequently, a vector $\mathbf{u} = \sum_{j=0}^{n-1} a_j \mathbf{v}_j$ belongs to $C_{\mathrm{TB}}^\perp$ if and only if $\langle\mathbf{u}, \mathbf{v}_i\rangle = 0$ for all $i \in S$, or equivalently, $j \in S^\perp \iff n-j \notin S$, where $S^\perp = \{j \in [n-1] : n-j \notin S\}$.

We determine $S^\perp$ explicitly. First, for $j=0$, $0 \in S \implies \langle\mathbf{v}_0, \mathbf{v}_0\rangle \neq 0$, so $0 \notin S^\perp$. Next, for $1 \le j \le n-\ell$, we have $n-j \ge \ell$, so $n-j \notin S$, which implies $j \in S^\perp$. Finally, for $n-\ell+1 \le j \le n-1$, we have $1 \le n-j \le \ell-1$. Hence, $n-j \notin S \iff n-j \equiv r \pmod{r+1}$. Since $(r+1) \mid n$, this simplifies to $j \equiv 1 \pmod{r+1}$. Thus,
\[
S^\perp = \{i \in [ q-2] : i \le q-\ell-1 \text{ or } i \equiv 1 \pmod{r+1}\}.
\]
The dual Tamo--Barg code is therefore given by
\[
C_{\mathrm{TB}}^\perp = \left\{ \operatorname{ev}(f) = \bigl(f(1), f(\alpha), \ldots, f(\alpha^{q-2})\bigr) : f(X) = \sum_{j \in S^\perp} a_j X^j, \; a_j \in \mathbb{F}_q \right\}.
\]

\begin{proposition}\label{prop:TBintersection}
Let $q$ be a prime power, $n = q-1$, and $r \ge 2$ an integer with $r+1$ dividing $n$. Let $1 \le \ell < n$, and let $S = \{i \in [ \ell-1]^\dagger : i \not\equiv r \pmod{r+1}\}$ and $S^\perp = \{i \in [ n-1] : i \le q-\ell-1 \text{ or } i \equiv 1 \pmod{r+1}\}$ be the exponent sets of $C_{\mathrm{TB}}$ and $C_{\mathrm{TB}}^\perp$, respectively. Then,
\[
S \cap S^\perp = 
\begin{cases} 
\{i \in [ \ell-1] : i \not\equiv r \pmod{r+1}\}, & \ell \le \frac{q}{2}, \\[2mm] 
\{i \in [ q-1-\ell] : i \not\equiv r \pmod{r+1}\} \cup \{i \in [ \ell-1] :i\ge q-\ell,\; i \equiv 1 \pmod{r+1}\}, & \ell > \frac{q}{2}. 
\end{cases}
\]
Moreover, $s := |S \cap S^\perp|$ is given by
\[
s = 
\begin{cases}
\ell - 1 - \left\lfloor\dfrac{\ell}{r+1}\right\rfloor, & \ell \le \dfrac{q}{2}, \\[3mm]
(q - 1 - \ell) - \left\lfloor\dfrac{q-\ell}{r+1}\right\rfloor + \left\lfloor\dfrac{\ell-2}{r+1}\right\rfloor - \left\lfloor\dfrac{q-\ell-2}{r+1}\right\rfloor, & \ell > \dfrac{q}{2}.
\end{cases}
\]
\end{proposition}

\begin{proof}
We decompose $S^\perp = S^\perp_1 \cup S^\perp_2$, where $S^\perp_1 = \{i \in [ n-1] : i \le q-\ell-1\}$ and $S^\perp_2 = \{i \in [n-1] : i \ge q-\ell, \; i \equiv 1 \pmod{r+1}\}$. Thus, $S \cap S^\perp = (S \cap S^\perp_1) \cup (S \cap S^\perp_2)$. 

Observe that $S \cap S^\perp_1 = \{i \in [m] : i \not\equiv r \pmod{r+1}\}$, where $m = \min\{\ell-1, q-\ell-1\}$. Among the integers $1, \ldots, m$, exactly $\left\lfloor\frac{m+1}{r+1}\right\rfloor$ are congruent to $r \pmod{r+1}$. Therefore, $|S \cap S^\perp_1| = m - \left\lfloor\frac{m+1}{r+1}\right\rfloor$. 

Next, $S \cap S^\perp_2 = \{i \in [\ell-1] : i\ge q-\ell,\;i \equiv 1 \pmod{r+1}\}$. The number of integers satisfying $n-\ell+1 \le i \le \ell-1$ with $i \equiv 1 \pmod{r+1}$ is $\left\lfloor\frac{\ell-2}{r+1}\right\rfloor - \left\lfloor\frac{n-\ell-1}{r+1}\right\rfloor$.

We consider two cases:
\begin{enumerate}
    \item If $\ell \le q/2$, then $\ell-1 \le q-\ell-1$, so $m = \ell-1$ and $S \cap S^\perp_2 = \varnothing$. Hence $S \cap S^\perp = S \cap S_1^\perp$, yielding $s = (\ell-1) - \left\lfloor\frac{\ell}{r+1}\right\rfloor$.
    \item If $\ell > q/2$, then $q-\ell-1 < \ell-1$, so $m = q-\ell-1$. Hence $S \cap S^\perp_1 = \{i \in [q-\ell-1] : i \not\equiv r \pmod{r+1}\}$ and $S \cap S^\perp_2 = \{i \in [\ell-1] :i\ge q-\ell,\;i \equiv 1 \pmod{r+1}\}$. Adding the cardinalities of these disjoint sets yields the claimed expression for $s$.
\end{enumerate}
\end{proof}

Using Proposition~\ref{CSS-EA-qLRC} and Proposition~\ref{prop:TBintersection}, we construct explicit pure CSS-like EA-qLRCs from Tamo--Barg codes.

\begin{theorem}\label{thm:TB-EAqLRC-1}
Let $q$ be a prime power, $r \ge 2$ an integer with $(r+1) \mid (q-1)$, and $0 < \ell \le \frac{q}{2}$. Then there exists a pure CSS-like EA-qLRC with locality $r$ and parameters {\qlb}$q - 1, \, 1, \, \ge q - \ell, \, r; \; q - 1 - \ell + \left\lfloor \frac{\ell}{r+1} \right\rfloor - s${\qrb}$_q$,
where $s = \ell - 1 - \left\lfloor \frac{\ell}{r+1}\right\rfloor$.
\end{theorem}
\begin{proof}
Follows from Propositions~\ref{CSS-EA-qLRC} and ~\ref{prop:TBintersection}.
\end{proof}
\begin{theorem}\label{thm:TB-EAqLRC-2}
Let $q$ be a prime power, $r \ge 2$ an integer with $(r+1) \mid (q-1)$, and $\frac{q}{2} < \ell < q-1$. Then there exists a pure CSS-like EA-qLRC with locality $r$ and parameters 
{\qlb}$q - 1, \, \ell - \left\lfloor \frac{\ell}{r+1} \right\rfloor - s, \, \ge q - \ell, \, r; \; q - 1 - \ell + \left\lfloor \frac{\ell}{r+1} \right\rfloor - s${\qrb}$_q$,
where $s = (q-1-\ell) - \left\lfloor \dfrac{q-\ell}{r+1}\right\rfloor + \left\lfloor \dfrac{\ell-2}{r+1}\right\rfloor - \left\lfloor \dfrac{q-\ell-2}{r+1}\right\rfloor$.
\end{theorem}
\begin{proof}
Follows from Propositions~\ref{CSS-EA-qLRC} and ~\ref{prop:TBintersection}.
\end{proof}
\begin{remark}
Every code produced by Theorem~\ref{thm:TB-EAqLRC-1} has $\kappa=1$, since
Proposition~\ref{prop:TBintersection} gives $s=k-1$ throughout the range $\ell\le q/2$. The
family therefore encodes a single logical qudit regardless of $\ell$, a degeneracy
independent of, and additional to, the optimality obstruction of
Remark~\ref{rem:TB-non-optimality}.
\end{remark}
\begin{remark}\label{rem:TB-non-optimality}
For the Tamo--Barg construction with $\ell \le q/2$, Proposition~\ref{prop:TBintersection} gives $s = k-1$, so by Theorem~\ref{cor:single-code-optimality} the resulting pure CSS-like EA-qLRC attains the bound in~\eqref{EA-qLRC_SB} if and only if $k-1 \le (k-1) \bmod r$, that is, $k \le r$. This window is degenerate in every respect. First, it forces the code to the boundary $k=r$ of the admissible range
$1\le r\le k$ fixed in Section~\ref{ssec:classical-LRC}, or beyond it: whenever
$k\le r$, every coordinate of an $[n,k,d]_q$ code with $d\ge2$ is recoverable from
an information set of size $k$ avoiding it, so the locality-$r$ property holds
automatically and carries no information about the code. Equivalently
$\lceil k/r\rceil=1$, so \eqref{cLRC Singleton Bound} collapses to the classical
Singleton bound and \eqref{EA-qLRC_SB} to the entanglement-assisted quantum
Singleton bound of Lemma~\ref{lem:EA-singleton}. Second, for $\ell \le r$, no exponent of $[\ell-1]^\dagger$ is congruent to $r \pmod{r+1}$, so no exponent is excluded, $S=[\ell-1]^\dagger$, $k=\ell$, and $C_{\mathrm{TB}}$ is simply a Reed--Solomon code. Third, $\kappa = k-s = 1$, so $\lceil\kappa/r\rceil = 1$ and \eqref{EA-qLRC_SB} reduces to the entanglement-assisted quantum Singleton bound, in which the locality parameter plays no role. Consequently, pure CSS-like EA-qLRCs constructed from Tamo--Barg codes do not attain the Singleton-type bound in~\eqref{EA-qLRC_SB} in any nontrivial regime; this gap motivates the use of cyclic and LCD LRCs in the subsequent sections, where strict optimality is achieved.

\end{remark}

\begin{example}\label{ex:TB}
To illustrate the mechanism concretely, take $q=7$ ($n=6$) and $r=2$, so that $r+1=3$ divides $q-1=6$. Fix the primitive root $\alpha=3$ of $\mathbb{F}_7$; the evaluation points are $P_j = \alpha^j$ for $j=0,\ldots,5$, namely $(P_0,\ldots,P_5) = (1,3,2,6,4,5)$. The order-$3$ subgroup $H=\langle\alpha^2\rangle=\{1,2,4\}$ partitions the coordinates into two cosets, $A=\{0,2,4\}$ (points $\{1,2,4\}$) and $B=\{1,3,5\}$ (points $\{3,6,5\}$), each of size $r+1=3$. A direct computation confirms that every codeword $(c_0,\ldots,c_5)$ of $C_{\mathrm{TB}}$ satisfies the fixed local parity-check relations
\[
c_0 + 2c_2 + 4c_4 \equiv 0, \qquad c_1 + 2c_3 + 4c_5 \equiv 0 \pmod 7,
\]
one for each coset; each is a length-$3$, dimension-$2$ single-parity-check relation, so any one erased coordinate in a coset is recovered from the other $r=2$ survivors of that same coset. For instance, taking $f(X)=1+2X+3X^3$ (so $S=\{0,1,3\}$, the $\ell=4$ case below) gives the codeword $(6,4,1,3,5,1)$; erasing coordinate $2$ (true value $1$) and solving $c_0+2c_2+4c_4\equiv0$ for $c_2$ from the surviving $c_0=6$, $c_4=5$ recovers $c_2 = -(1\cdot 6+4\cdot 5)\cdot 2^{-1} \equiv 1 \pmod 7$, correctly.

First, consider $\ell=3\le q/2$. Then $S=\{0,1\}$, so $C_{\mathrm{TB}}$ is the $[6,2]_7$ Reed--Solomon code: it is MDS with exact minimum distance $d=5$, exceeding the generic lower bound $q-\ell=4$. Direct computation gives $\dim(C_{\mathrm{TB}}\cap C_{\mathrm{TB}}^\perp)=s=1$, matching Proposition~\ref{prop:TBintersection}, and the exact quantum distance is $\delta=\mathrm{wt}(C_{\mathrm{TB}}\setminus(C_{\mathrm{TB}}\cap C_{\mathrm{TB}}^\perp))=5$. The resulting pure CSS-like EA-qLRC has exact parameters {\qlb}$6,1,5,2;3${\qrb}$_7$. Since $k=2=r$ here, this instance sits exactly on the boundary $k\le r$ of Remark~\ref{rem:TB-non-optimality}: indeed $2\delta=10=n-\kappa+c-2\lceil\kappa/r\rceil+4$, so the Singleton-like bound~\eqref{EA-qLRC_SB} is attained --- but, as the remark explains, only because the locality constraint is vacuous once $k\le r$. 

Next, choose $\ell=4>q/2$. Then $S=\{0,1,3\}$, giving $k=3>r$, so the locality-$2$ constraint is now genuinely nontrivial. Direct computation gives exact minimum distance $d=3$ (matching the bound $q-\ell=3$ with equality), dual intersection dimension $s=1$, and exact quantum distance $\delta=3$. The resulting pure CSS-like EA-qLRC has exact parameters {\qlb}$6,2,3,2;2${\qrb}$_7$. Here $2\delta=6$ while the Singleton-like bound gives $n-\kappa+c-2\lceil\kappa/r\rceil+4=8$ --- a strict gap of $2$, confirming Remark~\ref{rem:TB-non-optimality}'s conclusion that the construction fails to attain optimality once $k>r$. 
\end{example}

\subsection{Construction from Cyclic Codes}
A linear code $\mathcal{C}$ of length $n$ over $\mathbb{F}_q$ is \emph{cyclic} if $(c_0, \ldots, c_{n-1}) \in \mathcal{C} \implies (c_{n-1}, c_0, \ldots, c_{n-2}) \in \mathcal{C}$. Via the ring isomorphism $\mathbb{F}_q^n \cong R_n := \mathbb{F}_q[x]/\langle x^n-1\rangle$, cyclic codes correspond to ideals of $R_n$. Since $R_n$ is a principal ideal ring, each cyclic code is generated by a unique monic divisor $g(x)$ of $x^n-1$, called the \emph{generator polynomial} of $\mathcal{C}$, giving dimension $k = n - \deg g$. The \emph{parity polynomial} is $h(x) = (x^n-1)/g(x)$, and the dual code $\mathcal{C}^\perp$ is cyclic with generator $h^*(x) = x^{\deg h} h(x^{-1})$.

Assume $\gcd(n,q) = 1$, let $m = \mathrm{ord}_n(q)$, and let $\alpha \in \mathbb{F}_{q^m}$ be a primitive $n$-th root of unity. The \emph{$q$-cyclotomic coset} of $t$ modulo $n$ is $C_t = \{t \cdot q^j \bmod n : j \in [m_t-1]^\dagger\}$, where $m_t = |C_t|$ divides $m$. The minimal polynomial of $\alpha^t$ over $\mathbb{F}_q$ factors as $M_t(x) = \prod_{i \in C_t} (x - \alpha^i)$. The \emph{defining set} of $\mathcal{C}$ with respect to $\alpha$ is $Z = \{i \in [n-1]^\dagger : g(\alpha^i) = 0\}$, so $g(x) = \mathrm{lcm}\{M_t(x) : t \in Z\}$.

\begin{lemma}[BCH Bound~\cite{bose1960class, hocquenghem1959codes}]\label{BCH Bound}
Let $\mathcal{C}$ be a cyclic code of length $n$ over $\mathbb{F}_q$ with defining set $Z$. If $Z$ contains $\delta-1$ consecutive integers, i.e., $\{b, b+1, \ldots, b+\delta-2\} \subseteq Z$ for some $b \ge 0$, then $d(\mathcal{C}) \ge \delta$.
\end{lemma}

\begin{lemma}\label{Hull dim}
Let $\mathcal{C}_1$ and $\mathcal{C}_2$ be cyclic codes of length $n$ over $\mathbb{F}_q$ (respectively, over $\mathbb{F}_{q^2}$ for the Hermitian case) with defining sets $Z_1, Z_2 \subseteq \mathbb{Z}_n$, and dimensions $k_i = n - |Z_i|$ for $i = 1, 2$.
\begin{enumerate}
    \item[(1)] $\dim\left(\mathcal{C}_1^{\perp} \cap \mathcal{C}_2\right) = |Z_1| - |Z_2 \cap (-Z_1)|$, where $-Z_1 := \{-i \bmod n : i \in Z_1\}$.
    \item[(2)] If $\mathcal{C}_1, \mathcal{C}_2$ are cyclic over $\mathbb{F}_{q^2}$, then $\dim\left(\mathcal{C}_1^{\perp_h} \cap \mathcal{C}_2\right) = |Z_1| - |Z_2 \cap (-qZ_1)|$, where $-qZ_1 := \{-qi \bmod n : i \in Z_1\}$.
\end{enumerate}
\end{lemma}

\begin{proof}
We prove (1); the argument for (2) is analogous once the defining set of the Hermitian dual is identified.

Recall that if $\mathcal{C}_1$ is a cyclic code with defining set $Z_1$, then $h_1(x) = (x^n-1)/g_1(x)$ has root set $\mathbb{Z}_n \setminus Z_1$, and $\mathcal{C}_1^\perp$ is generated by $h_1^*(x) = x^{\deg h_1} h_1(x^{-1})$, whose roots are the inverses of those of $h_1$. Thus, $Z(\mathcal{C}_1^\perp) = \mathbb{Z}_n \setminus (-Z_1)$ and $|Z(\mathcal{C}_1^\perp)| = n - |Z_1| = k_1$.

For cyclic codes $\mathcal{D}_1, \mathcal{D}_2$, a codeword lies in $\mathcal{D}_1 \cap \mathcal{D}_2$ if and only if it vanishes on both root sets, so $Z(\mathcal{D}_1 \cap \mathcal{D}_2) = Z(\mathcal{D}_1) \cup Z(\mathcal{D}_2)$~\cite{pereira2022entanglement}. Setting $\mathcal{D}_1 = \mathcal{C}_1^{\perp}$ and $\mathcal{D}_2 = \mathcal{C}_2$, we obtain:
\[
\dim(\mathcal{C}_1^{\perp} \cap \mathcal{C}_2) = n - \bigl|(\mathbb{Z}_n \setminus (-Z_1)) \cup Z_2\bigr|.
\]
By inclusion--exclusion:
\[
\bigl|(\mathbb{Z}_n \setminus (-Z_1)) \cup Z_2\bigr| = (n - |Z_1|) + |Z_2| - \bigl|(\mathbb{Z}_n \setminus (-Z_1)) \cap Z_2\bigr|.
\]
Since $(\mathbb{Z}_n \setminus (-Z_1)) \cap Z_2 = Z_2 \setminus (Z_2 \cap (-Z_1))$, its cardinality equals $|Z_2| - |Z_2 \cap (-Z_1)|$. Substituting this yields $\bigl|(\mathbb{Z}_n \setminus (-Z_1)) \cup Z_2\bigr| = n - |Z_1| + |Z_2 \cap (-Z_1)|$. Therefore,
\[
\dim(\mathcal{C}_1^{\perp} \cap \mathcal{C}_2) = |Z_1| - |Z_2 \cap (-Z_1)|,
\]
establishing (1). Part (2) follows by replacing $-Z_1$ with $-qZ_1$ for the Hermitian conjugate root set.
\end{proof}

Combining Lemma~\ref{Hull dim} with Theorem~\ref{PC-EAQEC} gives the following explicit characterization.

\begin{theorem}\label{Cyclic EAQEC}
Let $\mathcal{C}_i$ be cyclic codes over $\mathbb{F}_q$ with parameters $[n, k_i, d_i]_q$ and defining sets $Z_i \subseteq \mathbb{Z}_n$, respectively, for $i=1,2$. Then there exists a pure EAQEC code with parameters
{\qlb}$n, \; n - |(-Z_1) \cup Z_2|, \; \min\{d_1, d_2\}; \; |(-Z_1) \cap Z_2|${\qrb}$_q$.
\end{theorem}

\begin{proof}
From Lemma~\ref{Hull dim}, $\dim(\mathcal{C}_1^\perp \cap \mathcal{C}_2) = |Z_1| - |Z_2 \cap (-Z_1)|$. Substituting $k_1 = n - |Z_1|$ and $k_2 = n - |Z_2|$ into Theorem~\ref{PC-EAQEC} yields quantum dimension:
\[
\kappa = n - |Z_1| - |Z_2| + |(-Z_1) \cap Z_2| = n - |(-Z_1) \cup Z_2|,
\]
and entanglement amount $c = |(-Z_1) \cap Z_2|$.
\end{proof}

The locality of cyclic codes can be ensured using the structure established by Tamo \emph{et al.}~\cite{tamo2015cyclic}.

\begin{lemma}[\cite{tamo2015cyclic}, Theorem~3.1]\label{cyclicLRC}
Let $\alpha$ be a primitive $n$-th root of unity, $\ell \in [r-1]^\dagger$ an integer, and $b \ge 1$ an integer coprime to $n$. Let $\mu = \frac{k}{r}$. Define:
\[
L = \{\alpha^i \mid i \bmod (r + 1) = \ell\} \quad \text{and} \quad D = \{\alpha^{j+sb} \mid s \in [n - \mu(r + 1)-1]^\dagger\},
\]
where $\alpha^j \in L$. Then the cyclic code with defining set $Z = L \cup D$ is an optimal $[n, k, d]_q$ cyclic LRC with locality $r$.
\end{lemma}

\begin{lemma}\label{lem:support-alignment}
Let $n$ be such that $(r+1)\mid n$, set $\nu=n/(r+1)$, and let $\alpha$ be a primitive
$n$-th root of unity. For $\ell\in[r]^{\dagger}$ let $\mathcal{C}$ be a cyclic code of
length $n$ over $\mathbb{F}_q$ whose defining set contains
$L=\{\,j\in[n-1]^{\dagger}: j\equiv\ell \!\!\pmod{r+1}\,\}$. Then for each
$u\in[\nu-1]^{\dagger}$ the vector $h_u^{(\ell)}\in\mathbb{F}_q^{\,n}$ with entries
\[
\bigl(h_u^{(\ell)}\bigr)_m=
\begin{cases}
\alpha^{\ell m}, & m\equiv u \!\!\pmod{\nu},\\[2pt]
0, & \text{otherwise,}
\end{cases}
\]
lies in $\mathcal{C}^{\perp}$ and satisfies
$\mathrm{supp}\bigl(h_u^{(\ell)}\bigr)=\{\,m\in[n-1]^{\dagger}: m\equiv u\!\!\pmod{\nu}\,\}$,
a set of size $r+1$ that is independent of $\ell$.

Consequently, if $\mathcal{C}_1,\mathcal{C}_2$ are two such codes with parameters
$\ell_1,\ell_2$ respectively, then for every coordinate $i\in[n]$ the checks
$h_{u(i)}^{(\ell_1)}\in\mathcal{C}_1^{\perp}$ and $h_{u(i)}^{(\ell_2)}\in\mathcal{C}_2^{\perp}$,
where $u(i)\equiv i \pmod{\nu}$, satisfy
\[
i\in\mathrm{supp}\bigl(h_{u(i)}^{(\ell_1)}\bigr)\cap\mathrm{supp}\bigl(h_{u(i)}^{(\ell_2)}\bigr),
\qquad
\bigl|\mathrm{supp}\bigl(h_{u(i)}^{(\ell_1)}\bigr)\cup\mathrm{supp}\bigl(h_{u(i)}^{(\ell_2)}\bigr)\bigr|=r+1 .
\]
\end{lemma}

\begin{proof}
Since $L$ is contained in the defining set, every $c=(c_0,\dots,c_{n-1})\in\mathcal{C}$
satisfies $c(\alpha^{\ell+t(r+1)})=0$ for $t\in[\nu-1]^{\dagger}$. Put
$\beta=\alpha^{r+1}$, an element of order $\nu$, and
$S_u=\sum_{m\equiv u\,(\mathrm{mod}\,\nu)}c_m\alpha^{\ell m}$ for $u\in[\nu-1]^{\dagger}$.
Grouping the coordinates by residue modulo $\nu$ gives
\[
0=c\bigl(\alpha^{\ell+t(r+1)}\bigr)=\sum_{m=0}^{n-1}c_m\alpha^{\ell m}\beta^{tm}
=\sum_{u=0}^{\nu-1}\beta^{tu}S_u ,\qquad t\in[\nu-1]^{\dagger},
\]
using $\beta^{tm}=\beta^{tu}$ whenever $m\equiv u\pmod{\nu}$. The coefficient matrix
$(\beta^{tu})_{0\le t,u\le\nu-1}$ is Vandermonde in the distinct powers
$1,\beta,\dots,\beta^{\nu-1}$, hence invertible, so $S_u=0$ for every $u$. As
$S_u=\langle c,h_u^{(\ell)}\rangle$, this says $h_u^{(\ell)}\in\mathcal{C}^{\perp}$.
Each residue class modulo $\nu$ meets $[n-1]^{\dagger}$ in exactly $n/\nu=r+1$
coordinates, and the entries $\alpha^{\ell m}$ are nonzero, so
$\mathrm{supp}(h_u^{(\ell)})$ is that residue class, independent of $\ell$. The final
display follows by taking $u(i)\equiv i\pmod{\nu}$ in each code.
\end{proof}
Applying Lemma~\ref{cyclicLRC} to Proposition~\ref{CSS-EA-qLRC} yields CSS-like EA-qLRCs from cyclic LRC pairs.
\begin{theorem}\label{Cyclic EA-qLRC}
Let $(r+1)\mid n$ and let $\mathcal{C}_i$ be cyclic LRC codes over $\mathbb{F}_q$ with
locality $r$ and defining sets $Z_i=L_i\cup D_i$ for $i\in\{1,2\}$, constructed as in
Lemma~\ref{cyclicLRC} with respect to a common primitive $n$-th root of unity
$\alpha$. Then there exists a pure CSS-like EA-qLRC with locality $r$ and parameters {\qlb}$n,\;n-|(-Z_1)\cup Z_2|,\;\min\{d_1,d_2\};\;|(-Z_1)\cap Z_2|${\qrb}$_q$.
\end{theorem}

\begin{proof}
By construction $L_i\subseteq Z_i$ for $i=1,2$, so Lemma~\ref{lem:support-alignment}
applies to each $\mathcal{C}_i$ and supplies, for every coordinate $i\in[n]$, checks
$c_1^{(i)}\in\mathcal{C}_1^{\perp}$ and $c_2^{(i)}\in\mathcal{C}_2^{\perp}$ whose
supports coincide, contain $i$, and have size $r+1$. The hypotheses of
Proposition~\ref{CSS-EA-qLRC} are therefore satisfied with recovery set
$\Gamma_i=\{\,m: m\equiv i \!\!\pmod{n/(r+1)}\,\}$, and the code has locality $r$.
The stated parameters follow from Theorem~\ref{Cyclic EAQEC}.
\end{proof}


\begin{example}\label{ex:cyclic-EAqLRC}
Let $q=37$, $n=36$, $r=5$, and $b = 1$. We select classical dimensions $k_1 = 20$ and $k_2 = 15$, yielding $\mu_1 = 20/5 = 4$ and $\mu_2 = 15/5 = 3$. For $\mathcal{C}_1$, set $\ell_1 = 0$. Then $L_1 = \{i \equiv 0 \pmod 6\} = \{0, 6, 12, 18, 24, 30\}$. Selecting base exponent $j=0\in L_1$, $s$ ranges over $0\le s\le 36-4(6)-1=11$, giving $D_1=[11]^{\dagger}$. The defining set $Z_1=L_1\cup D_1=\{0,1,\dots,11,12,18,24,30\}$ has $|Z_1|=16$, giving $k_1=36-16=20$; note $12\in L_1$, so $Z_1$ contains the $13$ consecutive integers $0,\dots,12$ and Lemma~\ref{BCH Bound} gives $d_1\ge14$. This forms an optimal $[36, 20, 14]_{37}$ cyclic LRC.

For $\mathcal{C}_2$, set $\ell_2 = 2$. Then $L_2 = \{i \equiv 2 \pmod 6\} = \{2, 8, 14, 20, 26, 32\}$. Selecting $j = 26 \in L_2$, $s$ ranges over $0 \le s \le 17$, giving $D_2 = \{26, 27, \ldots, 35, 0, 1, \ldots, 7\}$. The defining set $Z_2 = L_2 \cup D_2 = \{26, 27, \dots, 35, 0, 1, \dots, 8, 14, 20\}$ has $|Z_2| = 21$, giving $k_2 = 36 - 21 = 15$. Since $Z_2$ contains 19 consecutive roots, $d_2 \ge 20$. This forms an optimal $[36, 15, 20]_{37}$ cyclic LRC.

To construct the CSS-like EA-qLRC via Theorem~\ref{Cyclic EA-qLRC}, we compute $-Z_1 \bmod 36 = \{0, 6, 12, 18, 24, 25, 26, 27, 28, 29, 30, 31, 32, 33, 34, 35\}$. Their intersection is $(-Z_1) \cap Z_2 = \{26, 27, 28, 29, 30, 31, 32, 33, 34, 35, 0, 6\}$, giving entanglement amount $c = |(-Z_1) \cap Z_2| = 12$. By inclusion--exclusion, $|(-Z_1) \cup Z_2| = 16 + 21 - 12 = 25$. The quantum dimension is $\kappa = 36 - 25 = 11$, and minimum distance is $\delta = \min\{14, 20\} = 14$. Consequently, the resulting pure CSS-like EA-qLRC has locality $r = 5$ and parameters {\qlb}$36, 11, 14; 12${\qrb}$_{37}$.

Although $\ell_1=0\ne2=\ell_2$, Lemma~\ref{lem:support-alignment} shows the local
checks of $\mathcal{C}_1$ and $\mathcal{C}_2$ share the same supports, namely the
residue classes modulo $\nu=6$; the recovery sets therefore have size $r+1=6$ and the
locality of the resulting CSS-like EA-qLRC is $r=5$, not $2r$.
\end{example}
\begin{remark}
Substituting $n=36$, $\kappa=11$, $\delta=14$, $c=12$, $r=5$ into the Singleton-like
bound~\eqref{EA-qLRC_SB} gives $2\delta=28$ against a right-hand side of
$n-\kappa+c-2\lceil\kappa/r\rceil+4=35$: a strict gap of $7$. The construction of
Theorem~\ref{Cyclic EA-qLRC} therefore does not attain optimality in this instance,
illustrating that --- unlike the LCD cyclic constructions of
Section~\ref{sec:maximally-entangled-constructions} --- a general cyclic pair need not
satisfy the necessary and sufficient conditions of Theorem~\ref{pure_EA-qLRC} even when both
constituent codes are individually classically optimal; here $\mathcal{C}_1,\mathcal{C}_2$
achieve~\eqref{cLRC Singleton Bound} with $d_1\ne d_2$ ($14\ne20$), which by
Lemma~\ref{lemma_1} already precludes optimality of the resulting CSS-like EA-qLRC.
\end{remark}
\section{Optimal Constructions from Maximally Entangled LCD Codes}\label{sec:maximally-entangled-constructions}

In this section, we establish general existence and optimality results for maximally entangled CSS-like EA-qLRCs, building directly on the LCD-based single-code construction of Theorem~\ref{EA-qLRC LCD} and the optimality criterion of Theorem~\ref{pure_EA-qLRC}. We then exhibit explicit, deterministic families of classical LCD-LRCs --- cyclic codes of length dividing $q-1$ and $q+1$ --- that instantiate these general results concretely.

\begin{lemma}[\cite{yang1994condition}]\label{cyclic_LCD}
Let $\mathcal{C}$ be a cyclic code of length $n$ over $\mathbb{F}_q$ with generator polynomial $g(x)$ and defining set $Z$. Then $\mathcal{C}$ is LCD if and only if $g(x)$ is self-reciprocal, equivalently if and only if $Z=-Z \pmod n$. A sufficient (but not
necessary) condition for this is $q^\ell\equiv-1\pmod n$ for some positive integer $\ell$.
\end{lemma}

\begin{lemma}\label{lem:n-divides-qplus1-LCD}
Let $n\mid(q+1)$. Then every cyclic code of length $n$ over $\mathbb{F}_q$ is LCD.
\end{lemma}
\begin{proof}
Since $q\equiv-1\pmod n$, every nontrivial $q$-cyclotomic coset modulo $n$ has the form $C_i=\{i,-i\}=\{i,n-i\}$ (a two-element set, unless $i\equiv0$ or $2i\equiv0\pmod n$, which are fixed points of the inversion map). Hence every defining set $Z$, being a union of such cosets, automatically satisfies $Z=-Z\pmod n$. By Lemma~\ref{cyclic_LCD}, $\mathcal{C}$ is LCD.
\end{proof}

\begin{remark}\label{rem:boundary}
The smallest instance of Lemma~\ref{lem:n-divides-qplus1-LCD} illustrates the mechanism concretely.
Let $q=4$ and $n=5=q+1$. The $q$-cyclotomic cosets modulo $5$ are $\{0\}$, $\{1,4\}$
and $\{2,3\}$; taking the defining set $Z=\{1,4\}$ gives a cyclic code
$\mathcal{C}=[5,3,3]_4$ with dual $\mathcal{C}^{\perp}=[5,2,4]_4$. Since
$\mathcal{C}^{\perp}$ has minimum weight $4$, the cyclic shifts of a minimum-weight
dual codeword cover every coordinate, so $\mathcal{C}$ has locality $r=3$; and
$\dim(\mathcal{C}\cap\mathcal{C}^{\perp})=0$ automatically, by Lemma~\ref{lem:n-divides-qplus1-LCD}.

Here $k=3=r$, so $\lceil k/r\rceil=1$ and the bound \eqref{cLRC Singleton Bound}
reduces to the classical Singleton bound: the locality constraint is vacuous in the
sense of Section~\ref{ssec:classical-LRC}, and ``optimal LRC'' means simply ``MDS''.
Example~\ref{ex:9-4-5} below crosses this boundary, with $k>r$.
\end{remark}

\begin{theorem}\label{MEAQECC}
Let $\mathcal{C}$ be an LCD cyclic code with parameters $[n,k,d]_q$ and defining set $Z$. Then there exists a pure maximally entangled CSS-like EAQEC code with parameters {\qlb}$n,k,d, n - k${\qrb}$_q$.
\end{theorem}
\begin{proof}
Follows from Lemma \ref{cyclic_LCD} and Theorem  \ref{PC-EAQEC}.
\end{proof}

\begin{example}\label{ex:9-4-5}
Let $q=8$ and $n=9$, so that $n\mid q+1$ and $\gcd(n,q)=1$. Since $8\equiv-1
\pmod 9$, the $q$-cyclotomic cosets modulo $9$ are $\{0\},\{1,8\},\{2,7\},\{3,6\},\{4,5\}$,
and by Lemma~\ref{lem:n-divides-qplus1-LCD} every cyclic code of this length over $\mathbb{F}_8$ is LCD.

Take $r=2$, so that $r+1=3$ divides $n$, and set $\nu=n/(r+1)=3$, $\mu=2$ and
$k=\mu r=4$. Applying Lemma~\ref{cyclicLRC} with $\ell=0$ gives
$L=\{i:i\equiv0\!\!\pmod 3\}=\{0,3,6\}$; choosing the base exponent $j=3\in L$ and
$b=1$, the index $s$ ranges over $0\le s\le n-\mu(r+1)-1=2$, giving $D=\{3,4,5\}$.
The defining set is
\[
Z \;=\; L\cup D \;=\;\{0,3,4,5,6\}\;=\;C_0\cup C_3\cup C_4 ,
\]
a union of $q$-cyclotomic cosets with $|Z|=5$, so $k=9-5=4$.

\emph{Locality.} Since $L\subseteq Z$, every codeword satisfies $c(\alpha^{3t})=0$ for
$t=0,1,2$. Writing $\beta=\alpha^{3}$, an element of order $3$, and
$S_u=\sum_{j\equiv u\,(\mathrm{mod}\,3)}c_j$ for $u=0,1,2$, we have
$c(\alpha^{3t})=\sum_{u=0}^{2}\beta^{tu}S_u$. The matrix $(\beta^{tu})_{0\le t,u\le2}$
is Vandermonde in the distinct powers $1,\beta,\beta^{2}$, hence invertible, so
$S_0=S_1=S_2=0$. Each relation $S_u=0$ is a dual codeword of weight $r+1=3$ supported
on $\{u,u+3,u+6\}$, and the three relations cover all nine coordinates. Thus
$\mathcal{C}$ has locality $r=2$, with recovery sets
$\Gamma_i=\{\,i\bmod 3,\; (i\bmod 3)+3,\; (i\bmod 3)+6\,\}$.

\emph{Distance and classical optimality.} The set $Z$ contains the four consecutive
integers $\{3,4,5,6\}$, so Lemma~\ref{BCH Bound} gives $d\ge5$, while
\eqref{cLRC Singleton Bound} gives $d\le n-k-\lceil k/r\rceil+2=9-4-2+2=5$. Hence $d=5$
and $\mathcal{C}$ is an optimal $[9,4,5]_8$ cyclic LRC with locality $2$.

\emph{The resulting quantum code.} Since $-Z=\{0,6,5,4,3\}=Z$,
Lemma~\ref{cyclic_LCD} confirms $s=\dim(\mathcal{C}\cap\mathcal{C}^{\perp})=0$,
so $\mathcal{C}$ is LCD and the construction is pure, with
$\delta=\mathrm{wt}(\mathcal{C}\setminus\{0\})=5$. As $s=0\le(k-1)\bmod r=1$,
Theorem~\ref{cor:single-code-optimality} applies, and Theorem~\ref{EA-qLRC LCD} yields
the maximally entangled pure CSS-like EA-qLRC {\qlb}$\,9,\,4,\,5;\,5\,${\qrb}$_8$, with locality $r=2$.

Substituting into \eqref{EA-qLRC_SB} gives $2\delta=10$ and
$n-\kappa+c-2\lceil\kappa/r\rceil+4=9-4+5-4+4=10$, so equality holds and the code is
optimal. In contrast to Remark~\ref{rem:boundary}, here $k=4>r=2$ and
$\lceil\kappa/r\rceil=2$, so the locality term contributes strictly and
\eqref{EA-qLRC_SB} does not reduce to the entanglement-assisted quantum
Singleton bound of Lemma~\ref{lem:EA-singleton}. This code is an instance of
Row~5 of Table~\ref{tab:parameters_optimal_qLRCs} ($n\mid q+1$, $q$ even, $\mu$ even). It also
stands in contrast to Example~\ref{ex:cyclic-EAqLRC}, where the constituent cyclic codes
had unequal minimum distances ($d_1=14\ne d_2=20$) and Lemma~\ref{lemma_1}
already precluded optimality; taking $\mathcal{C}_1=\mathcal{C}_2$ satisfies that
necessary condition automatically, and classical optimality of $\mathcal{C}$ supplies
the rest.
\end{example}

\begin{theorem}\label{cor:cyclic-optimal-witnesses}
Every classically optimal cyclic LCD LRC (in particular, the explicit families of
Table~\ref{tab:parameters_optimal_qLRCs} below) yields, via Theorem~\ref{EA-qLRC LCD},
an optimal maximally entangled pure CSS-like EA-qLRC of the same locality.
\end{theorem}
\begin{remark}
When $\mathcal{C}_1=\mathcal{C}_2=\mathcal{C}$ is LCD, $\mathcal{C}\cap\mathcal{C}^{\perp}=\{0\}$
and Theorem~\ref{PC-EAQEC} gives
$\delta=\mathrm{wt}(\mathcal{C}\setminus\{0\})=d$ exactly. Every maximally entangled
construction of this section is therefore pure without further verification.
\end{remark}
\subsection{Construction from Cyclic Codes of Length \texorpdfstring{$n\mid(q-1)$}{n | (q-1)}} Rajput \textit{et al.}~\cite{rajput2020cyclic} present four families of cyclic LRC-LCD codes
of length $n\mid(q-1)$ over $\mathbb{F}_q$, covering $q=2^m$ (even characteristic) and
general $q$ (odd characteristic) with varying locality $r$; in each case the defining set $Z$
is chosen to satisfy both the local-repair condition and $Z=-Z$ (LCD, by
Lemma~\ref{cyclic_LCD}). Each family is classically optimal, so by
Theorem~\ref{cor:cyclic-optimal-witnesses}, each yields an optimal pure maximally entangled
CSS-like EA-qLRC of the same locality, listed in Rows~1--4 of Table~\ref{tab:parameters_optimal_qLRCs}.

\subsection{Construction from Cyclic Codes of Length \texorpdfstring{$n\mid(q+1)$}{n | (q+1)}}

By Lemma~\ref{lem:n-divides-qplus1-LCD}, every cyclic code of length $n\mid(q+1)$ is
automatically LCD, so the LCD condition need not be separately verified for any construction
in this regime. Guenda et al.~\cite{guenda2018constructions} establish the existence of MDS
maximal-entangled EAQEC codes of length $q+1$:

\begin{theorem}[\cite{guenda2018constructions}, Theorem~5.1]
If $q$ is even, then there exists an MDS maximal-entangled EAQECC with parameters 
$[[q+1, k, q-k+2; q+1-k]]_q$, for all integers $k$ such that $1 \le k \le q+1$. 
If $q$ is odd, then there exists an MDS maximal-entangled EAQECC with parameters 
$[[q+1, k, q-k+2; q+1-k]]_q$, for all odd integers $k$ such that $1 \le k \le q+1$.
\end{theorem}

\begin{remark}
This theorem establishes the entanglement and MDS parameters of these codes but does not by
itself establish locality; the local-recoverability structure specific to this length regime
is supplied by the Chen et al. constructions below.
\end{remark}

Chen \textit{et al.}~\cite[Section~IV(A)]{chen2017constructions} construct cyclic LRC-LCD
codes of length $q+1$, classically optimal, for a range of locality parameters. By
Theorem~\ref{cor:cyclic-optimal-witnesses}, each yields an optimal pure maximally entangled
CSS-like EA-qLRC of the same locality, listed in Rows~5--8 of Table~\ref{tab:parameters_optimal_qLRCs}.

\begin{remark}
The maximally entangled constructions developed in this section ($c = n-k$, built
from LCD cyclic codes of length dividing $q \pm 1$ via Lemma~\ref{cyclic_LCD} and
Theorem~\ref{MEAQECC}) occupy a regime structurally disjoint from the explicit
families of~\cite{li2026eaqlrc}, which are instead built from $\ell$-intersection
pairs of MDS codes and from block parity-check matrices assembled from generalized
Reed--Solomon codes. In particular, our constructions inherit locality directly from
the cyclic defining-set structure of Lemma~\ref{cyclicLRC}, whereas the constructions
of~\cite{li2026eaqlrc} inherit locality from the block-diagonal structure of stacked
GRS parity-check matrices. The two approaches are complementary rather than
overlapping, and it would be of interest in future work to determine whether either
technique can be adapted to reproduce the parameter ranges achieved by the other.
\end{remark}

\begin{table}[!ht]
\centering
\caption{Optimal pure CSS-like EA-qLRCs {\qlb}$n, k, d;n-k${\qrb}$_q$ with Locality $r$}
\label{tab:parameters_optimal_qLRCs}
\renewcommand{\arraystretch}{1.35}
\begin{tabular}{|c|c|c|c|}
 \hline
\textbf{\#}
  & \textbf{{\qlb}$n,k,d;n-k${\qrb}$_q$}
  & \textbf{Condition}
  & \textbf{Reference} 
\\[3mm] \hline
 1& {\qlb}$n=2^m-1,k=\frac{rn}{r+1},2;n-k${\qrb}$_q$&$q = 2^m,\;(r+1)\mid n$ & Theorem~3~\cite{rajput2020cyclic}
 \\[3mm] \hline
  2&{\qlb}$n=2^m-1,k=\frac{nr}{r+1}-2m,\ge 6;n-k${\qrb}$_q$ &$q = 2^m,\;(r+1)\mid n,\;d \geq 6$& Theorem~4~\cite{rajput2020cyclic} 
 \\[3mm] \hline
  3&{\qlb}$n,k,n-k-\frac{k}{r}+2;n-k${\qrb}$_q$ &$n\mid q-1,r\mid k,\;(r+1)\mid n,\;
      \bigl(\tfrac{n}{r+1}-\tfrac{k}{r}\bigr)\in 2\mathbb{Z}$& Theorem~5~\cite{rajput2020cyclic}
 \\[3mm] \hline
  4&{\qlb}$n,k,\ge n-k-\frac{k}{r}+1;n-k${\qrb}$_q$ &   $n\mid q-1,(r+1)\mid n,\;
      2r\mid\bigl(\tfrac{nr}{r+1}-k-2a\bigr)$ &  Theorem~6~\cite{rajput2020cyclic}
 \\[3mm] \hline

5& {\qlb}$n=\nu(r+1),k=\mu r,n-k-\frac{k}{r}+2;n-k${\qrb}$_q$
  &
  \begin{tabular}[c]{@{}c@{}}
    $n\mid q+1$,\;$q$ even,\; $\mu$ even \\[2pt]
    $n\mid q+1$,\;$q$ odd,\; $\mu$ even,\; $\nu$ odd,\; $n$ even \\[2pt]
    $n\mid q+1$,\;$q$ odd,\; $\mu,\nu$ even,\; $n$ even \\[2pt]
  \end{tabular}
  &

  \begin{tabular}[c]{@{}c@{}}
    Theorem~10~\cite{chen2017constructions} \\[2pt]
    Theorem~26~\cite{chen2017constructions} \\[2pt]
    Theorem~29~\cite{chen2017constructions} \\[2pt]
   
  \end{tabular} 
\\[3mm] \hline

6 & {\qlb}$n=\nu(r+1),k=\mu r,n-k-\frac{k}{r}+2;n-k${\qrb}$_q$
  &
  \begin{tabular}[c]{@{}c@{}}
    $n\mid q+1$,\;$q$ even,\; $\mu$ odd \\[2pt]
    $n\mid q+1$,\;$q$ odd,\; $\mu,\nu$ odd,\; $n$ even \\[2pt]

  \end{tabular}
  &
  \begin{tabular}[c]{@{}c@{}}
    Theorem~13~\cite{chen2017constructions} \\[2pt]
    Theorem~23~\cite{chen2017constructions} \\[2pt]
  \end{tabular} 
\\[3mm] \hline

7 & {\qlb}$n=\nu(r+1),k=\mu r,n-k-\frac{k}{r}+2;n-k${\qrb}$_q$
  &
  \begin{tabular}[c]{@{}c@{}}
    $n\mid q+1$,\;$q$ even,\; $\mu$ even \\[2pt]
  
  \end{tabular}

  &
  \begin{tabular}[c]{@{}c@{}}
    Theorem~16~\cite{chen2017constructions} \\[2pt]

  \end{tabular} 
\\[3mm] \hline

8 & {\qlb}$n=\nu(r+1),k=\mu r,n-k-\frac{k}{r}+2;n-k${\qrb}$_q$
  &
  \begin{tabular}[c]{@{}c@{}}
    $n\mid q+1$,\;$q$ even,\; $\mu$ odd \\[2pt]
   
  \end{tabular}
  
  &
  \begin{tabular}[c]{@{}c@{}}
    Theorem~19~\cite{chen2017constructions} \\[2pt]

  \end{tabular} 
\\[3mm] \hline
\end{tabular}

~\\ \footnotesize\emph{Note.} Rows~5--8 assume $\mu\ge2$. For $\mu=1$ one has $k=r$, the
 locality constraint is vacuous, and \eqref{EA-qLRC_SB} degenerates to the
 entanglement-assisted quantum Singleton bound (cf.\ Remark~\ref{rem:boundary}).

\end{table}

\section{Achievability Bounds for CSS-like EA-qLRCs}\label{EA-qLRC GV}

All results above are existence statements at a fixed parameter point $(n,k,d,r)$. We now
complement them with a guarantee across the entire asymptotic rate--distance--locality
region, via classical Gilbert--Varshamov-like results for LRCs.

\subsection{A Basic Achievability Bound via Parity-Check Augmentation}\label{sec:basic-GV}

In this subsection, we derive the Gilbert-Varshamov-like bound for a CSS-like EA-qLRC using
the classical LRC Gilbert-Varshamov-like bound of Lemma~\ref{lem:classical-LRC-GV}, due to
Cadambe \emph{et al.}~\cite{cadambe2015bounds}, and the monomial-equivalence-to-LCD result,
due to Carlet \emph{et al.}~\cite{carlet2018lcd}. What we do in this regard is that the
\emph{reduction} carrying these two classical existence results across the
entanglement-assisted framework: (i)~Lemma~\ref{lem:monomial-preserves-locality} shows the
monomial transformation making a code LCD does not disturb its locality, so the two
classical results can be applied to the \emph{same} code simultaneously
(Theorem~\ref{thm:LCD-transfer-basic}); and (ii)~Theorem~\ref{EA-qLRC LCD}, already established,
converts any classical LCD LRC meeting the resulting rate bound directly into a maximally
entangled CSS-like EA-qLRC of the same rate, distance, and locality, at the fixed entanglement cost
$c=n-k$. Dividing the resulting finite-length relations $\kappa=k$, $\delta=d$, $c=n-k$ by
$n$ is what produces the asymptotic bound~\eqref{EAqLRC-GV} below.

\begin{lemma}[\cite{cadambe2015bounds}]\label{lem:classical-LRC-GV}
For every locality $r\ge1$ and every $0\le d' \le1-q^{-1}$, there exists an infinite family
of classical $(n,k,d,r)_q$-LRCs with $d/n\to d'$ and rate
\begin{equation}\label{LRC-GV}
    R \;\ge\; \frac{r}{r+1}-H_q(d'),
\end{equation}
obtained by augmenting the parity-check matrix of a code meeting the classical
Gilbert--Varshamov-like bound $R\ge1-H_q(d')$ with $\lceil n/(r+1)\rceil$ additional rows of
Hamming weight $r+1$ and pairwise disjoint support.
\end{lemma}


Combining Lemma~\ref{lem:classical-LRC-GV} with Theorem~\ref{EA-qLRC LCD} yields an
asymptotic achievability bound, provided the classical family can be taken to be LCD;
the following two results supply that condition for $q>3$ and record the resulting
bound.

\begin{lemma}\label{lem:monomial-preserves-locality}
Let $\mathcal{C}$ be an $[n,k,d]_q$-LRC with locality $r$, and let $\mathcal{C}'=\varphi(\mathcal{C})$
for a monomial transformation $\varphi(x)_i=\lambda_i x_{\sigma(i)}$. Then $\mathcal{C}'$ is
an $[n,k,d]_q$-LRC with locality $r$.
\end{lemma}
\begin{proof}
Since $\varphi$ is invertible and acts coordinatewise by permutation and nonzero scaling, it
preserves dimension and Hamming weight, so $\mathcal{C}'$ has parameters $[n,k,d]_q$. The
companion dual map $\psi(y)_i=\lambda_i^{-1}y_{\sigma(i)}$ satisfies
$\langle\varphi(x),\psi(y)\rangle=\langle x,y\rangle$, so $(\mathcal{C}')^\perp=\psi(\mathcal{C}^\perp)$.
Fix $j\in\{1,\dots,n\}$ and set $i=\sigma(j)$. Since $\mathcal{C}$ has locality $r$, there is
$h_i\in\mathcal{C}^\perp$ with $i\in\operatorname{supp}(h_i)$, $|\operatorname{supp}(h_i)|\le r+1$.
For $h'_j:=\psi(h_i)$, $(h'_j)_j=\lambda_j^{-1}(h_i)_{\sigma(j)}=\lambda_j^{-1}(h_i)_i\ne0$, and
$\operatorname{supp}(h'_j)=\sigma^{-1}(\operatorname{supp}(h_i))$, so $|\operatorname{supp}(h'_j)|\le r+1$.
Hence $\mathcal{C}'$ has locality $r$, with recovery sets $\Gamma'_j=\sigma^{-1}(\Gamma_{\sigma(j)})$.
\end{proof}

\begin{theorem}\label{thm:LCD-transfer-basic}
For $q>3$, every rate--distance--locality point achievable by Lemma~\ref{lem:classical-LRC-GV}
is also achievable by an LCD LRC of the same parameters.
\end{theorem}

\begin{proof}
Let $\mathcal{C}_{\mathrm{aug}}$ be a classical $(n,k,d,r)_q$-LRC from the family of
Lemma~\ref{lem:classical-LRC-GV}. Since $q>3$, Carlet \emph{et al.}~\cite{carlet2018lcd} give a
monomial transformation $\varphi$ under which $\varphi(\mathcal{C}_{\mathrm{aug}})$ is
LCD. By Lemma~\ref{lem:monomial-preserves-locality}, $\varphi(\mathcal{C}_{\mathrm{aug}})$ has the
same parameters $[n,k,d]_q$ and the same locality $r$ as $\mathcal{C}_{\mathrm{aug}}$,
and is therefore an LCD LRC meeting the Gilbert--Varshamov-like rate bound
\eqref{LRC-GV}.
\end{proof}

\begin{theorem}\label{EA-LRC-GV}
Let $q>3$ and $0\le\Delta\le1-q^{-1}$. Then there exists an infinite family of pure
maximally entangled CSS-like EA-qLRCs with locality $r$, obtained via
Proposition~\ref{CSS-EA-qLRC} from classical LRCs of locality $r$ drawn from the
family of Lemma~\ref{lem:classical-LRC-GV}, whose relative distance tends to $\Delta$ and whose
rate, entanglement rate and net rate satisfy
\begin{equation}\label{EAqLRC-GV}
R \;\ge\; \frac{r}{r+1}-H_q(\Delta),
\qquad \gamma=1-R, \qquad R_{net}=2R-1 .
\end{equation}
The same conclusion holds for any field size $q$ for which a member of the family of Lemma~\ref{lem:classical-LRC-GV} can be taken to be LCD; for $q=2,3$ this is not known in general (Remark~\ref{rem:LCD-gap}).
\end{theorem}

\begin{proof}
By Theorem~\ref{thm:LCD-transfer-basic}, for $q>3$ each member $\mathcal{C}=[n,k,d]_q$
of the family of Lemma~\ref{lem:classical-LRC-GV} may be replaced by a monomially equivalent LCD
code of the same parameters and locality. Applying Theorem~\ref{EA-qLRC LCD} with
$\mathcal{C}_1=\mathcal{C}_2=\mathcal{C}$ gives $\kappa=k$, $\delta=d$ and $c=n-k$.
Dividing by $n$: the rate bound \eqref{LRC-GV} gives
$R=\kappa/n=k/n\ge r/(r+1)-H_q(d/n)\to r/(r+1)-H_q(\Delta)$ as $d/n\to\Delta$, while
$\gamma=c/n=1-k/n=1-R$ and $R_{net}=(\kappa-c)/n=2(k/n)-1=2R-1$ hold identically for
every $n$. Letting $n\to\infty$ yields \ref{EAqLRC-GV}. The final claim follows by
the same argument, with Theorem~\ref{thm:LCD-transfer-basic} replaced by the assumed
LCD-ness.
\end{proof}




\begin{remark}\label{rem:LCD-gap}
For $q>3$, Theorem~\ref{thm:LCD-transfer-basic} closes the LCD-GV-locality gap completely: the
asymptotic bound~\eqref{EAqLRC-GV} holds unconditionally, via the monomial-equivalence
result of Carlet \emph{et al.}~\cite{carlet2018lcd} applied directly to the
locality-augmented code of Lemma~\ref{lem:classical-LRC-GV}, together with
Lemma~\ref{lem:monomial-preserves-locality} showing this equivalence preserves locality
exactly. For $q=2,3$, Carlet et al.'s equivalence is known to fail in general, and closing
the gap there remains open. In the interim, the explicit cyclic families of
Theorem~\ref{cor:cyclic-optimal-witnesses} continue to provide unconditional, deterministic
LCD-LRC witnesses at any field size, including $q=2,3$, for the parameter ranges they cover.
\end{remark}

\begin{remark}[Asymptotic discussion]\label{rem:GV-asymptotic-discussion}
\emph{Tightness.} Bound~\eqref{EAqLRC-GV} is an achievability (lower) bound and is not
claimed to be tight in general; as with the classical Gilbert--Varshamov-like bound it
inherits from, closing the gap to a matching converse for $\Delta>0$ is a long-standing open
problem even classically. It is, however, tight at the boundary $\Delta=0$: substituting the
maximally entangled regime $\gamma=1-R$ into the asymptotic Singleton-like
bound~\eqref{ASB} of Section~\ref{Bounds EA-qLRC} and solving for the largest admissible $R$
gives the converse $R\le\frac{r}{r+1}(1-\Delta)$, which equals $\frac{r}{r+1}$ at
$\Delta=0$ --- exactly the value~\eqref{EAqLRC-GV} gives at $\Delta=0$ (since $H_q(0)=0$).
The two bounds therefore meet exactly at $\Delta=0$, though for $\Delta>0$ the achievability
bound falls strictly below the converse, since $H_q(\Delta)$ has unbounded slope as
$\Delta\to0^+$ (from the $-\Delta\log_q\Delta$ term) while the converse
$\frac{r}{r+1}(1-\Delta)$ has the finite slope $-\frac{r}{r+1}$; this produces the same
qualitative gap familiar from the classical Gilbert--Varshamov-versus-Singleton comparison, not a new
phenomenon introduced by entanglement assistance.

\emph{Comparison with the classical Gilbert--Varshamov bound.} Bound~\eqref{EAqLRC-GV} has exactly the same
rate--distance trade-off $\frac{r}{r+1}-H_q(\Delta)$ as the classical LRC
Gilbert--Varshamov-like bound~\eqref{LRC-GV} it is built from; entanglement assistance is
not paid for by a worse rate--distance trade-off here, but by the fixed asymptotic
entanglement cost $\gamma=1-R$. This is the expected price of removing the dual-containing
constraint entirely: the construction is free to start from \emph{any} classical LRC
meeting~\eqref{LRC-GV}, at the cost of consuming $\gamma=1-R$ ebits per transmitted qudit
asymptotically.

\emph{Comparison with CSS-qLRC asymptotics.} A CSS-qLRC construction without
entanglement assistance requires $\mathcal{C}_1^{\perp}\subseteq\mathcal{C}_2$. In the
single-code case $\mathcal{C}_1=\mathcal{C}_2=\mathcal{C}$ used throughout this
subsection, this is the dual-containing condition
$\mathcal{C}^{\perp}\subseteq\mathcal{C}$, equivalently
$\dim(\mathcal{C}\cap\mathcal{C}^{\perp})=n-k$, at the opposite extreme from the LCD
condition $\dim(\mathcal{C}\cap\mathcal{C}^{\perp})=0$ exploited here. Dual-containment
forces $k\ge n/2$, so such a construction can only access asymptotic rates
$R\ge 1/2$: the entire low-rate half of the classical Gilbert--Varshamov-like region
\eqref{LRC-GV} is unreachable, and with it the high relative-distance
regime, since \eqref{LRC-GV} gives $R<1/2$ precisely when
$H_q(\Delta)>r/(r+1)-1/2$. Matching the classical GV rate with an explicit family of
dual-containing LRCs is therefore not merely a harder existence question than matching
it with an LCD family---it is impossible over an entire subregion. The LCD condition,
by contrast, imposes no restriction on $k$ at all, and by
Theorem~\ref{thm:LCD-transfer-basic} is attainable at every point of
\eqref{LRC-GV} for $q>3$. The entanglement-assisted route of this
subsection thus accesses the full classical Gilbert--Varshamov rate--distance--locality region at the
fixed entanglement cost $\gamma=1-R$, rather than forfeiting the low-rate,
high-distance portion of it to the CSS-compatible dual-containment condition.
(For a general pair $\mathcal{C}_1\ne\mathcal{C}_2$ the condition
$\mathcal{C}_1^{\perp}\subseteq\mathcal{C}_2$ requires only $k_1+k_2\ge n$ rather than
$k_i\ge n/2$ individually; the restriction above is specific to the single-code
construction.)
\end{remark}

\begin{example}\label{ex:basic-GV-worked}
Let $q=4$, $n=q+1=5$. By Lemma~\ref{lem:n-divides-qplus1-LCD}, every cyclic code of this
length over $\mathbb{F}_4$ is LCD. The $q$-cyclotomic cosets modulo $5$ are $\{0\}$,
$\{1,4\}$, and $\{2,3\}$. Choosing the defining set $T=\{1,4\}$ gives a cyclic code
$\mathcal{C}=[5,3,3]_4$ with dual $\mathcal{C}^\perp=[5,2,4]_4$; since $\mathcal{C}^\perp$
has minimum weight $4$, cyclic shifts of a minimum-weight dual codeword cover every
coordinate, giving $\mathcal{C}$ locality $r=3$. Direct computation confirms
$\dim(\mathcal{C}\cap\mathcal{C}^\perp)=0$. Applying Theorem~\ref{EA-qLRC LCD} yields an
explicit maximally entangled CSS-like EA-qLRC {\qlb}${5},{3},{3};{2}${\qrb}$_4$ with locality $r=3$ over
$\mathbb{F}_4$, for which $R=3/5$, $\gamma=2/5=1-R$, and $R_{net}=1/5=2R-1$, confirming the
relations of Theorem~\ref{EA-LRC-GV} exactly.
\end{example}

\subsection{A Sharper Achievability Bound via Concatenated Codes}\label{sec:CM-refined}

The simple Gilbert--Varshamov-type bound~\eqref{EAqLRC-GV} inherits its rate--distance
trade-off from the parity-check augmentation of Lemma~\ref{lem:classical-LRC-GV}, which
spends redundancy purely on enforcing locality without improving distance. Cadambe and
Mazumdar's concatenated-code construction~\cite{cadambe2015bounds} instead builds locality
directly into the code's block structure, so that the same redundancy simultaneously
provides local recoverability and distance protection. As in Section~\ref{sec:basic-GV},
closing the LCD-existence gap for this sharper family uses the same two ingredients ---
Carlet et al.'s monomial-equivalence result~\cite{carlet2018lcd} and
Lemma~\ref{lem:monomial-preserves-locality} --- applied to the concatenated-code family in
place of the parity-check-augmented one. By the identical reduction used to establish
Theorem~\ref{EA-LRC-GV}, this yields a strictly sharper achievability bound for CSS-like
EA-qLRCs.

\begin{lemma}[Concatenated-Code Achievability~\cite{cadambe2015bounds}, Theorem~2]\label{lem:CM-refined-achievability}
There exists an infinite family of $[n,k,d]_q$ \emph{linear} concatenated codes with
locality $r$, constructed from an outer random $q^r$-ary linear code of length $n/(r+1)$
and a $q$-ary single-parity-check inner code of length $r+1$, such that
\begin{equation}\label{eq:CM-general}
    \frac{k}{n} \;=\; 1-\min_{0\le x\le1}\left[\log_q\bigl(1+x(q-1)\bigr)+\frac{1}{r+1}\log_q\!\left(1+(q-1)\Bigl(\frac{1-x}{1+x(q-1)}\Bigr)^{r+1}\right)-\frac{d}{n}\log_q x\right].
\end{equation}
In particular, for relative distance $\frac{d}{n} \in \left(0, 1-\frac{1}{q}\right)$,
substituting $x=\dfrac{d}{(q-1)(n-d)}$ into~\eqref{eq:CM-general} yields the explicit rate
bound:
\begin{equation}\label{LRC-CM-refined}
    \frac{k}{n} \;=\; 1-H_q\!\left(\frac{d}{n}\right)-\frac{1}{r+1}\log_q\!\left[1+(q-1)\left(1-\frac{d}{n}\cdot\frac{q}{q-1}\right)^{r+1}\right].
\end{equation}
\end{lemma}

\begin{theorem}\label{thm:CM-LCD-closure}
For $q>3$, every rate--distance--locality point achievable by the concatenated-code family
of Lemma~\ref{lem:CM-refined-achievability} is also achievable by an LCD LRC of the same
parameters $[n,k,d]_q$ and locality $r$.
\end{theorem}
\begin{proof}
Let $\mathcal{C}_{\mathrm{cat}}$ be a classical $[n,k,d]_q$ linear concatenated code of
locality $r$ from Lemma~\ref{lem:CM-refined-achievability}. Since $q>3$, Carlet
\emph{et al.}~\cite{carlet2018lcd} guarantee the existence of a monomial transformation
$\varphi$ under which $\varphi(\mathcal{C}_{\mathrm{cat}})$ is LCD. By
Lemma~\ref{lem:monomial-preserves-locality}, monomial transformations preserve distance,
dimension, and coordinate repair sets. Thus, $\varphi(\mathcal{C}_{\mathrm{cat}})$ maintains
the parameters $[n,k,d]_q$ and locality $r$, yielding an LCD LRC meeting
rate~\eqref{LRC-CM-refined}.
\end{proof}

\begin{theorem}\label{thm:EA-LRC-CM-refined}
Let $q>3$ and $\Delta \in \left(0, 1-\frac{1}{q}\right)$. A CSS-like EA-qLRC $\mathcal{Q}_E$
with locality $r$ obtained via Proposition~\ref{CSS-EA-qLRC} using a classical LRC
$\mathcal{C}$ drawn from the family of Lemma~\ref{lem:CM-refined-achievability} satisfies
\begin{equation}\label{EAqLRC-GV-refined}
    R \;\ge\; 1-H_q(\Delta)-\frac{1}{r+1}\log_q\!\left[1+(q-1)\left(1-\Delta\cdot\frac{q}{q-1}\right)^{r+1}\right],
    \qquad \gamma=1-R, \qquad R_{\mathrm{net}}=2R-1.
\end{equation}
\end{theorem}
\begin{proof}
By Theorem~\ref{thm:CM-LCD-closure}, for $q>3$, any member $\mathcal{C}=[n,k,d]_q$ of the
family in Lemma~\ref{lem:CM-refined-achievability} can be transformed into an LCD code
without loss of rate or locality. Applying Theorem~\ref{EA-qLRC LCD} with
$\mathcal{C}_1=\mathcal{C}_2=\mathcal{C}$ yields $\kappa=k$, $\delta=d$, and required ebits
$c=n-k$. Normalizing by $n$ gives $R=\kappa/n=k/n$, so the classical rate
bound~\eqref{LRC-CM-refined} transfers directly as $n\to\infty$ with $d/n\to\Delta$. The
consumption parameter $\gamma=c/n=1-R$ and net rate $R_{\mathrm{net}}=(\kappa-c)/n=2R-1$
follow immediately.
\end{proof}

\begin{remark}\label{rem:sharper}
Theorem~\ref{thm:EA-LRC-CM-refined} sharpens Theorem~\ref{EA-LRC-GV}: both bounds coincide
at the endpoints $\Delta=0$ (giving $R=\tfrac{r}{r+1}$) and $\Delta=1-q^{-1}$ (giving
$R=0$), but~\eqref{EAqLRC-GV-refined} is strictly tighter in the interior of this range,
reflecting the improved rate--distance trade-off of the underlying classical
concatenated-code construction over the simple parity-check augmentation of
Lemma~\ref{lem:classical-LRC-GV}.
\end{remark}

The bound~\eqref{EAqLRC-GV} is achieved only via the random-coding argument of
Lemma~\ref{lem:classical-LRC-GV}, not by explicit algebraic constructions such as the
cyclic family of Example~\ref{ex:basic-GV-worked}; the same is true of the sharper
bound~\eqref{EAqLRC-GV-refined} of this subsection, via the concatenated-code argument of
Lemma~\ref{lem:CM-refined-achievability}.

\subsection{Numerical Analysis of the Sharper Achievability Bound}\label{sec:numerical-analysis}

To quantify the impact of the locality parameter $r$ on the achievable rate capacity of the
constructed CSS-like EA-qLRCs, we decompose the sharpened lower bound from
Theorem~\ref{thm:EA-LRC-CM-refined} into the standard Gilbert--Varshamov (GV) rate and an
explicit locality penalty term:
\begin{equation}\label{eq:rate-decomposition}
    R_{\mathrm{LRC}}(\Delta, r) = \underbrace{1 - H_q(\Delta)}_{\text{Standard GV Rate } R_{\mathrm{GV}}} - \underbrace{\frac{1}{r+1} \log_q \!\left[ 1 + (q-1)\left(1 - \Delta \cdot \frac{q}{q-1}\right)^{r+1} \right]}_{\text{Locality Penalty } \delta_r(\Delta)}.
\end{equation}

The penalty term $\delta_r(\Delta)$ exhibits three main asymptotic characteristics:
\begin{enumerate}
    \item \textbf{Low-Noise Limit ($\Delta \to 0$):} Taking the limit yields
    $\lim_{\Delta \to 0} \delta_r(\Delta) = \frac{1}{r+1}$. The maximum rate simplifies to
    $1 - \frac{1}{r+1} = \frac{r}{r+1}$, which matches the Singleton-type upper bound on LRC
    rate bounds.
    \item \textbf{High-Noise Limit ($\Delta \to \frac{q-1}{q}$):} As $\Delta$ approaches the
    GV channel threshold, the term $\left(1 - \Delta \frac{q}{q-1}\right) \to 0$, driving
    $\delta_r(\Delta) \to 0$. Near the maximum supportable relative distance, local parity
    constraints impose virtually no additional asymptotic rate penalty over standard random
    linear codes.
    \item \textbf{Unconstrained Locality Limit ($r \to \infty$):} For fixed $\Delta$ with
    $0<\Delta<\frac{q-1}{q}$, write $x:=1-\Delta\frac{q}{q-1}\in(0,1)$. Since
    $\log_q(1+u)=u/\ln q+O(u^2)$ as $u\to0$, the penalty satisfies
    \[
        \delta_r(\Delta) = \frac{(q-1)}{(r+1)\ln q}\,x^{r+1}\bigl(1+o(1)\bigr), \qquad r\to\infty,
    \]
    so $\delta_r(\Delta)$ vanishes \emph{exponentially} in $r$ (since $0<x<1$), not merely
    at the polynomial rate $O(1/r)$ the prefactor alone might suggest; the bound therefore recovers the standard Gilbert--Varshamov rate $R_{\mathrm{GV}}(\Delta)$ very rapidly as locality is
    relaxed.
\end{enumerate}

\begin{figure}[!ht]
\centering
\includegraphics[width=\linewidth]{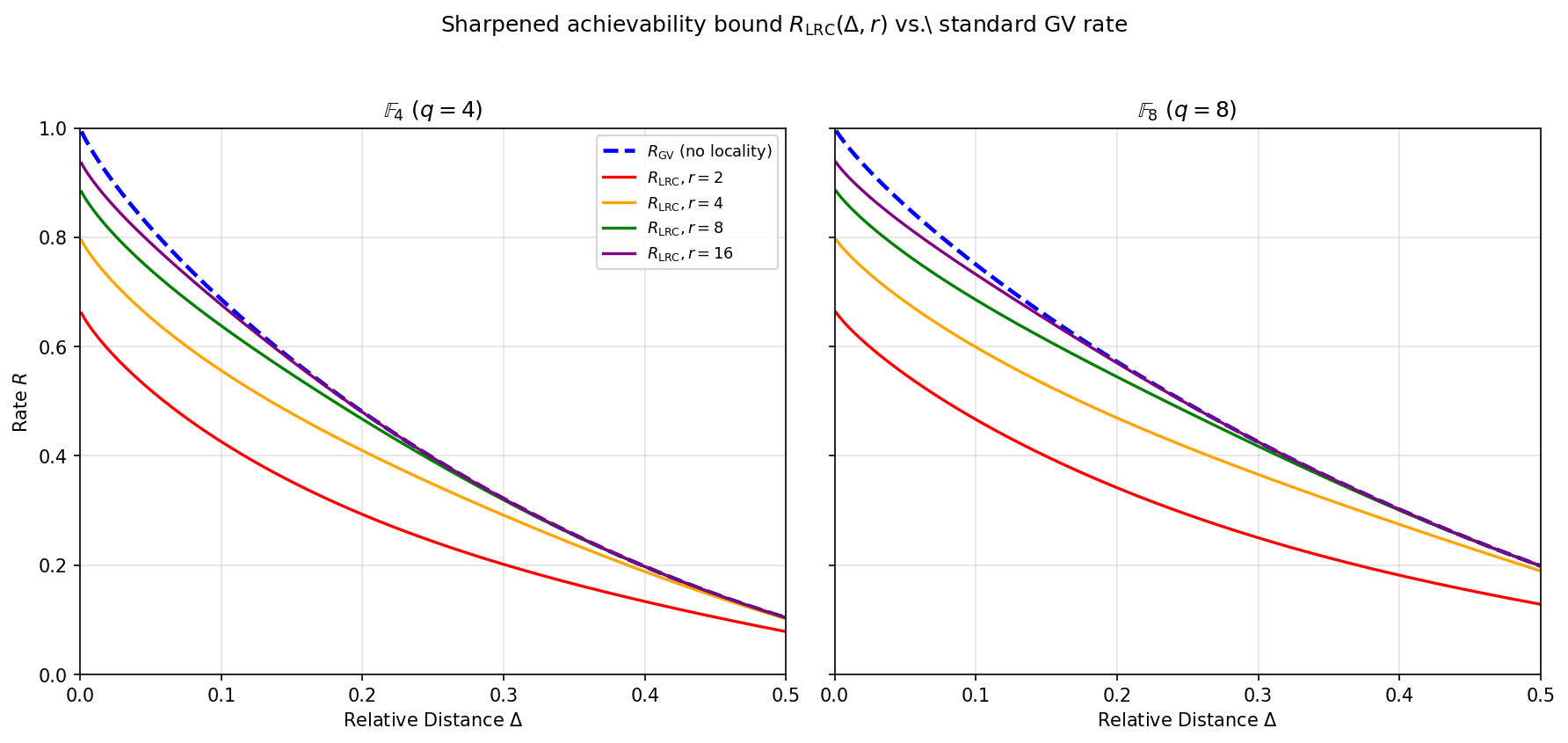}
\caption{The sharpened achievability bound $R_{\mathrm{LRC}}(\Delta,r)$ of
Theorem~\ref{thm:EA-LRC-CM-refined}, plotted against the standard Gilbert--Varshamov rate
$R_{\mathrm{GV}}(\Delta)=1-H_q(\Delta)$, for $r=2,4,8,16$ over $\mathbb{F}_4$ and
$\mathbb{F}_8$. As $r$ increases, $R_{\mathrm{LRC}}$ converges rapidly toward
$R_{\mathrm{GV}}$, consistent with the exponential decay of the locality penalty established
above.}
\label{fig:lrc-gv-comparison}
\end{figure}

Figure~\ref{fig:lrc-gv-comparison} illustrates this behavior directly: the gap between
$R_{\mathrm{LRC}}$ and $R_{\mathrm{GV}}$ narrows visibly as $r$ grows from $2$ to $16$, and
the convergence is markedly faster than a naive $1/r$ rate would predict, consistent with
the exponential decay of $\delta_r(\Delta)$ derived above.

\subsection{Unified Comparison Under Maximally Entangled Constraints}\label{sec:unified-comparison}

The converse bounds of Section~\ref{Bounds EA-qLRC} were derived for arbitrary entanglement
rate $\gamma$; we now specialize them to the maximally entangled regime $\gamma=1-R$ of
Theorem~\ref{EA-qLRC LCD}, placing every bound of this paper on common ground for direct
comparison against the achievability bounds of Theorem~\ref{EA-LRC-GV} and
Theorem~\ref{thm:EA-LRC-CM-refined}.

Each of the Singleton-, Griesmer-, and Plotkin-like asymptotic bounds \eqref{ASB},
\eqref{AGB}, \eqref{APB} shares the common form
\begin{equation}\label{eq:star}
    R \;\le\; \frac{r}{r+2}(1+\gamma) - \frac{2r}{r+2}\,C\,\Delta,
\end{equation}
where $C=1$ for the Singleton-like bound, $C=\frac{q^m-1}{q^{m-1}(q-1)}$ for the
Griesmer-like bound, and $C=\frac{q^r-1}{q^{r-1}(q-1)}$ for the Plotkin-like bound.

\begin{proposition}\label{prop:maxent-collapse}
Under the maximally entangled constraint $\gamma=1-R$, each bound of the form~\eqref{eq:star}
reduces to
\begin{equation}\label{eq:maxent-linear}
    R \;\le\; \frac{r}{r+1}\bigl(1-C\Delta\bigr).
\end{equation}
In particular, the Singleton-, Griesmer-, and Plotkin-like bounds~\eqref{ASB}, \eqref{AGB},
\eqref{APB} become, respectively,
\[
    R\le\frac{r}{r+1}(1-\Delta), \qquad
    R\le\frac{r}{r+1}\Bigl(1-\frac{q^m-1}{q^{m-1}(q-1)}\Delta\Bigr), \qquad
    R\le\frac{r}{r+1}\Bigl(1-\frac{q^r-1}{q^{r-1}(q-1)}\Delta\Bigr).
\]
\end{proposition}
\begin{proof}
Substituting $\gamma=1-R$ into~\eqref{eq:star} and multiplying by $(r+2)$ gives
$R(r+2)\le r(2-R)-2rC\Delta$. Collecting the $R$ terms,
$R(2r+2)\le 2r(1-C\Delta)$, and dividing by $2(r+1)$ yields~\eqref{eq:maxent-linear}.
\end{proof}

\begin{proposition}\label{prop:spb-collapse}
Under the maximally entangled constraint $\gamma=1-R$, the sphere-packing-like
bound~\eqref{asySPB} reduces to the closed form
\begin{equation}\label{eq:spb-maxent}
    R \;\le\; 1-H_q\!\left(\frac{\Delta}{2}\right).
\end{equation}
\end{proposition}
\begin{proof}
Recall~\eqref{asySPB}: $R\le 1+\gamma-2\mu H_q\bigl(\Delta/(2\mu)\bigr)$, where
$\mu=(1+R+\gamma)/2$. Substituting $\gamma=1-R$ gives $\mu=\bigl(1+R+(1-R)\bigr)/2=1$,
independent of $R$ and $\Delta$. Substituting $\mu=1$ and $\gamma=1-R$ into~\eqref{asySPB}
gives $R\le 2-R-2H_q(\Delta/2)$, and dividing by $2$ after adding $R$ to both sides
yields~\eqref{eq:spb-maxent}.
\end{proof}

\begin{remark}
Unlike the four converse bounds above, the Gilbert--Varshamov-like bounds~\eqref{EAqLRC-GV}
and~\eqref{EAqLRC-GV-refined} require no substitution to reach closed form in $\Delta$
alone: they are derived directly from Theorem~\ref{EA-qLRC LCD}'s maximally entangled
construction, for which $c=n-k$ exactly, so $\gamma=c/n=1-k/n=1-R$ holds identically by
construction rather than as an imposed constraint. Propositions~\ref{prop:maxent-collapse}
and~\ref{prop:spb-collapse} show that restricting the four converse bounds to this same
regime is precisely what is needed to compare them against~\eqref{EAqLRC-GV} and
\eqref{EAqLRC-GV-refined} on common ground, as in Figure~\ref{fig:comprehensive-comparison}.
\end{remark}

\begin{figure}[!ht]
\centering
\includegraphics[width=0.9\linewidth]{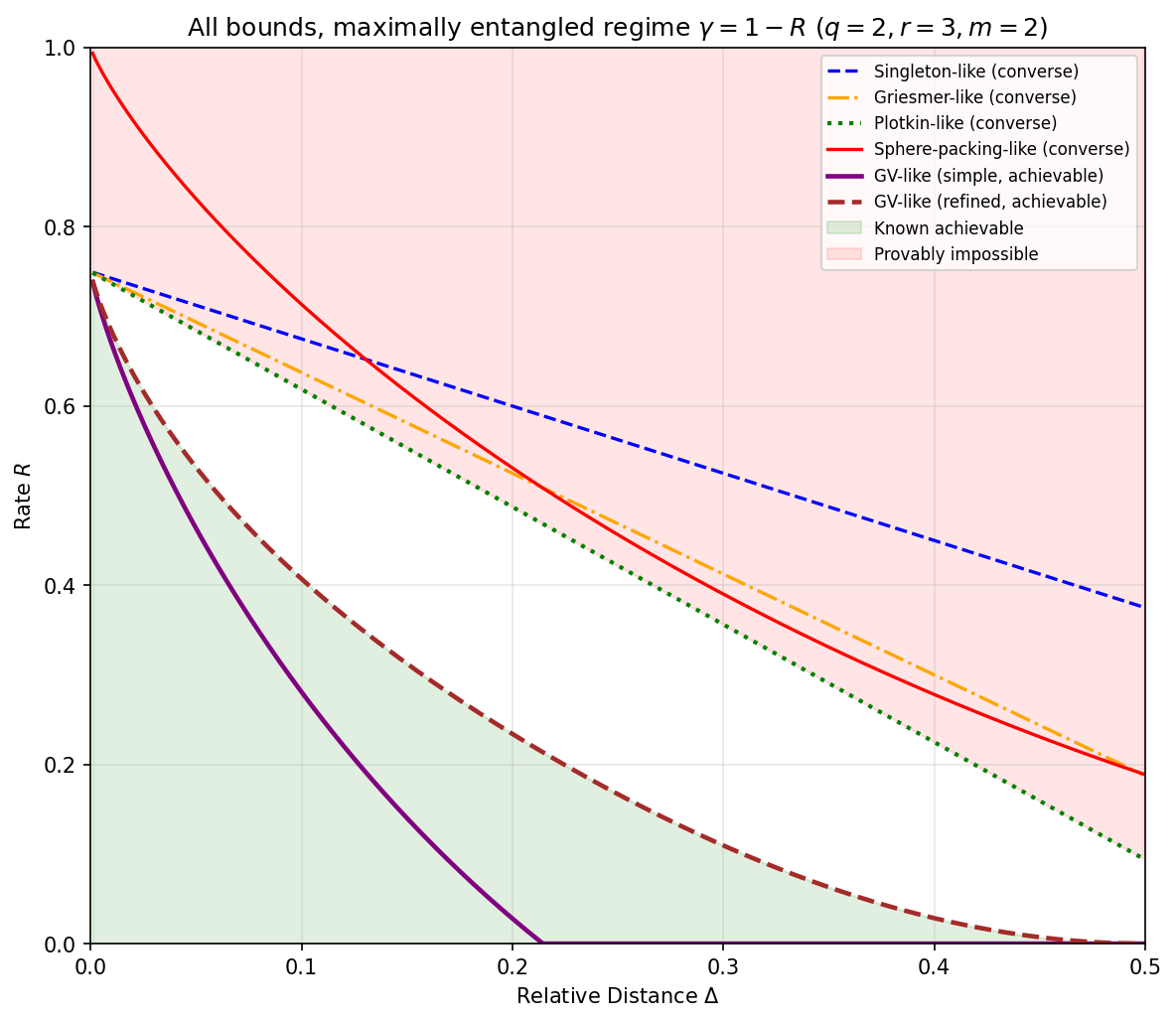}
\caption{The shaded red region is provably impossible, bounded below by the pointwise minimum
of the four converse bounds (at $r=3$ this is the Plotkin-like bound throughout, by
Remark~\ref{rem:sphere-loosest}); the shaded green region is provably achievable,
bounded above by the pointwise maximum of the two achievability bounds (the refined
Gilbert--Varshamov-like bound throughout, by Remark~\ref{rem:sharper}). The unshaded
gap between them, familiar from the classical Singleton-versus-Gilbert--Varshamov
comparison, remains open.}
\label{fig:comprehensive-comparison}
\end{figure}

Figure~\ref{fig:comprehensive-comparison} places every bound developed in this paper on a
single axis, making visible the gap between what is provably achievable and what is
provably impossible for entanglement-assisted quantum locally recoverable codes ---
mirroring the analogous, long-open gap between the classical Gilbert--Varshamov and
Singleton-type bounds that this paper's constructions inherit.

\section{Conclusion}\label{conclusion}

In this paper, we developed a systematic framework for the study of entanglement-assisted
quantum locally recoverable codes (CSS-like EA-qLRCs). We first established the fundamental theory
of CSS-like EA-qLRCs by deriving five converse bounds --- Singleton-like, Griesmer-like,
Plotkin-like, and sphere-packing-like in explicit closed form, together with a
Cadambe--Mazumdar-like bound that, as in the classical case, is not explicit --- extending
their classical locally recoverable code counterparts to the entanglement-assisted quantum
setting. We analyzed the relative tightness of the four explicit bounds, showing that the
Griesmer-like and Plotkin-like bounds are asymptotically strictly tighter than the
Singleton-like bound, while the sphere-packing-like bound is tightest in specific regimes of
relative distance and locality.

Building upon these theoretical results, we characterized the necessary and sufficient
conditions under which pure CSS-like EA-qLRCs attain the Singleton-like bound. This
characterization enabled explicit constructions of optimal pure CSS-like EA-qLRCs from
cyclic locally recoverable codes, and revealed a structural obstruction showing that the
Tamo--Barg-based construction, while yielding valid CSS-like EA-qLRCs, attains the Singleton-like bound only when $k\le r$, a regime in which the locality constraint is vacuous---every coordinate is recoverable from an information set
avoiding it---and \eqref{EA-qLRC_SB} degenerates to the entanglement-assisted
quantum Singleton bound. The construction therefore attains
\eqref{EA-qLRC_SB} in no regime where locality carries information. We further constructed a family of maximally
entangled CSS-like EA-qLRCs using cyclic locally recoverable codes and their complementary
dual codes.

Finally, we established two Gilbert--Varshamov-like achievability bounds for CSS-like EA-qLRCs --- a basic bound via classical parity-check augmentation and a sharper bound via concatenated codes --- and showed both hold unconditionally for field size $q>3$ via a monomial-equivalence argument that preserves locality. We closed the paper by unifying every bound established across this work, converse and achievability alike, under the shared maximally entangled regime, yielding closed-form expressions for each and a single
comparison delineating the boundary between provably achievable and provably impossible rate--distance--locality trade-offs for CSS-like EA-qLRCs.

EA-qLRCs are an emerging area of research in quantum error correction. The results
presented in this paper naturally lead to several promising directions for future research.

\begin{itemize}
\item The constructions developed in this work are primarily based on pure CSS-like
EA-qLRCs derived from cyclic and Tamo--Barg locally recoverable codes. It would be of
significant interest to investigate broader classes of classical locally recoverable codes
that give rise to new families of EA-qLRCs with improved parameters. In particular, the
systematic construction and characterization of impure CSS-like EA-qLRCs remain largely
unexplored and may lead to improved trade-offs among locality, minimum distance, and
entanglement consumption.

\item While this work provides explicit constructions attaining the Singleton-like bound,
an important open problem is the construction of EA-qLRCs that are optimal with respect to
the stronger Griesmer-like, Plotkin-like, and sphere-packing-like bounds established in this
paper. Determining whether infinite families of CSS-like EA-qLRCs can attain these bounds constitutes
a promising direction for future investigation.

\item Closing the Gilbert--Varshamov-LCD gap for $q=2,3$ --- where the monomial-equivalence argument used for $q>3$ is known to fail --- remains open.

\item Narrowing the gap between the achievability bounds established in this work and the converse bounds of Section~\ref{Bounds EA-qLRC} would provide a deeper understanding of the fundamental limits of entanglement-assisted locally recoverable quantum codes.
\end{itemize}

We hope that the framework developed in this paper will stimulate further research on entanglement-assisted locally recoverable quantum codes and strengthen the connections between classical locality theory and quantum error correction.

\section*{Acknowledgments}
\bibliographystyle{IEEEtran}
\bibliography{bibliography}
\end{document}